\documentclass[sigconf]{acmart}

\AtBeginDocument{%
  \providecommand\BibTeX{{%
    \normalfont B\kern-0.5em{\scshape i\kern-0.25em b}\kern-0.8em\TeX}}}

\copyrightyear{2024}
\acmYear{2024}
\setcopyright{rightsretained}
\acmConference[KDD '24]{Proceedings of the 30th ACM SIGKDD Conference on Knowledge Discovery and Data Mining}{August 25--29, 2024}{Barcelona, Spain}
\acmBooktitle{Proceedings of the 30th ACM SIGKDD Conference on Knowledge Discovery and Data Mining (KDD '24), August 25--29, 2024, Barcelona, Spain}\acmDOI{10.1145/3637528.3671990}
\acmISBN{979-8-4007-0490-1/24/08}

\makeatletter \gdef\@copyrightpermission{
\begin{minipage}{0.3\columnwidth} \href{https://creativecommons.org/licenses/by/4.0/}{\includegraphics[width=0.90\textwidth]{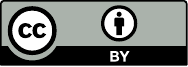}} \end{minipage}\hfill \begin{minipage}{0.7\columnwidth}
\href{https://creativecommons.org/licenses/by/4.0/}{This work is licensed under a Creative Commons Attribution International 4.0 License.}
\end{minipage}
\vspace{5pt} }
\makeatother

\usepackage{amsmath,amsfonts,bm}

\def\eqref#1{equation~\ref{#1}}

\def\1{\bm{1}}

\def\eps{{\epsilon}}

\def\vx{{\bm{x}}}

\DeclareMathAlphabet{\mathsfit}{\encodingdefault}{\sfdefault}{m}{sl}
\SetMathAlphabet{\mathsfit}{bold}{\encodingdefault}{\sfdefault}{bx}{n}

\newcommand{\R}{\mathbb{R}}

\DeclareMathOperator*{\argmin}{arg\,min}

\usepackage[utf8]{inputenc} %
\usepackage[T1]{fontenc}    %
\usepackage{hyperref}       %
\usepackage{url}            %
\usepackage{booktabs}       %
\usepackage{amsfonts}       %
\usepackage{nicefrac}       %
\usepackage{microtype}      %
\usepackage{xcolor}         %
\usepackage[normalem]{ulem}
\usepackage{natbib}

\usepackage{amsthm}
\usepackage{dirtytalk}
\usepackage{algorithm}
\usepackage[noend]{algpseudocode}
\usepackage{multirow}
\usepackage{multicol}

\usepackage{lipsum}  
\usepackage{wrapfig}
\usepackage{graphicx}  
\usepackage{subfig}

\usepackage{stmaryrd}

\usepackage{anyfontsize}

\algnewcommand\algorithmicinput{\textbf{Input:}}
\algnewcommand\algorithmicoutput{\textbf{Output:}}
\algnewcommand\Input{\item[\algorithmicinput]}%
\algnewcommand\Output{\item[\algorithmicoutput]}%

\newtheorem{theorem}{Theorem}[section]

\newtheorem{lemma}[theorem]{Lemma}
\newtheorem{definition}[theorem]{Definition}

\newtheorem*{example*}{Example}

\newcommand{\prob}{\mathbb{P}}
\renewcommand{\eps}{\varepsilon}

\usepackage{xspace}

\renewcommand{\paragraph}[1]{\noindent\textbf{#1}}

\newcommand{\highlight}[1]{\textcolor{blue}{\underline{#1}}}
\newcommand{\naive}[1]{\textcolor{red}{#1}}
\newcommand{\nonpriv}[1]{\textcolor{orange}{#1}}
\definecolor{ForestGreen}{RGB}{34,139,34}
\newcommand{\flaim}[1]{\textcolor{ForestGreen}{#1}}

\newcommand{\review}[1]{\textcolor{black}{#1}}

\begin{document}

\title{FLAIM: AIM-based Synthetic Data Generation in the Federated Setting}

\author{Samuel Maddock}
\authornote{Author correspondence to s.maddock@warwick.ac.uk}
\affiliation{%
  \institution{University of Warwick}
  \country{}
}

\author{Graham Cormode}
\affiliation{%
  \institution{Meta AI \& University of Warwick}
  \country{}
  }

\author{Carsten Maple}
\affiliation{%
  \institution{University of Warwick}
  \country{}
  }

\begin{abstract}
    Preserving individual privacy while enabling collaborative data sharing is crucial for organizations. Synthetic data generation is one solution, producing artificial data that mirrors the statistical properties of private data. While numerous techniques have been devised under
    differential privacy, they predominantly assume data is centralized. However, data is often distributed across multiple clients in a federated manner.
    In this work, we initiate the study of federated synthetic tabular data generation. Building upon a SOTA central method known as AIM, we present \textit{DistAIM} and \textit{FLAIM}. We first show that it is straightforward to distribute AIM, extending a recent approach based on secure multi-party computation which necessitates additional overhead, making it less suited to federated scenarios. We then demonstrate that naively federating AIM can lead to substantial degradation in utility under the presence of heterogeneity. To mitigate both issues, we propose an augmented FLAIM approach that maintains a private proxy of heterogeneity. We simulate our methods across a range of benchmark datasets under different degrees of heterogeneity and show we can improve utility while reducing overhead.
\end{abstract}



\maketitle

\section{Introduction}
\label{sec:intro}

Modern computational applications are predicated on the availability of significant volumes of high-quality data.
Increasingly, such data is not freely available: it may not be collected in the volume needed, and may be subject to privacy concerns.
Recent regulations such as the General Data Protection Regulation (GDPR) restrict the extent to which data collected for a specific purpose may be processed for some other goal.  
The aim of \textit{synthetic data generation} (SDG) %
is to solve this problem by allowing the creation of realistic artificial data that shares the same structure and statistical properties as the original data source. 
SDG is an active area of research, offering the potential for organisations to share useful datasets while protecting the privacy of individuals \citep{assefa2020generating, mendelevitch2021fidelity, van2023beyond}.

SDG methods fall into two %
categories: deep learning \citep{ kingma2019introduction, xu2019modeling, goodfellow2020generative} %
and %
statistical models %
\citep{young2009using, zhang2017privbayes}. %
Nevertheless, without strict privacy measures in place, it is possible for SDG models to leak information about the data it was trained on \citep{murakonda2021quantifying, stadler2022synthetic, houssiau2022tapas}. %
It is common for deep learning approaches such as Generative Adversarial Networks (GANs) to produce verbatim copies of training data, breaching privacy
\citep{srivastava2017veegan, vanbreugel2023membership}. %
A standard approach to prevent leakage is to use Differential Privacy (DP) \citep{dwork2014foundations}. 
DP is a formal definition which ensures the output of an algorithm does not depend %
heavily on any one individual's data by introducing calibrated random noise. %
Under DP, statistical models have become state-of-the-art (SOTA) for tabular data and often outperform deep learning %
counterparts \citep{tao2021benchmarking, liu2022utility, ganev2023understanding}. 
Approaches are based on Bayesian networks \citep{zhang2017privbayes}, Markov random fields \citep{mckenna2019graphical} and iterative marginal-based methods \citep{liu2021iterative, aydore2021differentially, mckenna2022aim}.

Private SDG methods perform well in centralized settings where a trusted curator holds all the data. However, in many 
settings, %
data cannot be easily centralized. 
Instead, there are multiple participants %
each holding a small private dataset who wish to generate synthetic data. %
Federated learning (FL) is a paradigm that applies when multiple parties wish to collaboratively train a model without sharing data directly \citep{kairouz2019advances}. 
In FL, local data remains on-device, and only model updates are transmitted back to a central %
aggregator~\citep{mcmahan2017learning}%
. FL methods commonly adopt differential privacy to provide formal privacy guarantees and is widely used in deep learning %
\citep{mcmahan2017communication, kairouz2021practical, huba2022papaya}. %
However, there has been minimal focus on federated SDG: we only identify a recent effort of Pereira et al. to distribute Multiplicative Weights with Exponential Mechanism (MWEM) via secure multi-party computation (SMC)~\citep{pereira2022secure}. Their work focuses on a distributed setting which assumes a small number of participants are \emph{all} available to secret-share data before the protocol begins. This is not suited for the fully federated setting where there may be thousands of clients and only a small proportion available at a particular round.

In this work, we study generating differentially private tabular data in the federated setting where only a small proportion of clients are available per-round who exhibit strong data heterogeneity. %
We propose FLAIM, a novel federated analogue to the current SOTA central DP algorithm AIM \citep{mckenna2022aim}. 
We show how an analog to traditional FL training can be formed with clients performing a number of local steps before sending model updates to the server in the form of noisy marginals. 
We highlight how this naive extension can suffer severely under strong heterogeneity %
which is exacerbated when only a few clients participate per round. %
To circumvent this, we modify FLAIM %
by replacing components of central AIM with newly-built steps that are better suited %
to the federated setting, such as augmenting clients' local choices via a private proxy of skew to ensure decisions are not adversely affected by heterogeneity.

\begin{figure}[t]
    \centering   
    \includegraphics[width=0.4\textwidth]{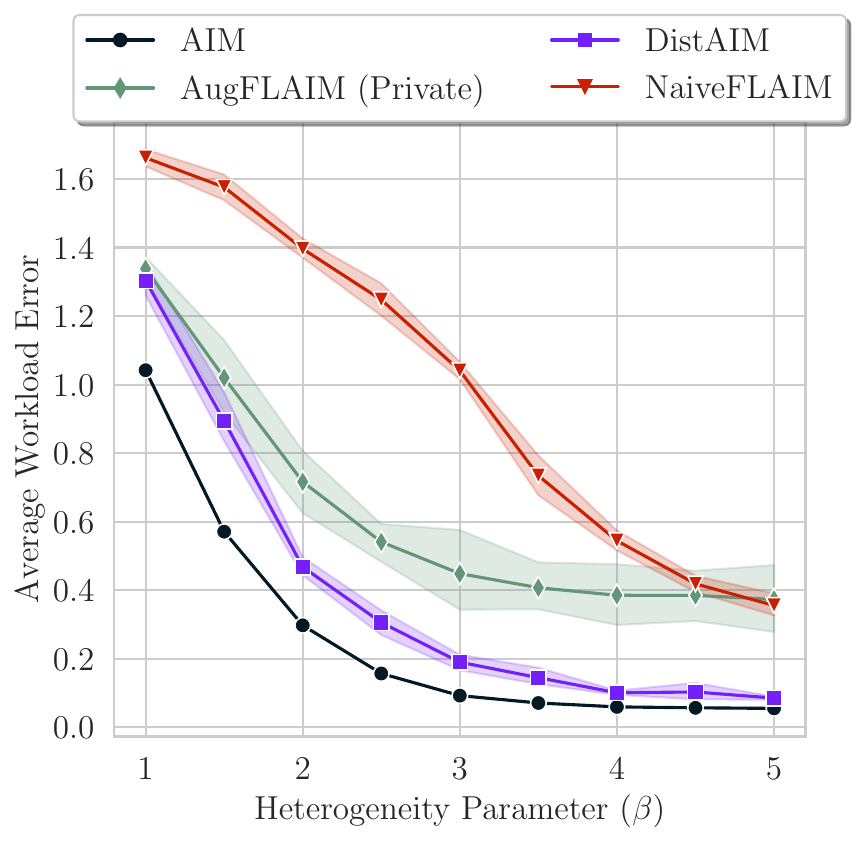}
    \caption{%
    Average error over a workload of marginals for (FL)AIM trained with $\eps=1$ on a toy federated dataset. $\beta$ varies client feature skew where large $\beta$ results in less skew. \label{fig:1}}
    \Description{The figure shows the average error over a workload of marginals whilst varying heterogeneity and demonstrates that the AugFLAIM method achieves better performance than NaiveFLAIM whilst matching that of DistAIM in high heterogeneity scenarios.}
\end{figure}
\begin{example*} Figure \ref{fig:1} presents a federated scenario where $10\%$ of 100 clients participate per round. 
Each client holds data with varying degrees of feature skew, where a larger $\beta$ implies less heterogeneity. 
We use four variations: Centralised AIM (black); Distributed AIM, our adaptation of \citet{pereira2022secure} (purple); our naive federated AIM approach (red); and our improved federated version (green). %
We plot the $L_1$ error over a workload of marginal queries trained with $\eps=1$. 
Due to client availability, there is an inevitable utility gap between central and distributed AIM. 
By naively federating AIM, client decisions made in local training are strongly affected by heterogeneity while distributed AIM is not, resulting in a big loss in utility. 
This gap is almost closed in high skew scenarios (small $\beta$) by penalising clients' local decisions via a private measure of heterogeneity (AugFLAIM).
\end{example*}
\noindent Our main contributions are as follows:
\begin{itemize}
    \item We are the first to study marginal-based methods in the federated setting. We extend the work of Pereira et al. \cite{pereira2022secure} who focus on a strongly synchronized distributed setting with MWEM to instead form a distributed protocol that replaces MWEM with AIM to obtain greater utility (DistAIM).
    \item Motivated to reduce the overheads present in DistAIM, we propose FLAIM, our federated analogue of AIM \citep{mckenna2022aim} that is designed specifically for the federated setting. We propose novel extensions based on augmenting utility scores in AIM decisions via a private proxy that reduces the effect heterogeneity has on local decisions, resulting in increased model performance and smaller overheads.
    \item We show empirically our FLAIM method outperforms federated deep learning approaches such as DP-CTGAN, which extends conclusions of prior studies to the federated setting.
    \item We perform an extensive empirical study on $7$ realistic tabular datasets. We show FLAIM obtains utility matching DistAIM but reduces the need for heavyweight SMC, resulting in less overhead. Furthermore, we show our FLAIM approaches are resistant to varying levels of heterogeneity\footnote{Our %
    code is available at 
    \url{https://github.com/Samuel-Maddock/flaim}}.
\end{itemize}
\section{Preliminaries}
\label{sec:prelim}
We assume the existence of $K$ participants each holding local datasets $D_1, \dots, D_K$ over a set of $d$ attributes such that the full dataset is denoted $D := \cup_{k} D_k$. Additionally, we assume that each attribute is categorical\footnote{We discretize continuous features via uniform binning, see Appendix \ref{appendix:datasets} for details.}. For a record $\vx := (x_1, \dots, x_d) \in D$ we denote $x_i$ as the value of attribute $i$. 
For each attribute $i \in [d]$, we define $A_i$ as the set of discrete values that $x_i$ can take. For a subset of attributes $q \subseteq [d]$ we abuse notation and let $x_q$ be the subset of $\vx$ with attributes in the set $q$. We are mostly concerned with computing marginal queries over $D$ (or individual $D_k$). Let $q \subseteq [d]$ and define $A_q := \prod_{i \in q} A_i$, as the set of values $q$ can take and $n_q := |A_q|$ as the cardinality of $q$. 

\begin{definition}[Marginal Query]
A marginal query for a subset of features $q \subseteq [d]$ is a function $M_q: \mathcal{D} \rightarrow \R^{n_q}$ where each entry is a count of the form $(M_q(D))_j := \sum_{x \in D} \1[x_q = a_j]$, $\forall j \in [n_q], a_j \in A_q$. %
\end{definition}%
As an example, consider a dataset with two features: unemployment and age where $A_1 = \{0, 1\}$ and $A_2 = \{1, 2, ..., 99\}$. The output of the marginal query $q = \{unemployment, age\}$ is a vector where an entry is a count of each record that satisfies a possible combination of feature values e.g., $\{unemp=0, age=18\}$. The goal in workload-based synthetic data generation is to generate a synthetic dataset $\hat D$ that minimises $\operatorname{Err}(D, \hat D)$ over a given workload of (marginal) queries $Q$.
We follow existing work and study the average workload error under the $L_1$ norm \citep{mckenna2022aim}.

\begin{definition}[Average Workload Error]
Denote the workload $Q = \{q_1, \dots, q_m\}$ as a set of marginal queries where each $q \subseteq [d]$. The average workload error for synthetic dataset $\hat D$ is defined %
$
\textstyle
    \operatorname{Err}(D, \hat D; Q) := \frac{1}{|Q|} \sum_{q \in Q} \|M_{q}(D) - M_{q}(\hat D)\|_1
$
\end{definition}%
We are interested in producing a synthetic dataset $\hat D$ with marginals close to that of $D$. However, in the federated setting it is often impossible to form the global dataset $D := \cup_k D_k$ due to privacy restrictions or client availability. 
Instead the goal is to gather sufficient information from local datasets $D_k$ and train a model that learns $M_q(D)$. 
For any $D_k$, the marginal query $M_q(D_k)$ and local workload error $\operatorname{Err}(D_k, \hat D)$ are defined analogously.

\paragraph{Differential Privacy (DP)}
\citep{dwork2006calibrating} is a formal notion %
that guarantees the output of an algorithm does not depend heavily on any individual. 
We seek to guarantee $(\eps,\delta)$-DP, where the parameter $\eps$ is called the \underline{privacy budget} and determines an upper bound on the privacy leakage of the algorithm. 
The parameter $\delta$ defines the probability of failing to meet this, and is set very small.
DP has many attractive properties including sequential composition, meaning that if two algorithms are $(\eps_1, \delta_1)$-DP and $(\eps_2, \delta_2)$-DP respectively, then their joint output on a specific dataset satisfies $(\eps_1 + \eps_2, \delta_1 + \delta_2)$-DP. Tighter %
bounds are obtained via zero-Concentrated DP (zCDP)~\citep{bun2016concentrated}: %

\begin{definition}[$\rho$-zCDP] \label{def:zcdp}
A mechanism $\mathcal{M}$ is $\rho$-zCDP if for any two neighbouring datasets $D, D^\prime$ and all $\alpha \in (1,\infty)$ we have $D_\alpha(\mathcal{M}(D) | \mathcal{M}(D^\prime) \leq \rho \cdot \alpha$,
where $D_\alpha$ is Renyi divergence of order $\alpha$.
\end{definition}
One can convert $\rho$-zCDP to obtain an $(\eps,\delta)$-DP guarantee. The notion of \say{adjacent} datasets can lead to different privacy definitions. We assume \underline{example-level privacy}, which defines two datasets $D$, $D^\prime$ to be adjacent if $D^\prime$ can be formed from the addition/removal of a single row from $D$. To satisfy DP it is common to require bounded \underline{sensitivity} of the function we wish to privatize. %
\begin{definition}[Sensitivity]
Let $f: \mathcal{D} \rightarrow \R^d$ be a function over a dataset. The $L_2$ sensitivity of $f$, denoted $\Delta_2(f)$, is defined as $\Delta_2(f) := \max_{D \sim D^\prime} \|f(D) - f(D^\prime)\|_2$,
where $D \sim D^\prime$ represents the example-level relation between datasets. Similarly, $\Delta_1(f)$ is defined with the $L_1$ norm as $\Delta_1(f) := \max_{D \sim D^\prime} \|f(D) - f(D^\prime)\|_1$.
\end{definition}
\noindent
We use two foundational DP 
methods %
that are core to many DP-SDG algorithms, the Gaussian and Exponential mechanisms \citep{dwork2014foundations}.

\begin{definition}[Gaussian Mechanism]
Let $f: \mathcal{D} \rightarrow \R^d$, the Gaussian mechanism is defined as $GM(f) = f(D) + \Delta_2(f) \cdot \mathcal{N}(0, \sigma^2 I_d)$.
The Gaussian mechanism satisfies $\frac{1}{2\sigma^2}$-zCDP.
\end{definition}
\begin{definition}[Exponential Mechanism]
Let $\review{u(q; \cdot)} : \mathcal{D} \rightarrow \R$ be a utility function defined for all $q \in Q$. The exponential mechanism releases $q$ with probability $\mathbb{P}[\mathcal{M}(D) = q] \propto \exp(\frac{\eps}{2\Delta} \cdot \review{u(q; D))}$,
with $\Delta := \max_q \Delta_1(\review{u(q; D)})$. This satisfies $\frac{\eps^2}{8}$-zCDP.
\end{definition}

\paragraph{Iterative Methods (Select-Measure-Generate).}
\label{sec:aim}
Recent methods for private tabular data generation follow the \say{Select-Measure-Generate} paradigm which is also the core focus of our work. %
These are broadly known as iterative methods \citep{liu2021iterative} and usually involve training a graphical model via noisy marginals over a number of steps. %
In this work, we focus on %
AIM \citep{mckenna2022aim}, an extension of the classical MWEM algorithm \citep{hardt2012simple}, which replaces the multiplicative weight update with a graphical model inference procedure called Private-PGM \citep{mckenna2019graphical}. 
PGM learns a Markov Random Field (MRF) %
and applies post-processing optimisation to ensure consistency in the generated data. PGM can answer queries without directly generating data from the model, thus avoiding additional sampling errors. 

In outline, given a workload of queries $Q$, AIM proceeds as follows (further details are in the full technical report): 
\begin{enumerate}
    \item At each round $t$, via the exponential mechanism, \textbf{select} a query $q \in Q$ that is worst-approximated by the current synthetic dataset.
    \item Under the Gaussian mechanism \textbf{measure} the chosen marginal and update the graphical model via PGM.
    \item At any point, we can \textbf{generate} synthetic data via PGM that best explains the observed measurements.
\end{enumerate}
   AIM begins round $t$ by computing utility scores for each query $q \in Q$ of the form,
\begin{align*}\textstyle
    \review{u(q; D)} = w_q \cdot (\|M_q(D) - M_q(\hat D^{(t-1)})\|_1 - \sqrt{\frac{2}{\pi}} \cdot \sigma_t \cdot n_q),
\end{align*}
where $\hat D^{(t-1)}$ is the current PGM model. 
The core idea is to select marginals that are high in error (first term) balanced with the expected error from measuring the query under Gaussian noise with variance $\sigma_t^2$ (second term). The utility scores are weighted by $w_q := \sum_{r \in Q} |r \cap q|$, which calculates the overlap of other marginals in the workload with $q$. The sensitivity of the resulting exponential mechanism is $\Delta = \max_q w_q$ since measuring $\|M_q(D) - M_q(\hat D^{(t-1)})\|_1$ has sensitivity $1$ which is weighted by $w_q$. Once a query is selected %
it is measured by the Gaussian mechanism with variance $\sigma_t^2$ and sensitivity $1$. An update to the model via PGM is then applied using all observed measurements so far. See Appendix \ref{appendix:aim} for full details.

\paragraph{Towards Decentralized Synthetic Data.} 
Given a set of $K$ clients with datasets $D_1, \dots \ D_K$ and workload $Q$, the goal is to learn a synthetic dataset $\hat D$ that best approximates $D := \cup_k D_k$ over $Q$ e.g., $M_q(\hat D) \approx M_q(D), \forall q \in Q$. 
However, computing statistics directly from $D$ is not possible as each $D_k$ is private. 
We make an important distinction here between the highly-synchronized \underline{distributed} and loosely-coordinated \underline{federated} settings. 
In the \underline{distributed} setting, all participants are available to collaboratively share $M_q(D_k)$ and some central server(s) compute steps of AIM in a strongly synchronized manner, with high communication overhead. 
This is the original setting of Pereira et al. \cite{pereira2022secure}. 
Instead we are mainly interested in the \underline{federated} setting where we assume 
that participants are more weakly engaged, and may become unavailable or dropout at any moment. 
We model this by assuming that each participant participates in the current round only with probability $p$.
We also assume each $D_k$ exhibits heterogeneity which could manifest as significant feature-skew or a varying number of samples. 
We detail how we model heterogeneity in Section \ref{sec:exp}. %

\section{Distributed AIM}
\label{sec:distaim}
Our first proposal, DistAIM, translates the AIM algorithm directly into the federated setting by having computing servers jointly calculate each step, attempting to mirror what would be computed in the central setting.
To do so, computing servers must collaborate privately and securely, such that no one participant's raw query answers, $M_q(D_k)$, are revealed. The \say{select} and \say{measure} steps require direct access to private local datasets $D_k$, 
and hence we need to implement distributed DP mechanisms for these steps. \review{We present an overview here with full details in Appendix~\ref{appendix:distaim}}.
 
Pereira et al. \cite{pereira2022secure}, describe one such approach for MWEM. They utilize various secure multi-party computation (SMC) primitives based on secret-sharing \citep{araki2016high}. However, a key difference is they assume a \underline{distributed} setting where \underline{all} participants first secret-share their workload answers to computing servers before the protocol begins. These computing servers implement secure exponential and Laplace mechanisms over shares of marginals via standard SMC operations \citep{keller2020mp}. This is a key difference to our federated setting where we assume partial participation of clients over multiple rounds. Their approach also has two drawbacks: first, their cryptographic solution incurs both a computation and communication overhead which may be prohibitive in federated scenarios. %
Secondly, their approach is based on MWEM which results in a significant loss in utility. Furthermore, MWEM is memory-intensive and does not scale to high-dimensional datasets. 

Instead, we apply the framework of Pereira et al. \cite{pereira2022secure} to AIM, and adapt this for our federated setting. Compared to AIM and Pereira et al., our DistAIM approach has important differences:

\noindent \textbf{Client participation:} At each round only a subset of participants are available to join the AIM round. For simplicity, we assume clients are sampled with probability $p$. In expectation, $pK$ clients contribute their secret-shared workload answers e.g., $\{\llbracket M_q(D_k) \rrbracket: q \in Q\}$ to computing servers. This is an immediate difference with the setting of Pereira et al. \cite{pereira2022secure}, where it is assumed all clients are available to secret-share marginals before training. Instead, in DistAIM the secret-shares from participants are aggregated across rounds and the \say{select} and \say{measure} steps are carried out via computing servers over the updated aggregate shares at a particular round. Compared to the central setting, DistAIM incurs additional error due to subsampling. 

\noindent \textbf{Select step:} A key difficulty extending AIM (or MWEM) to a distributed setting is the use of the exponential mechanism. 
In order to apply this, the utility scores \review{$u(q; D)$} must be calculated. Following Protocol 2 of Pereira et al. \cite{pereira2022secure}, sampling from the exponential mechanism can be done over secret-shares of the marginal $\llbracket M_q(D_k) \rrbracket$ since utility scores \review{$u(q; D)$} depend only linearly in $M_q(D_k)$.

\noindent \textbf{Measure step:} Once a marginal has been sampled, it must be measured. Protocol 3 in Pereira et al. \cite{pereira2022secure} proposes one way to securely generate Laplace noise between computing servers. This is then added to the aggregate sum of a secret-shared marginal. To remain consistent with AIM, we use Gaussian noise instead.
    
\noindent \textbf{Estimate step:} 
Under the post-processing properties of DP the computing server(s) are free to use the measured noisy marginals with PGM to update the graphical model, as in the centralized case.

\section{FLAIM: FL analog for AIM}
\begin{algorithm}[t]
\caption{FLAIM}\label{alg:flaim}
\begin{algorithmic}[1]
\Input $K$ participants with data $D_1, \dots, D_k$, sampling rate $p$, global rounds $T$, local rounds $s$, workload $Q$, privacy parameters $(\eps, \delta)$
\For{each global round $t=1 \dots T$}
    \State Form %
    $P_t$ by sampling each client $k$ with probability $p$
    \For{each client $k \in P_t$}
    \For{\highlight{each local step $l=1 \dots s$}} %
            \State Filter workload $Q \leftarrow Q \setminus \{|q| =  1: q \in Q\}$ \flaim{(Private)} \label{alg:filter}
            \State \highlight{Compute a heterogeneity measure for each $q \in Q$}
            \begin{align*}
                \tilde\tau_k(q) &\leftarrow 
                \begin{cases}
                    0 & \text{ \naive{Naive}} \\
                    \tau_k(q) := \|M_q(D_k) - M_q(D)\|_1 & \text{ \nonpriv{Oracle}} \\
                    \frac{1}{|q|} \sum_{j \in q} \|M_{\{j\}}(D_k) - M_{\{j\}}(\hat D^{(t-1)})\|_1 & \text{ \flaim{Private}}
                \end{cases}
            \end{align*}
            \State \textbf{Select} $q_{t+l}\in Q$ using ExpMech with utility score(s) $u(q; D_k) := w_q \left(\|M_q(D_k) - M_q(\hat D^{(t-1)+l})\|_1 - \sqrt{\frac{2}{\pi}} \cdot \sigma_{(t-1)+l} \cdot n_q \highlight{- \tilde\tau_k(q)}\right)$
            \State \textbf{Measure} marginal %
            $ \tilde{M}_{q}(D_k) := M_q(D_k) + \mathcal{N}\left(0, \sigma^2_{(t-1)+l} I\right)$
            \State \textbf{Estimate} new \highlight{local model} via PGM as\newline
            \hspace*{4mm}$
                \hat D_k^{(t-1)+l} \leftarrow \argmin_{p \in \mathcal{S}} \sum_{i=1}^{(t-1)+l} \frac{1}{\sigma_i} \| M_{q_i}(p) - \tilde{M}_{q_{i}}(D_k)\|_2
            $
        \EndFor
        \State Share all $1$-way marginals \review{ under SecAgg}, \newline \hspace*{7mm}$\mathcal{M}_k^1 \leftarrow \{(t, j, \llbracket {M}_{\{j\}}(D_k) \rrbracket \}_{j \in [d]}$  \flaim{(Private only)} \label{alg:measure1}
        \State Compute $\mathcal{M}_k \leftarrow \{(t, q_{t+l}, \review{\llbracket M_{q_{t+l}}(D_k) \rrbracket)\}_{l \in [s]}}$ \flaim{$\cup \mathcal{M}_k^1$} 
        \State Send $\mathcal{M}_k$ to the server via SecAgg \label{alg:measure2} %
    \EndFor
    \State Server updates measurement list $\mathcal{M}^{t} := \cup_{k \in P_t} \mathcal{M}_k$
    \State \textbf{for} each unique $q \in \mathcal{M}^{t}$ \review{aggregate marginals and add noise \newline \hspace*{5mm}$\tilde M^t_q := \sum_{\{k: M_q \in \mathcal{M}_k\}} \llbracket M_q(D_k) \rrbracket + N(0, \sigma^2_t)$}
    \State \label{alg:weight}
    $\text{\textbf{for} each } \review{\tilde M_q^t \text{ compute  }
        \alpha^t_q}  := \begin{cases} %
      1/\review{\sigma_t}, & \text{\naive{Naive}} \\
      N_q^t/\review{\sigma_t}, & \text{\nonpriv{Oracle}} \\
      \tilde N_q^t / \review{\sigma_t}, & \text{\flaim{Private}}
    \end{cases} 
    $
    \State \review{Server updates measurement list $\mathcal{M}$ with each $(t, q, \tilde M_q^t, \alpha_q^t)$ and updates the global model}
    \review{
    \begin{align*} \textstyle
        \hat D_{t} \leftarrow \argmin_{p \in \mathcal{S}} \sum_{ (t, q, \tilde M_q^t, \alpha_q^t) \in \mathcal M} \alpha_q^t \| M_{q}(p) - \tilde M_q^t\|_2
    \end{align*}}
\EndFor
\end{algorithmic}
\end{algorithm}

\review{While DistAIM is one solution, it is not defined within the standard federated paradigm where clients typically perform a number of local steps before sending model updates to a server. Furthermore, the SMC-based approach can have large overheads which is prohibitive for federated clients who have limited bandwidth (we explore this in Section \ref{sec:paramsetting}). This leads us to}
design an AIM approach that is analogous to traditional Federated Learning (FL), where %
only lightweight SMC is needed in the form of secure-aggregation (SecAgg) \citep{bell2020secure}. 
In FL, the paradigm for training models is to do computation on-device, having clients perform multiple local steps before sending a model update.
The server aggregates all client updates and performs an update to the global model \citep{mcmahan2017learning}. 
When combined with DP, model updates are aggregated via SecAgg schemes and noise is added either by a trusted server or in a distributed manner. In the case of AIM, we denote our analogous FL approach as FLAIM. 
In FLAIM, the selection step of AIM is performed locally by clients (across multiple local training steps). Each clients' chosen marginals are then sent to a trusted server via SecAgg and noise is added.

In more detail, FLAIM is outlined in Algorithm \ref{alg:flaim}. We present three variations, with differences highlighted in color.
Shared between all variations are the key differences with DistAIM displayed in \highlight{blue underline}. 
First is \naive{NaiveFLAIM}, a straightforward translation of AIM into the federated setting. 
In Section \ref{sec:naive}, we explain the shortcomings of such an approach which stems from scenarios where clients' local data exhibits strong heterogeneity. Motivated by this, Section \ref{sec:nonpriv} proposes \nonpriv{AugFLAIM (Oracle)} a variant of FLAIM that assumes oracle access to a measure of skew which can be used to augment local utility scores. 
This skew measure is non-private and not obtainable in practice, but provides an idealized baseline.
Lastly, Section \ref{sec:flaim} introduces \flaim{AugFLAIM (Private)}, which again augments local utility scores but with a private proxy of heterogeneity alongside other heuristics to improve utility.

All FLAIM variants proceed by sampling clients %
to participate in round $t$.
Each client performs a number of local steps $s$, which consist of performing a local selection step using the exponential mechanism, measuring the chosen marginal under local noise and updating their local model via PGM. When each client finishes local training, they send back each chosen query $q$ alongside the associated marginal \review{$\llbracket M_q(D_k)\rrbracket$,} which are aggregated via secure-aggregation and noise is added by the central server. Hence, local training is done under local differential privacy (LDP) to not leak any privacy between steps, 
whereas the resulting global update is under a form of distributed DP \review{where noise is added by the central server to the securely-aggregated marginals}. We assume all AIM methods run for $T$ global rounds. AIM can also set $T$ adaptively via budget annealing and we explore this in our experiments (see Appendix \ref{appendix:hyper} for details). 

\subsection{NaiveFLAIM and Heterogeneous Data}
\label{sec:naive}

NaiveFLAIM is our first attempt at a SDG in the federated setting, by directly translating the AIM algorithm. 
However, in federated settings, participants often exhibit strong heterogeneity in their local datasets. That is, clients' local datasets $D_k$ can differ significantly from the global dataset $D$.
Such heterogeneity will affect AIM in both the \say{select} and \say{measure} steps. If $D_k$ and $D$ are significantly different then the local marginal $M_q(D_k)$ will differ from the true marginal $M_q(D)$. 
We quantify heterogeneity for a client $k$ and query $q \in Q$ via the %
$L_1$ distance:
$$\tau_k(q) := \|M_q(D_k) - M_q(D)\|_1.$$
This can be viewed as a measure of query skew. In FLAIM, we proceed by clients perform a number of local steps. The first stage involves carrying out a local \say{select} step based on utility scores of the form $ u(q; D_k) \propto \|M_q(D_k) - M_q(\hat D^{(t-1)})\|_1$.
Suppose for a particular client $k$ there exists a query $q \in Q$ such that $M_q(D_k)$ exhibits strong heterogeneity. 
If at step $t$ the current model $\hat D^{(t-1)}$ is a good approximation of $D$, then it is probable that client $k$ ends up selecting any query that has high heterogeneity since $u(q; D_k) \propto ||M_q(D_k) - M_q(\hat D^{(t-1)})||_1 \approx ||M_q(D_k) - M_q(D)||_1 = \tau_k(q)$. 
This mismatch can harm model performance and is compounded by having many clients select (multiple and possibly differing) skewed marginals and so the model is updated in a way that drifts from $D$.

\subsection{AugFLAIM (Oracle): Tackling Heterogeneity}
\label{sec:nonpriv}
The difficulty above arises as clients choose marginals via local applications of the exponential mechanism with a score that does not account for underlying skew. %
We have $u(q; D_k) \propto$
\begin{align*}
    \|M_q(D_k) - M_q(\hat D)\|_1 &\leq \|M_q(D) - M_q(\hat D) \| + \|M_q(D_k) - M_q(D) \|_1 \\&\propto \review{u(q;D)} + \tau_k(q)
\end{align*}
To circumvent this, we should correct local utility scores by down-weighting marginals based on $\tau_k(q)$ and modify utility scores as: $$u(q; D_k) \propto \|M_q(D_k) - M_q(\hat D) \| - \tau_k(q)$$ where $\tau_k(q)$ is an exact $L_1$ measure of heterogeneity for client $k$ at a marginal $q$. 
Unfortunately, measuring $\tau_k(q)$ under privacy constraints is not feasible. 
That is, $\tau_k(q)$ depends directly on $M_q(D)$, which is exactly what we are trying to learn via AIM.
Still, we introduce 
\nonpriv{AugFLAIM (Oracle)} as 
an idealized baseline to compare with.  
AugFLAIM (Oracle) is a variation of FLAIM that assumes oracle access to $\tau_k(q)$ and augments local utility scores as above.

\subsection{AugFLAIM (Private): Heterogeneity Proxy}
\label{sec:flaim}
Since augmenting utility scores directly via $\tau_k(q)$ is not feasible, we seek a proxy $\tilde \tau_k(q)$ that is reasonably correlated with $\tau_k(q)$ and can be computed under privacy. 
This proxy measure can be used to correct local utility scores, %
penalising queries via $\tilde \tau_k(q)$. This helps ensure clients select queries that are not adversely affected by heterogeneity. We propose the following proxy
\begin{align*}
\textstyle
    \tilde \tau_k(q) := \frac{1}{|q|} \sum_{j \in q} \| M_{\{j\}}(D_k)  - \tilde M_{\{j\}}(D)\|_1
\end{align*}
Instead of computing a measure for each $q \in Q$, we compute one for each feature $j \in [d]$, where $\tilde M_{\{j\}}(D_k)$ is a noisy estimate of the 1-way marginal for feature $j$. For a particular query $q \in Q$, we average the skew of the associated features contained in $q$. 
Such a $\tilde \tau_k(q)$ relies only on estimating the distribution of each feature. 
This estimate can be refined across multiple federated rounds as each participant can measure $M_{\{j\}}(D_k)$ for each $j \in [d]$ and have the server sum and add noise (via SecAgg) to produce a new private estimate $\tilde M_{\{j\}}(D)$ each round. We add two further enhancements:

\noindent
\textbf{1. Filtering and combining 1-way marginals \review{(Line \ref{alg:measure1})}.} As we require clients to estimate all features at every round, we remove 1-way marginals from the workload to prevent clients from measuring the same marginal twice. All 1-way marginals that are estimated for $\tilde \tau_k(q)$ are fed back into PGM to improve the global model. %

\noindent
\textbf{2. Weighting Marginals \review{(Line \ref{alg:weight})}.} In PGM, measurements are weighted by $\alpha = 1/\sigma_t$, so those that are measured with less noise have more importance in the optimisation of model parameters. 
Both AugFLAIM variations adopt
an additional weighting scheme that includes the total sample size \review{that contributed to a particular marginal $q$ at round $t$, $N_q^t := \sum_{\{k: M_q^t \in \mathcal{M}_k\}} |D_k|$ where the weight becomes $\alpha_q^t = N_q^t/\sigma_t$. This relies on knowing the number of samples that are aggregated.}
In some cases, 
the size of datasets may be deemed private. 
\review{In such scenarios, it can be estimated \review{from the noisy marginal $\tilde M_q$ by summing the counts to produce $\tilde N_q$}}. %

The privacy guarantees of all FLAIM variations follow  directly from those of AIM. 
The use of a heterogeneity measure incurs an additional sensitivity cost for the exponential mechanism and AugFLAIM (Private) incurs an additional privacy cost as it measures each of the $d$ features at every round. 
The following lemma captures this. See Appendix \ref{appendix:flaim} for the full proof.
\begin{lemma}
\label{lem:dp}
    For any number of global rounds $T$ and local rounds $s$, FLAIM satisfies $(\eps,\delta)$-DP , under Gaussian budget allocation $r \in (0,1)$ by computing $\rho$ according to Lemma \ref{lemma:cdp}, and setting
    \begin{align*}
        \sigma_t &= \begin{cases}
          \sqrt{\frac{Ts + d}{2\cdot r \cdot \rho}}, & \text{\naive{Naive} or \nonpriv{Oracle}} \\
          \sqrt{\frac{T(s+d)}{2\cdot r \cdot \rho}}, & \text{\flaim{Private}}
        \end{cases}, \; \eps_t = \sqrt{\frac{8 \cdot (1-r) \cdot \rho}{Ts}}
    \end{align*}
    For AugFLAIM methods, the exponential mechanism is applied with sensitivity $\Delta := \max_q 2w_q$.
\end{lemma}

\begin{table}[t!]
    \caption{Comparison of FLAIM approaches against baselines for negative log-likelihood (NLL), $\eps=5$. Smaller NLL is better.}
    \label{tab:baselines}
    \centering
    \begin{tabular}{llll}
    \toprule
     \textbf{Method / Dataset} & Adult & Credit & Covtype \\
    \midrule
    Fed DP-CTGAN & 37.1 & 83.8 & 62.7 \\
    FedNaiveBayes &  25.33 & 18.02  & 44.9 \\
    FLAIM (Random) & 83.9 & 47.7 & 58.4 \\
    \naive{NaiveFLAIM} & 29.4 & 18 & 45.4 \\
    \flaim{AugFLAIM (Private)} &  \textbf{20.87} & \textbf{16.2} & \textbf{41.6} \\
    \hline
    DP-CTGAN & 28.6 & 27.6 & 45.9 \\
    AIM &  \textbf{19.2} &  \textbf{15.57} & \textbf{40.92} \\
    \hline
    \end{tabular}
\end{table}

\begin{figure*}[t]
\centering
  \subfloat[Adult]{%
       \includegraphics[width=0.33\linewidth]{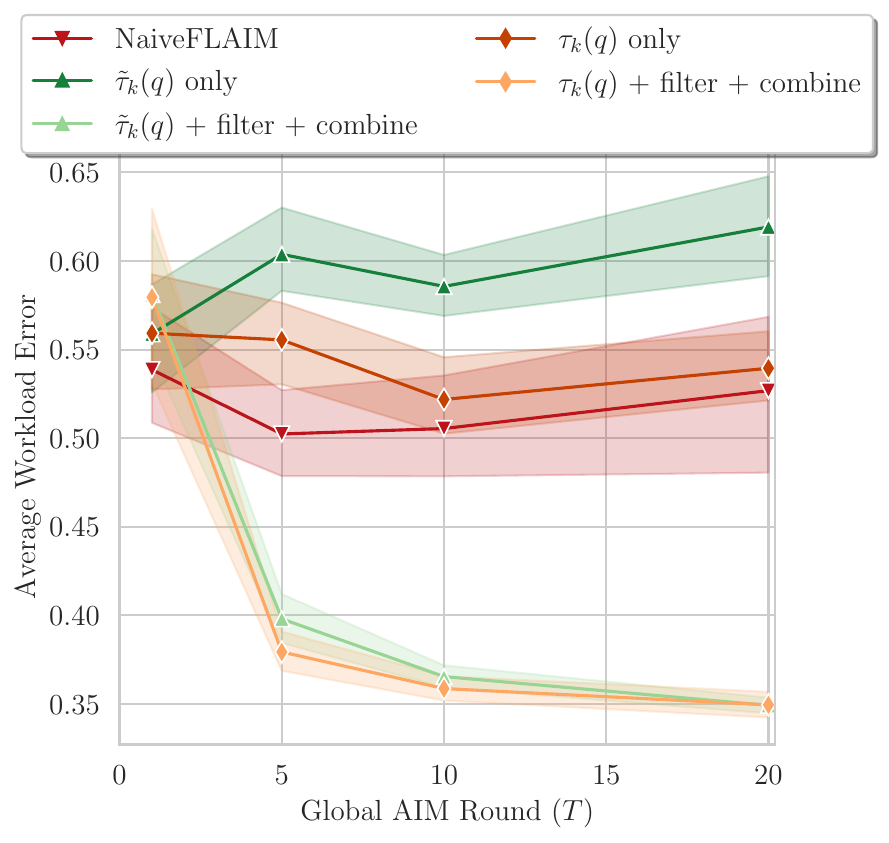}}
  \subfloat[Credit]{%
        \includegraphics[width=0.33\linewidth]{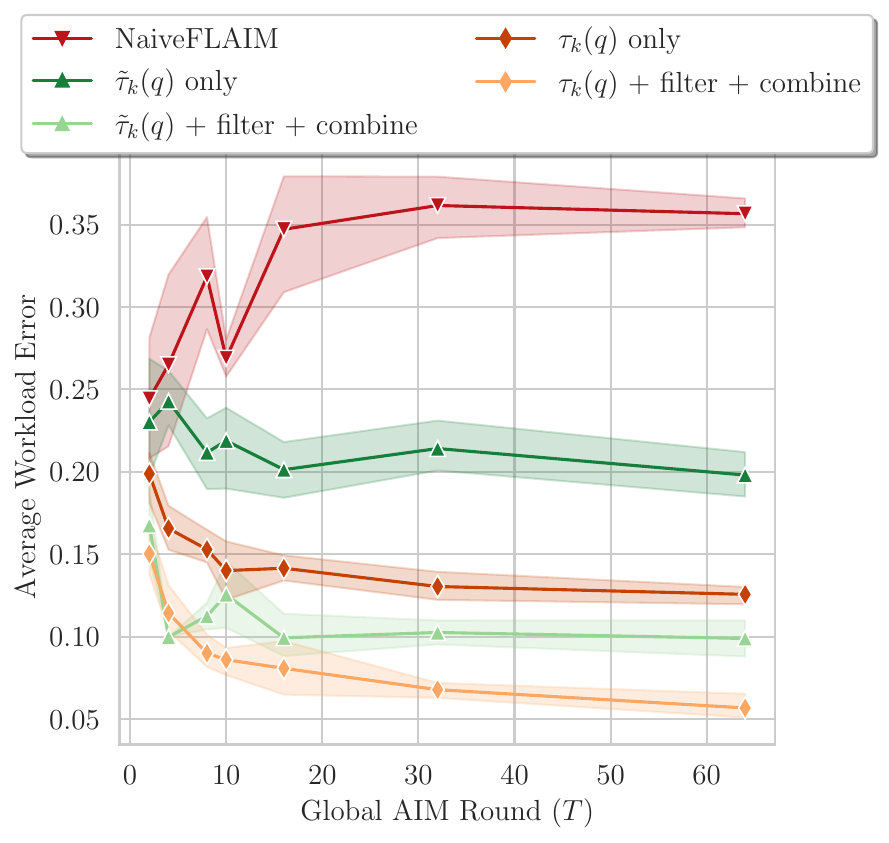}}
  \subfloat[Covtype]{%
       \includegraphics[width=0.33\linewidth]{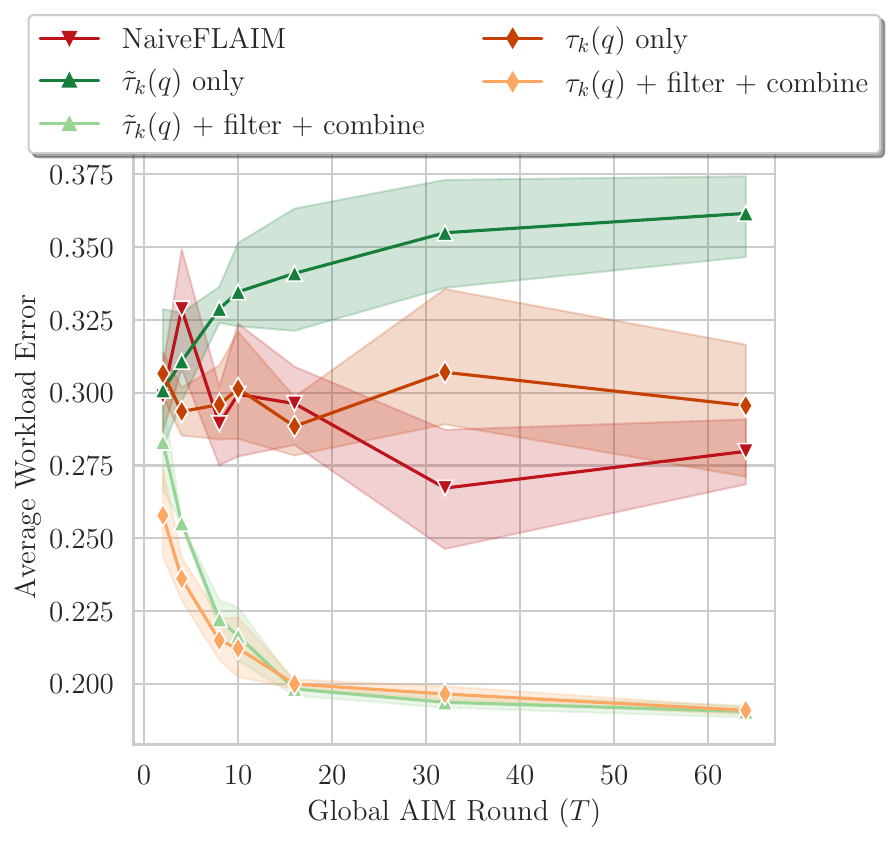}}
  \caption{Ablation study, comparing utility for FLAIM variations that augment local utility scores, $\eps=5, T=10, s=1, p=0.1$\label{fig:ablation}}
  \Description{Ablation study over the Adult, Credit and Covtype datasets.}
\end{figure*}

\section{Experimental Evaluation}
\label{sec:exp}

For our experiments, we utilize realistic benchmark tabular datasets from the UCI repository \citep{Dua:2019}: Adult, Magic, Marketing and Covtype. We further use datasets common for benchmarking synthetic data: Census and Intrusion from the Synthetic Data Vault (SDV) \citep{SDV}. We also construct a toy dataset with feature-skew denoted SynthFS. Full details on all datasets are contained in Appendix \ref{appendix:datasets}.

We evaluate our methods in three ways: average workload error (as defined in Section \ref{sec:prelim}), average negative log-likelihood (evaluated on test data) and the area under the curve (AUC) of a decision tree model trained on synthetic data and evaluated on test data. 

For all datasets, we simulate heterogeneity by forming non-IID splits in one of two ways: The first is by performing dimensionality reduction and clustering nearby points to form client partitions that have strong feature-skew. We call this the \say{clustering} approach. For experiments that require varying heterogeneity, we form splits via an alternative label-skew method popularized by Li et al. \cite{li2022federated}. This samples a label distribution $p_c \in [0,1]^K$ for each class $c$ from a Dirichlet$(\beta)$ where larger $\beta$ results in less heterogeneity. See Appendix \ref{appendix:hetero} for full details. In the following sections, all experiments have $K=100$ clients with partitions formed from the clustering approach unless stated otherwise. We train (FL)AIM models on a fixed workload of $3$-way marginal queries chosen at random with $|Q| = 64$ and average results over $10$ independent runs. Further experiments on datasets besides Adult are contained in Appendix \ref{appendix:experiments}.

\begin{figure*}[t]
\centering
  \subfloat[Varying privacy budget $(\eps)$\label{fig:vary_eps}]{%
       \includegraphics[width=0.32\linewidth]{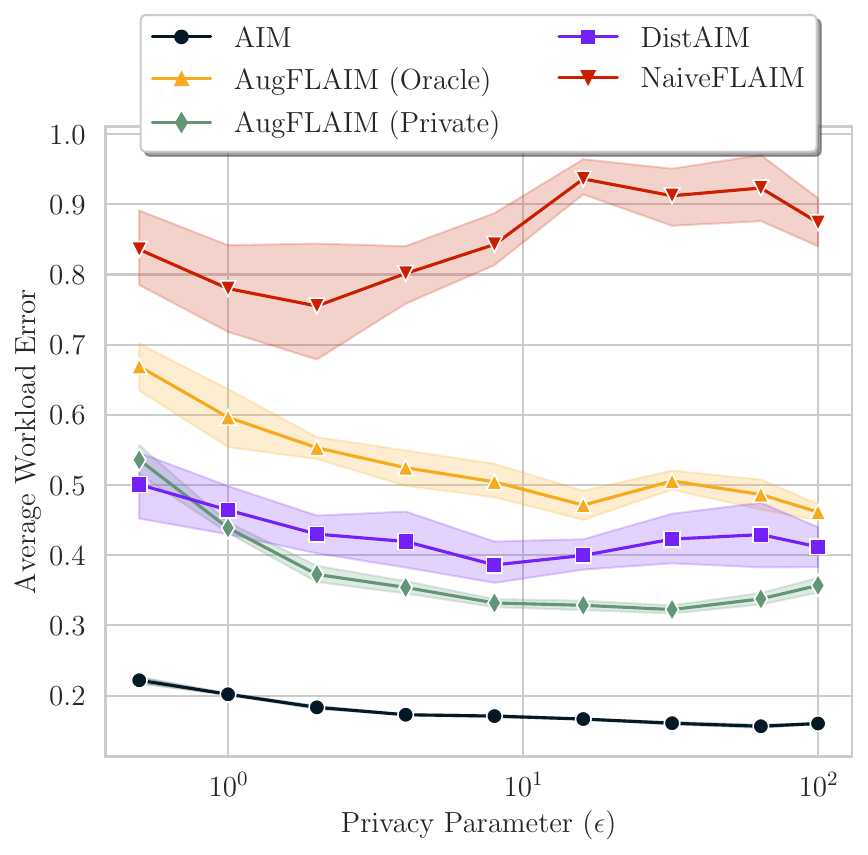}}
  \subfloat[Varying global rounds $(T)$\label{fig:vary_t}]{%
        \includegraphics[width=0.32\linewidth]{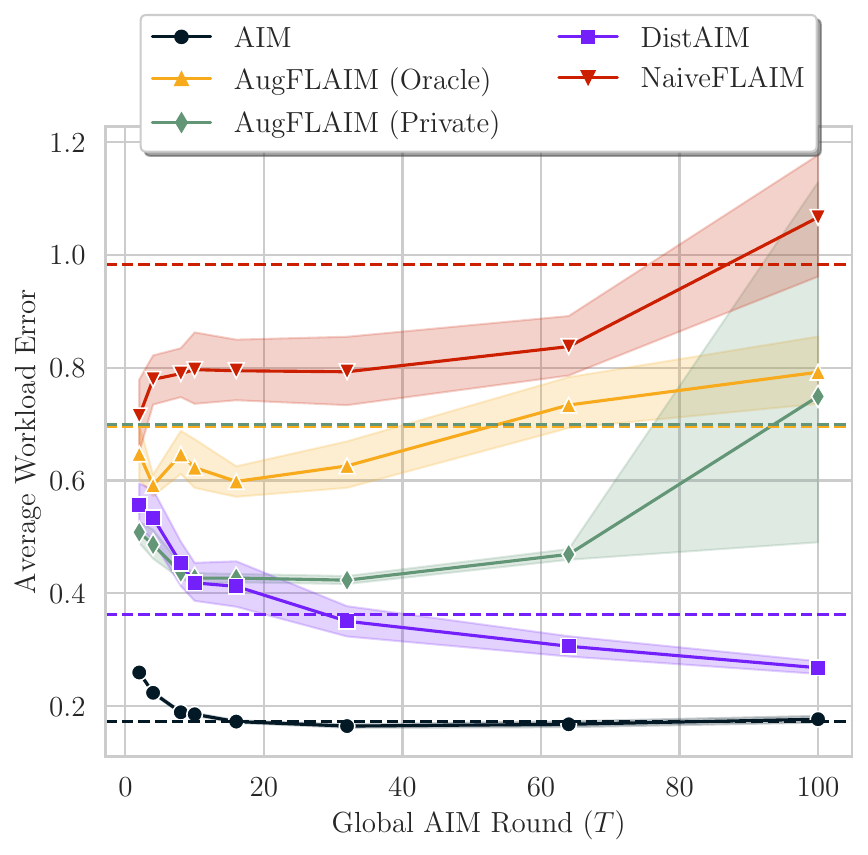}}
  \subfloat[Varying client sample-rate $(p)$ \label{fig:vary_p}]{%
       \includegraphics[width=0.32\linewidth]{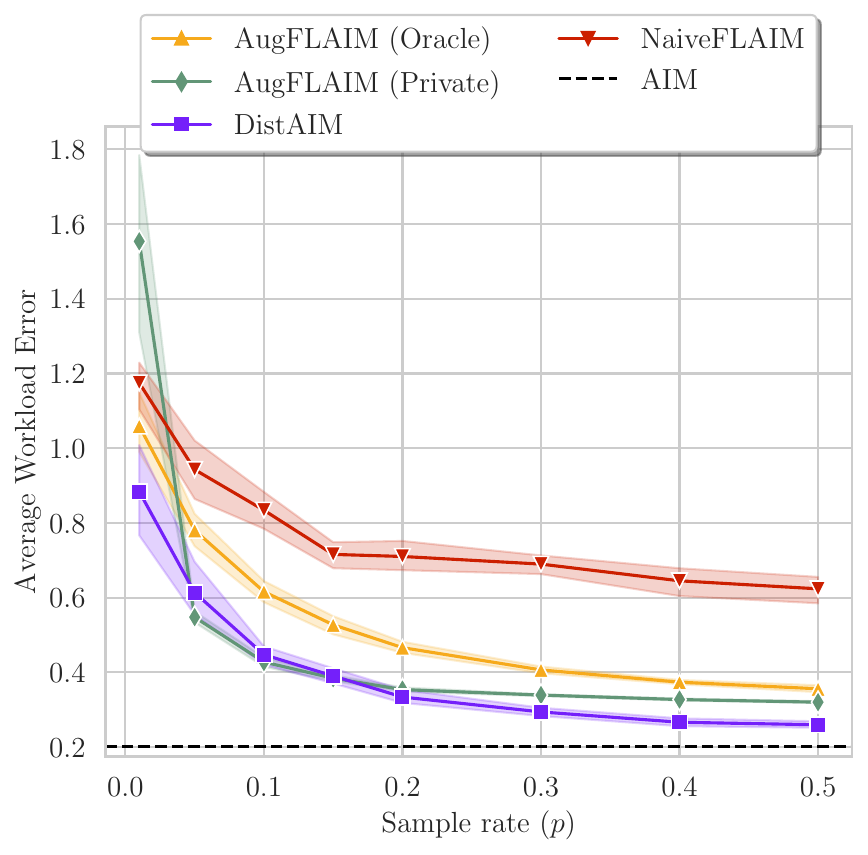}} \\
  \subfloat[Varying $\beta$ in label-skew \label{fig:vary_beta}]{%
       \includegraphics[width=0.32\linewidth]{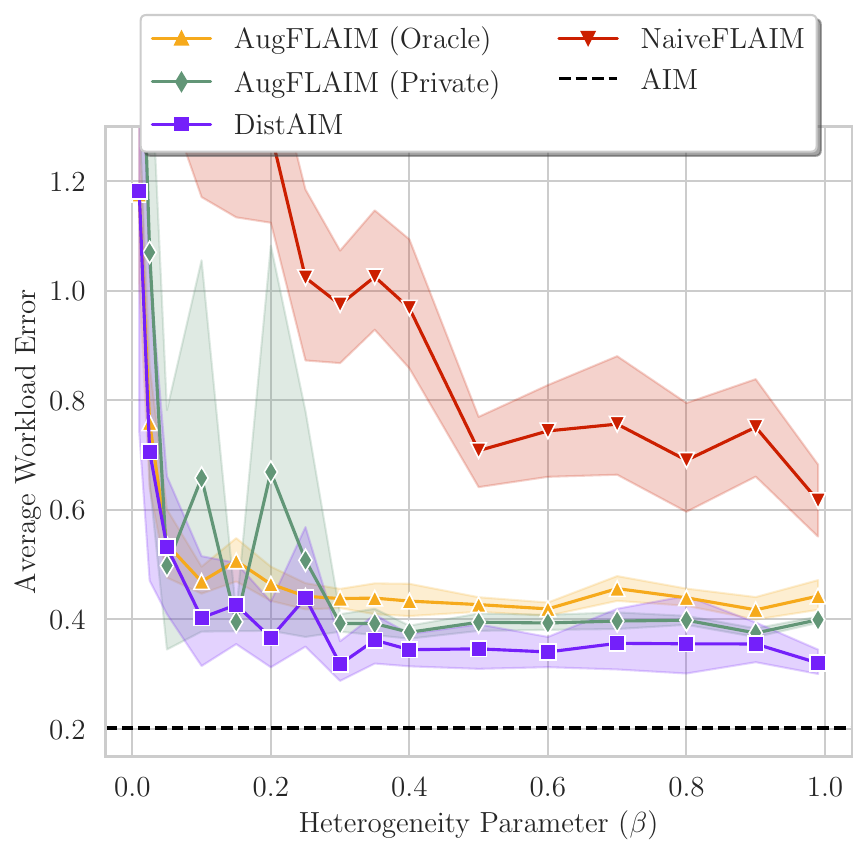}}
  \subfloat[Varying $T,s$ - Workload Error\label{fig:vary_s_eps1}]{%
        \includegraphics[width=0.32\linewidth]{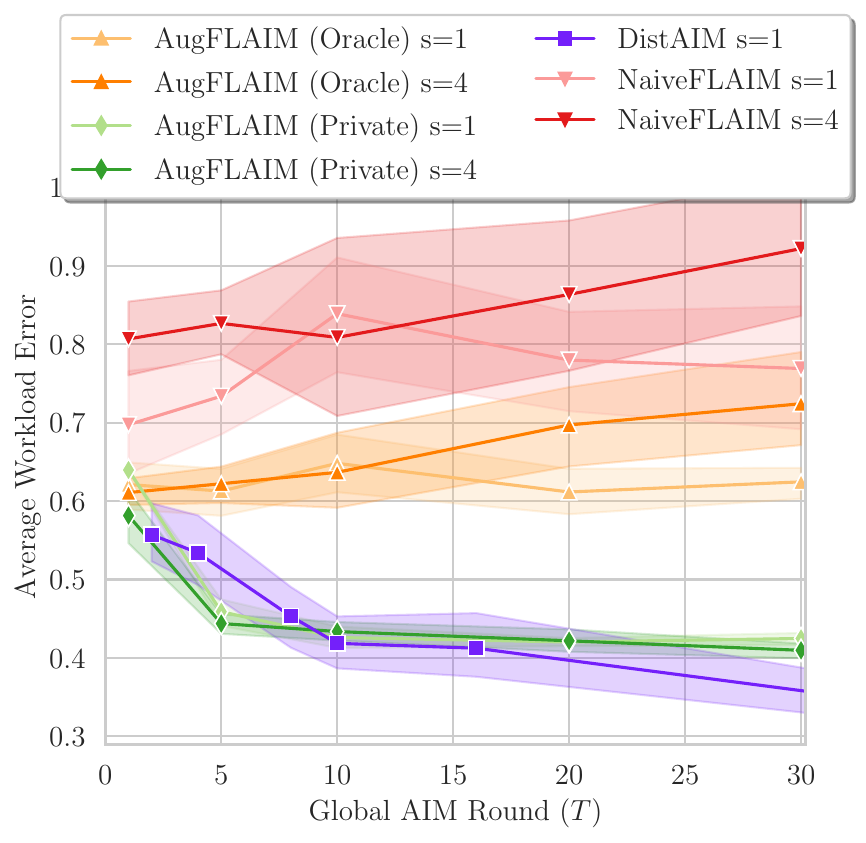}}
  \subfloat[\review{Varying $T,s$ - Test AUC \label{fig:vary_s_eps10}}]{%
        \includegraphics[width=0.325\linewidth]{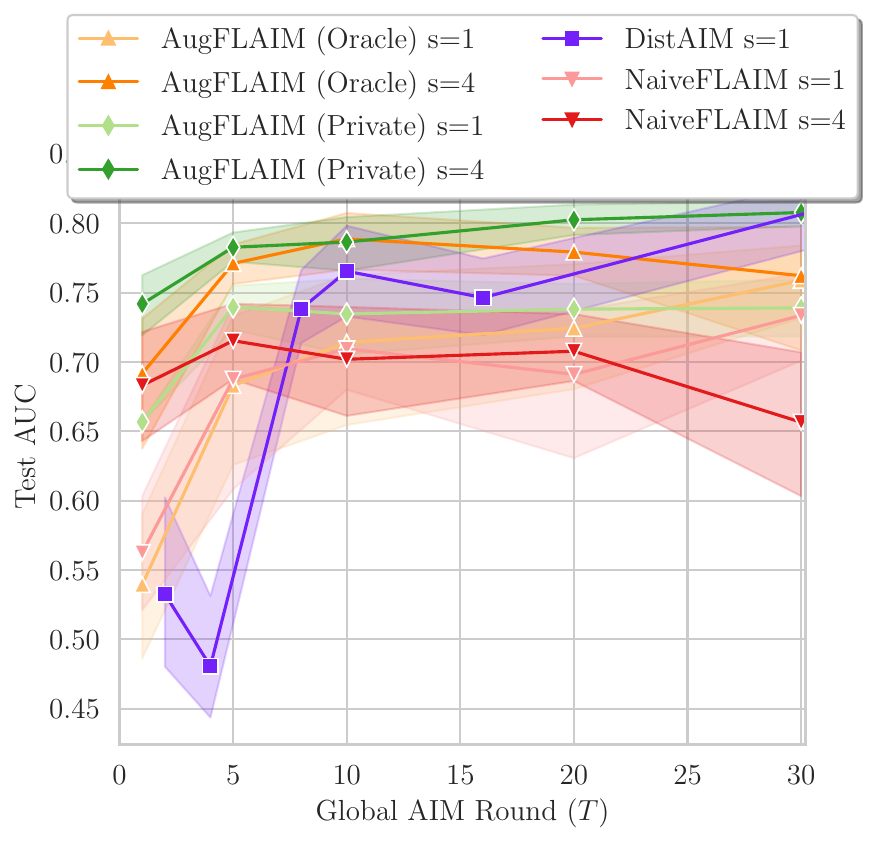}}
  \caption{Varying (FL)AIM Parameters on Adult; Unless otherwise stated $T=10, s=1, p=0.1, K=100$, $\eps=1$}
  \Description{Varying FLAIM hyper-parameters on Adult for $T, p, \beta$ and $s$.}
\end{figure*}

\subsection{Comparison with Existing Baselines}\label{sec:baselines}

We begin with an experiment comparing AugFLAIM (Private) to other federated baselines. One such SOTA approach is CTGAN \citep{xu2019modeling}. We utilise the DP-CTGAN implementation within OpenDP's smartnoise-sdk \citep{opendp} to compare to AIM. For the federated setting we train CTGAN using DP-FedSGD via FLSim \citep{flsim}. For details on hyperparameters see Appendix \ref{appendix:hyper}. We further compare AugFLAIM (Private) against two AIM baselines. \underline{FLAIM (Random)} which takes FLAIM and randomly chooses a query $q \in Q$ without utilising the exponential mechanism. 
Instead, all privacy budget is spent on the \say{Measure} step. 
The other is \underline{FedNaiveBayes} which restricts the workload $Q$ to only 1-way marginals and is equivalent to training a NaiveBayes model. In Table \ref{tab:baselines} we present the negative log-likelihood (NLL) for models trained to an $\eps=5$ across three datasets. Methods achieving lowest NLL for a particular dataset are in bold. 

For the central setting, AIM achieves better performance than DP-CTGAN across each dataset. This confirms prior studies such as \cite{liu2022utility} that show graphical model approaches achieve better utility than deep learning methods for tabular data. For the federated setting, we note FedNaiveBayes and FLAIM (Random) both perform poorly in comparison to AugFLAIM (Private). This illustrates two main points: utilising the exponential mechanism does result in a substantial increase in utility (i.e., randomly choosing $q \in Q$ is poor) and that utilising a workload of $k$-way marginals with $k > 1$ gives best utility (i.e., NaiveBayes is poor). Further note, AugFLAIM (Private) has better utility than NaiveFLAIM which shows augmenting utility scores in the Exponential mechanism does improve utility. We explore this further in Section \ref{sec:paramsetting}. Finally, Fed DP-CTGAN performs very poorly compared to AugFLAIM. Even NaiveFLAIM and occasionally FedNaiveBayes outperform it. There are further issues for practitioners: first, Fed DP-CTGAN requires a large number of hyperparameters to be tuned for best utility such as client and server learning rates and the clipping norm. This is in contrast to (FL)AIM methods that only have a single hyperparameter - the total number of global rounds $T$. Secondly, CTGAN requires a large number of training epochs. In this experiment we train for 50 epochs which is equivalent to $T=500$ rounds whereas the FLAIM methods achieve better utility in only $T=10$ rounds. For these reasons, in further experiments, we do not compare to federated CTGAN.

\subsection{Ablation Study: Utility of AugFLAIM} \label{sec:ablation}
To understand what determines the utility of AugFLAIM (Private) we present an ablation study in Figure \ref{fig:ablation}. Here we train FLAIM models with $\eps=5$ on Adult, Credit and Covtype whilst varying the global rounds $T$. We present NaiveFLAIM compared with variations that augment the utility scores of the Exponential mechanism. These are: using the true heterogeneity measure $\tau_k(q)$ only (otherwise denoted \nonpriv{AugFLAIM (Oracle)}); using $\tau_k(q)$ with the filter-and-combine heuristic; using the private heterogeneity proxy $\tilde \tau_k(q)$ only and using $\tilde \tau_k(q)$ with the filter and combine heuristic (otherwise denoted \flaim{AugFLAIM (Private)}). On the Credit dataset, using $\tau_k(q)$ or $\tilde \tau_k(q)$ only results in a clear improvement over NaiveFLAIM, and when combined with the heuristics the lowest error is obtained. On Adult and Covtype, using only $\tau_k(q)$ or the private proxy $\tilde \tau_k(q)$ does not immediately result in lower workload error than NaiveFLAIM. Instead, utilising the filter and combine heuristics results in the best workload error overall. In further experiments, we denote \flaim{AugFLAIM (Private)} as the method which augments utility scores with $\tilde \tau_k(q)$ and uses the filter and combine heuristic whereas we denote \nonpriv{AugFLAIM (Oracle)} as the method that has oracle access to $\tau_k(q)$ only (without further heuristics).

\subsection{Parameter Settings} \label{sec:paramsetting}

Having concluded that AugFLAIM (Private) achieves the best performance against other federated baselines, we now present a detailed set of experiments comparing FLAIM methods with DistAIM across a variety of federated settings. We compare AIM and ~DistAIM against NaiveFLAIM and our two variants that augment local utility scores: AugFLAIM (Oracle) using $\tau_k(q)$ \review{only} and AugFLAIM (Private) using proxy $\tilde \tau_k(q)$ \review{with filtering and combining 1-ways}. 

\paragraph{Varying the privacy budget $(\eps)$.}  In Figure \ref{fig:vary_eps}, we plot the workload error whilst varying $\eps$ on Adult, sampling $10\%$ of clients per round and setting $T=10$. First, we observe a clear gap in performance between DistAIM and central AIM due to the error from subsampling a small number of clients per round. We observe that naively federating AIM gives the worst performance even as $\eps$ becomes large. Furthermore, augmenting utility scores makes a clear improvement in workload error, particularly for $\eps > 1$. By estimating feature distributions at each round, AugFlaim (Private) can obtain performance that matches or sometimes improves upon DistAIM for larger values of $\eps$. \review{We further note that AugFLAIM (Private) has lower error than AugFLAIM (Oracle) which may seem counter-intuitive. However, AugFLAIM (Oracle) is still trained under DP, only oracle access to $\tau_k(q)$ is assumed which is non-private and trained without heuristics as explored in Section \ref{sec:ablation}.}

\begin{table*}[t]
    \caption{Performance on datasets $K=100, p=0.1, \eps=1, T=10$. Results show workload error and negative log-likelihood. Metrics are bold if a federated method achieves lowest on a specific dataset.}
    \label{tab:1}
    \centering
\begin{tabular}{llllllll}
\toprule
\textbf{Method / Dataset} &  Adult &   Magic  & Census & Covtype & Credit & Intrusion & Marketing \\
\midrule
\naive{NaiveFLAIM} &   0.8 / 29.28 &  1.64 / 2587.5 & 0.72 / 322.29	& 0.3 / 47.44	& 0.42 / 20.22	& 0.7 / 27.92	& 1.05 / 186.44\\
\nonpriv{AugFLAIM (Oracle)} &  0.62 / 23.91 &   1.18 / 62.84  & 0.61 / 44.5	& 0.26 / 45.11	& 0.21 / 16.84 & 0.48 / 17.92 & 0.87 / 38.27\\
\flaim{AugFLAIM (Private)} &  0.43 / 21.74 &    \textbf{1.07} / 
\textbf{28.9}  & \textbf{0.48} / \textbf{41.87} & \textbf{0.17} / \textbf{42.61} & \textbf{0.14} / \textbf{16.33}	& \textbf{0.32} / 15.05 & \textbf{0.64} / \textbf{30.17}
 \\ 
DistAIM             &  \textbf{0.42} / \textbf{21.41} &   1.08 / 35.04   & 0.54 / 55.19 & 0.2 / 45.94	& 0.18 / 16.9 & \textbf{0.32} / \textbf{13.19} & 0.66 / 40.57\\
\midrule
AIM  &    0.2 / 19.3 &    0.85 / 23.7 & 0.28 / 34.26 & 0.06 / 41.06 &	0.07 / 15.62 &	0.12 / 9.96	& 0.24 / 20.87\\
\bottomrule
\end{tabular}
\end{table*}

\begin{table}[t]
    \caption{DistAIM vs. FLAIM at optimal $T$, metrics with $\uparrow$ show overhead of DistAIM over FLAIM and $\downarrow$ show $\%$ improvement of DistAIM over FLAIM. Client throughput is additionally stated in megabytes (MB) for DistAIM vs. FLAIM.}
    \label{tab:overhead}
    \centering
    \begin{tabular}{lllll}
    \toprule
     \textbf{Dataset} & $T (\uparrow)$ & \textbf{Throughput} ($\uparrow$) &  \textbf{Err} $(\downarrow)$ & \textbf{NLL} ($\downarrow$) \\
    \midrule
    Adult & 2$\times$& 1300$\times$ ($80$ / 0.06) & 58$\%$ & 11$\%$ \\
    Magic & 3.2$\times$ & $1643\times$ (80 / 0.04) & $20\%$ & $14\%$ \\
    Census &          1.5x &                 64x  (29.6 / 0.46) &       $79\%$ &               $33\%$ \\
    Intrusion &          2.5x &                366x (101 / 0.28)&      $82\%$ &               $52\%$ \\
    Marketing &          2.0x &                 97x (18 / 0.19) &      $77\%$ &               $35\%$ \\
    Credit &          1.0x &                167x (93 / 0.55)&      $45\%$ &               $6\%$ \\
    Covtype &         1.25x &                 10x (7.6 / 0.76) &      $64\%$ &               $3\%$ \\
    \bottomrule
    \end{tabular}
\end{table}

\paragraph{Varying the number of global AIM rounds $(T)$.} In Figure \ref{fig:vary_t}, we vary the number of global AIM rounds and fix $\eps=1$. Additionally, we plot the setting where $T$ is chosen adaptively by budget annealing. This is shown in dashed lines for each method. First observe with DistAIM, the workload error decreases as $T$ increases. Since computing servers aggregate secret-shares across rounds, then as $T$ grows large, most clients will have been sampled and the server(s) have workload answers over most of the (central) dataset. For all FLAIM variations, the workload error usually increases when $T$ is large, since they are more sensitive to the increased amount of noise that is added. For NaiveFLAIM, this is worsened by client heterogeneity. Further, we observe that for AugFLAIM (Private), the utility matches that of DistAIM unless the choice of $T$ is very large. At $T=100$, the variance in utility is high, sometimes even worse than that of NaiveFLAIM. This is since the privacy cost scales in both the number of rounds and features, resulting in too much noise. In the case of annealing, $T$ is chosen adaptively by an early stopping condition \review{(see Appendix \ref{appendix:hyper})}. %
While annealing has good performance in central AIM, it obtains poor utility across all federated methods. 
For annealing on Adult, AugFLAIM (Private) matches AugFLAIM (Oracle) and both perform better than NaiveFLAIM. Overall, we found choosing $T$ to be small $(\leq 30)$ gives best performance for AugFLAIM and should avoid using budget annealing. 

\paragraph{Client-participation $(p)$.} In Figure \ref{fig:vary_p}, we plot the average workload error whilst varying the per-round participation rate ($p$) with $T=10, \eps=1$.
We observe clearly the gap in performance between central AIM and DistAIM is caused by the error introduced by subsampling and when $p \geq 0.5$ performance is almost matched. 
For NaiveFLAIM, we observe the performance improvement as $p$ increases is slower than other methods. 
When $p$ is large, NaiveFLAIM receives many measurements, each likely to be highly heterogeneous %
and thus the model struggles to learn consistently. For both AugFLAIM variations, we observe the utility improves with client participation but does eventually plateau. AugFLAIM (Private) consistently matches the error of DistAIM except when $p$ is large, but we note this is not a practical regime in FL.

\paragraph{Varying heterogeneity ($\beta$).} In Figure \ref{fig:vary_beta}, we plot the average workload error on the Adult dataset over client splits formed by varying the heterogeneity parameter $(\beta)$ to produce label-skew. Here, a larger $\beta$ corresponds to a more uniform partition and therefore less heterogeneity. %
In the label-skew setting, data is both skewed according to the class attribute of Adult and the number of samples, with only a few clients holding the majority of the dataset. We observe that when the skew is large ($\beta < 0.1$), all methods struggle. As $\beta$ increases and skew decreases, NaiveFLAIM performs the worst and AugFLAIM (Private) has stable error, close to that of DistAIM.

\paragraph{Varying local rounds ($s$).} 
A benefit of the federated setting is that clients can perform a number of local steps before sending all measured marginals to the server. 
However, for FLAIM methods, this incurs an extra privacy cost in the number of local rounds $(s)$. 
In Figure \ref{fig:vary_s_eps1}, we vary $s \in \{1,4\}$ and plot the workload error. %
Although there is an associated privacy cost with increasing $s$, the errors are not significantly different for small $T$. 
As we vary $T$, the associated privacy cost becomes larger and the workload error increases for methods that perform $s=4$ local updates. 
Although increasing the number of local rounds $(s)$ does not result in lower workload error, and in cases where $T$ is misspecified can give far worse performance, it is instructive to instead study the test AUC of a classification model trained on the synthetic data. %
In Figure~\ref{fig:vary_s_eps10} we see that performing more local updates can give better test AUC after fewer global rounds. 
For AugFLAIM (Private), this allows us to match the test AUC performance of DistAIM on Adult. %

\paragraph{Comparison across datasets.} Table \ref{tab:1} presents results across all datasets with client data partitioned via the clustering approach. We set $\eps=1, p = 0.1$ and $T=10$. 
For each method we present both the average workload error and the negative log-likelihood over a holdout set. 
The first is a form of training error and the second a measure of generalisation. We observe that on $5$ of the 
$7$ datasets AugFLAIM (Private) achieves the lowest negative log-likelihood and workload error. On the other datasets, AugFLAIM (private) closely matches DistAIM in utility but with lower overheads. 

\paragraph{Distributed vs. Federated AIM.} Table \ref{tab:overhead} presents the overhead of DistAIM compared to AugFLAIM (Private) including average client throughput (sent and received communication) across protocols. We set $T \in [1,200]$ that achieves lowest workload error. Observe on Adult, DistAIM requires twice as many rounds to achieve optimal error and results in a large (1300$\times$) increase in client throughput compared to AugFLAIM. However, this results in $2\times$ lower workload error and an $11\%$ improvement in NLL. \review{This highlights one of the chief advantages of FLAIM, where, for a small loss in utility, we can obtain much lower overheads. 
Furthermore, while a $2\times$ gap in workload error seems significant, we refer back to Figure~\ref{fig:vary_s_eps10}, which shows the resulting classifier has AUC that is practical for downstream tasks.} We note the overhead of DistAIM is significantly larger than FLAIM when queries in the workload have large cardinality (e.g., on Adult and Magic). \review{Datasets with much smaller feature domains still have communication overhead but it is not as significant (e.g., Covtype which has many binary features)}.

\section{Related Work}

Synthetic data has gained substantial traction due to its potential to mitigate privacy concerns and address limitations for sharing real-world data. Many generative deep learning approaches exist including GANs \cite{goodfellow2020generative}, VAEs \cite{wan2017variational} and diffusion models \cite{ho2020denoising}. Recent work has extended these synthetic data generators (SDGs) to satisfy central differential privacy (DP) \cite{torkzadehmahani2019dp, weggenmann2022dp, lyu2023differentially} and whilst results are promising for image data, performance on tabular data remains limited. Only a few generative tabular approaches exist including that of CTGAN \cite{xu2019modeling, fang2022dp}.
However, recent work has shown that private tabular approaches like DP-CTGAN often fail to provide good utility when compared to simpler models \cite{tao2021benchmarking, ganev2023understanding}. Indeed, in the central setting of DP many successful methods are based on graphical models such as PrivBayes \cite{zhang2017privbayes}, PrivSyn \cite{zhang2021privsyn}, PGM \cite{mckenna2019graphical} and AIM \cite{mckenna2022aim}. Recent work has shown the class of iterative methods \cite{liu2021iterative} are SOTA on tabular data and we choose to focus on one of these methods, AIM, in our work.
Meanwhile, research into SDGs in the federated setting remains limited. Recent federated SDGs are focused on image data such as MD-GAN \cite{hardy2019md}, FedGAN \cite{rasouli2020fedgan} and FedVAE \cite{yang2023fedvae}. %
We do not compare with these in our work as they do not support tabular data or DP.
The closest work to ours is that of Pereira et al. \cite{pereira2022secure} who propose a distributed DP version of MWEM \cite{hardt2012simple} using secure multiparty computation (SMC) to distribute noise generation across computing servers. This approach has two main drawbacks: it assumes all clients are available to secret-share their workload answers and as it is based on MWEM, obtains subpar utility. Our work is motivated to extend their approach to AIM and to study an alternative and more natural federation of these methods. We also note the concurrent work of Pentyala et al. \cite{pentyala2024caps} which extends \cite{pereira2022secure} to work with AIM via SMC.

\section{Conclusion}
\label{sec:conc}
Overall, we have shown that naively federating AIM under the challenges of FL causes a large decrease in utility when compared to the SMC-based DistAIM. 
\review{To counteract this,} we propose AugFLAIM (Private), which augments local decisions with a proxy for heterogeneity and obtains utility close to DistAIM while lowering overheads. In the future, we plan to extend our approaches to support user-level DP where clients hold multiple data items related to the same individual.

\begin{acks}
Work performed at Warwick University is supported by the UKRI Engineering and Physical Sciences Research Council (EPSRC) under grant EP/W523793/1; the UKRI Prosperity Partnership Scheme (FAIR) under EPSRC grant EP/V056883/1; and the UK NCSC Academic Centre of Excellence in Cybersecurity Research (ACE-CSR).
\end{acks}

\bibliographystyle{ACM-Reference-Format}
\bibliography{refs}

\appendix

\newpage
\section{Algorithm Details}
\subsection{AIM}
\label{appendix:aim}
\begin{algorithm}[t]
\caption{AIM \citep{mckenna2022aim}}\label{alg:aim}
\begin{algorithmic}[1]
\Input Dataset $D \in \R^{N \times d}$, workload $Q$, privacy parameters $\eps, \delta$, Maximum model size $S$
\Output Synthetic dataset $\hat D$
\State Initialise the zCDP budget $\rho_\text{total} \label{alg:aim:1} \leftarrow \rho(\eps, \delta)$ via Lemma \ref{lemma:cdp}
\State Set $\sigma^2_0 \leftarrow \frac{16d}{0.9 \cdot \rho_\text{total}}, \eps_0 \leftarrow \sqrt{8 \cdot 0.1 \cdot \rho_\text{total}/16d}$ \label{alg:aim:budget}
\State Set $t \leftarrow 0$ and $Q \leftarrow \operatorname{Completion}(Q)$
\State For each marginal $q \in Q$ initialise weights \label{alg:aim:4} %
$w_q := \sum_{r \in Q} | q \cap r |$  \label{alg:aim:weight}
\While{$\rho_\text{used} < \rho_\text{total}$}
    \State $t \leftarrow t+1$
    \If{$t = 0$}\Comment{Initialise with one-way marginals} \label{alg:aim:init}
        \State Filter current rounds workload as \label{alg:aim:init_filter}
        \begin{align*}
            Q_0 \leftarrow \{ q \in Q : |q| = 1 \}
        \end{align*} 
        \State Measure $\tilde{M}_{q}(D) \leftarrow M_{q}(D) + N(0, \sigma_0^2I), \forall q \in Q_0$ and use PGM to estimate $\hat D_0$
    \State $\rho_\text{used} \leftarrow \rho_\text{used} + \frac{d}{2\sigma_t^2}$ \label{alg:aim:init_end}
    \Else
    \State Filter the workload \label{alg:aim:filter}
    \begin{align*}
        Q_t \leftarrow \{q \in Q: \operatorname{ModelSize}(\hat D_{t-1}, q) \leq \frac{\rho_\text{used}}{\rho_\text{total}} \cdot S \}
    \end{align*}
    \State \textbf{Select} $q_t \in Q_t$ using the exponential mechanism with parameter $\eps_t$ and utility function \label{alg:aim:select}
    \begin{align*}
        \review{u(q;D)} \leftarrow w_q \cdot \left(\|M_q(D) - M_q(\hat D_{t-1})\|_1 - \sqrt{\frac{2}{\pi}} \cdot \sigma_t \cdot n_q\right)
    \end{align*}
    \State \textbf{Measure} the chosen marginal $q_t$ with the Gaussian mechanism i.e., \label{alg:aim:measure}
    \begin{align*}
        \tilde{M}_{q_t}(D) \leftarrow M_{q_t}(D) + \mathcal{N}(0, \sigma^2_t I)
    \end{align*}
    \State \textbf{Estimate} the new model via PGM \citep{mckenna2019graphical} \label{alg:aim:estimate}
    \begin{align*}
        \hat D_t \leftarrow \argmin_{p \in \mathcal{S}} \sum_{i=1}^t \frac{1}{\sigma_i} \| M_{q_i}(p) - \tilde{M}_{q_i}(D)\|_2
    \end{align*}
    \State $\rho_\text{used} \leftarrow \rho_\text{used} + \frac{\eps_t^2}{8} + \frac{1}{2\sigma_t^2}$
    \EndIf
    \If{$\|M_{q_t}(\hat D_t) - M_{q_t}(\hat D_{t-1})\|_1 \leq \sqrt{2/\pi} \cdot \sigma_t \cdot n_{q_t}$} \Comment{Budget Annealing} \label{alg:aim:anneal}
        \State Set $\sigma_{t+1} \leftarrow \sigma_{t}/2, \eps_{t+1} \leftarrow 2 \cdot \eps_{t}$
    \EndIf
    \If{$(\rho_\text{total} - \rho_\text{used}) \leq 2(1/2\sigma^2_{t+1} + \frac{1}{8}\eps^2_{t+1})$} \Comment{Final round}
        \State Set $\sigma_{t+1}^2 \leftarrow {1/(2 \cdot 0.9 \cdot (\rho_\text{total} - \rho_\text{used})}, \eps_{t+1} \leftarrow \sqrt{8 \cdot 0.1 \cdot (\rho_\text{total}-\rho_\text{used})}$ \label{alg:aim:final}
    \EndIf
\EndWhile
\State \Return $\hat D_t$
\end{algorithmic}
\end{algorithm}

The current SOTA method, and the core of our federated algorithms is AIM, introduced by \citet{mckenna2022aim}. AIM extends the main ideas of MWEM \citep{hardt2012simple} but augments the algorithm with an improved utility score function, a graphical model-based inference approach (via Private-PGM) and more efficient privacy accounting with zero-Concentrated Differential Privacy (zCDP). The full details of AIM are outlined in Algorithm \ref{alg:aim}. We refer to this algorithm as `Central AIM', to distinguish it from the distributed and federated versions we consider in the main body of the paper. 
It is important to highlight the following details:
\begin{itemize}
    \item \textbf{zCDP Budget Initialisation:} In central AIM, the number of global rounds $T$ is set adaptively via budget annealing. To begin, $T := 16 \cdot d$ where $d$ is the number of features. This is the maximum number of rounds that will occur in the case where the annealing condition is never triggered. This initialisation occurs in Line \ref{alg:aim:budget}.
    \item \textbf{Workload Filtering:} The provided workload of queries, $Q$, is extended by forming the completion of $Q$. That is to say, all lower order marginals contained within any $q \in Q$ are also added to the workload. Furthermore, for the first round the workload is filtered to contain only 1-way marginals to initialise the model. This occurs in Line \ref{alg:aim:init_filter}. In subsequent rounds, the workload is filtered to remove any queries that would force the model to grow beyond a predetermined maximum size $S$. This occurs at Line \ref{alg:aim:filter}.
    \item \textbf{Weighted Workload:} Each marginal $q \in Q$ is assigned a weight via $w_q= \sum_{r \in Q} |q \cap r|$. Thus, marginals that have high overlap with other queries in the workload are more likely to be chosen. This is computed in Line \ref{alg:aim:weight}.
    \item \textbf{Model Initialisation}: Instead of initialising the synthetic distribution uniformly over the dataset domain, the synthetic model is initialised by measuring each 1-way marginal in the workload $W$ and using PGM to estimate the initial model. This corresponds to measuring each feature's distribution once before AIM begins and occurs in Lines \ref{alg:aim:init}-\ref{alg:aim:init_end}.
    \item \textbf{Query Selection:} A marginal query is selected via the exponential mechanism with utility scores that compare the trade-off between the current error and the expected error when measured under Gaussian noise. The utility scores and selection step occur at Line \ref{alg:aim:select}.
    \item \textbf{Query Measurement:} Once a query has been chosen, it is measured under the Gaussian mechanism. This occurs at Line \ref{alg:aim:measure}. 
    \item \textbf{PGM model estimation:} The current PGM model is updated by adding the newly measured query to the set of previous measurements. The PGM model parameters are then updated by a form of mirror descent for a number of iterations. The precise details of PGM can be found in \cite{mckenna2019graphical}. This occurs at Line \ref{alg:aim:estimate}.
    \item \textbf{Budget Annealing:} At the end of every round, the difference between the measured query of the new model and that of the previous model is taken. If this change is smaller than the expected error under Gaussian noise, the noise parameters are annealed by halving the amount of noise. This occurs at Line \ref{alg:aim:anneal}. If after this annealing there is only a small amount of remaining privacy budget left, the noise parameters can instead be calibrated to perform one final round before finishing. This occurs at Line~\ref{alg:aim:final}.
\end{itemize}

\subsection{DistAIM}
\label{appendix:distaim}

\begin{algorithm}
\caption{DistAIM\label{alg:daim}}
\begin{algorithmic}[1]
\Input Participants $P_1, \dots P_k$ with local datasets $D_1, \dots, D_k$, privacy parameters $(\eps,\delta)$
\State Initialise AIM parameters as in Lines \ref{alg:aim:1}-\ref{alg:aim:4} of Algorithm \ref{alg:aim}
\For{each round $t$}
    \State Sample participants $P_t \subseteq [k]$ with probability $p$ and remove those who have already participated
    \State For each $k \in P_t$ who have not participated before, secret-share the workload answers $\{\llbracket M_q(D_k) \rrbracket : q \in W\}$ to the compute servers \citep{araki2016high}
    \State \textbf{Aggregate:} The compute servers aggregate shares of the received answers and combine with previously received shares $\llbracket M_q(\tilde{D}_t) \rrbracket := \sum_{i=1}^{t-1}\sum_{k \in P_i}\llbracket M_q(D_k) \rrbracket + \sum_{k \in P_t} \llbracket M_q(D_k) \rrbracket$
    \State \textbf{Select:} Compute servers select $q_t \in Q$ using the exponential mechanism over secret shares $\llbracket M_q(\tilde{D}_t) \rrbracket$ via Protocol 2 in \cite{pereira2022secure} with AIM utility scores
    \begin{align*}
        u(q;D_t) := w_q \cdot (\|M_q(\tilde D_t) - M_q(\hat D_{t-1})\|_1 - \sqrt{\frac{2}{\pi}} \cdot \sigma_t \cdot n_q)
    \end{align*}
    \State \textbf{Measure:} $q_t$ is measured using $\llbracket M_{q_t}(\tilde{D}_t) \rrbracket$ under a variation of Protocol 3 in \cite{pereira2022secure}, replacing Laplace noise with Gaussian to produce $\tilde M_{q_t}(\tilde D_t)$
    \State \textbf{Estimate} the new model via PGM using the received noisy measurements e.g.
    \begin{align*}
        \hat D_t \leftarrow \argmin_{p \in \mathcal{S}} \sum_{i=1}^t \frac{1}{\sigma_i} \| M_{q_i}(p) - \tilde M_{q_i}(\tilde D_i)\|_2
    \end{align*}
\EndFor
\end{algorithmic}
\end{algorithm}

We describe in full detail the DistAIM algorithm introduced in Section \ref{sec:distaim} and outlined in Algorithm \ref{alg:daim}. The algorithm can be seen as an adaptation of \cite{pereira2022secure} who propose a secure multi-party computation (SMC) approach for distributing MWEM. The key differences are that we replace MWEM with AIM and consider a distributed setting where not all participants are available at any particular round. The approach relies on participants secret-sharing their query answers to compute servers who then perform a number of SMC operations over these shares to train the model. The resulting algorithm is identical to AIM in outline but has a few subtle differences:
\begin{itemize}
    \item \textbf{Secret Sharing:} Participants must secret-share the required quantities to train AIM. In \cite{pereira2022secure}, it is assumed that the full workload answers $\{\llbracket M_q(D) \rrbracket : q \in Q\}$ have already been secret-shared between a number of compute servers. In DistAIM, we assume that clients sampled to participate at a particular round contribute their secret-shared workload answers $\{\llbracket M_q(D_k) \rrbracket : q \in Q\}$ which are aggregated with the shares of current and past participants from previous rounds. Thus, as the number of global rounds $T$ increases, the secret-shared answers approach that of the central dataset. We assume the same SMC framework as \cite{pereira2022secure} which is a 3-party scheme based on \cite{araki2016high}.
     
    \item \textbf{Client participation:} At each round only a subset of the participants are available to join the AIM round. In expectation $pK$ clients will contribute their local marginals $\llbracket M_q(D_k) \rrbracket$ in the form of secret-shares. Compared to the central setting, DistAIM incurs additional error due to this subsampling. 

    \item \textbf{Select step:} One key obstacle in extending AIM to a distributed setting is the exponential mechanism. Since each client holds a local dataset $D_k$, they cannot share their data with the central server. Instead the quality functions $u(q;D)$ must be computed in a distributed manner between the compute servers who hold shares of the workload answers.

    \item \textbf{Measure step:} Once the marginal $q_t$ has been selected by a secure exponential mechanism, it must be measured. As \cite{pereira2022secure} utilise MWEM, they measure queries under Laplace noise which can be easily generated in an SMC setting. AIM instead uses Gaussian noise and this is also what we use in DistAIM. In practice, one can also implement this under SMC e.g., using the Box-Muller method.
\end{itemize}

\subsection{FLAIM: Privacy Guarantees}
\label{appendix:flaim}
In this section, we present and prove the privacy guarantees of the FLAIM approach. 
For completeness, we provide additional definitions and results, starting with the definition of $(\eps,\delta)$-Differential Privacy.
\begin{definition}[Differential Privacy~\citep{dwork2014foundations}]
A randomised algorithm $\mathcal{M}\colon \mathcal{D} \rightarrow \mathcal{R}$ satisfies $(\varepsilon, \delta)$-differential privacy if for any two adjacent datasets $D, D^\prime \in \mathcal{D}$ and any subset of outputs $S \subseteq \mathcal{R}$, $$\prob(\mathcal{M}(D) \in S) \leq e^\eps \prob(\mathcal{M}(D^\prime) \in S) + \delta.$$
\end{definition}
While we work using the more convenient formulation of $\rho$-zCDP (Definition \ref{def:zcdp}), it is common to translate this guarantee to the more interpretable $(\eps,\delta)$-DP setting via the following lemma.

\begin{lemma}[zCDP to DP \citep{canonne2020discrete}]
\label{lemma:cdp}
If a mechanism $\mathcal{M}$ satisfies $\rho$-zCDP then it satisfies $(\eps,\delta)$-DP for all $\eps > 0$ with $$\delta = \min_{\alpha > 1} \frac{\exp((\alpha-1)(\alpha\rho - \eps))}{\alpha - 1} \left(1-\frac{1}{\alpha}\right)^\alpha$$
\end{lemma}

We restate the privacy guarantees of FLAIM and its variations.

\begin{lemma}[Lemma~\ref{lem:dp} restated]
    For any number of global rounds $T$ and local rounds $s$, FLAIM satisfies $(\eps,\delta)$-DP , under Gaussian budget allocation $r \in (0,1)$ by computing $\rho$ according to Lemma \ref{lemma:cdp}, and setting
    \begin{align*}
        \sigma_t &= \begin{cases}
          \sqrt{\frac{Ts + d}{2\cdot r \cdot \rho}}, & \text{\naive{Naive} or \nonpriv{Oracle}} \\
          \sqrt{\frac{T(s+d)}{2\cdot r \cdot \rho}}, & \text{\flaim{Private}}
        \end{cases}, \; \eps_t = \sqrt{\frac{8 \cdot (1-r) \cdot \rho}{Ts}}
    \end{align*}
    For AugFLAIM methods, the exponential mechanism is applied with sensitivity $\Delta := \max_q 2w_q$.
\end{lemma}

\begin{proof}
    For NaiveAIM, the result follows almost directly from AIM, since $T$ rounds in the latter correspond to $T \cdot s$ in the former. We then apply the existing privacy bounds for AIM. Similarly, for AugFLAIM (Private), the 1-way marginals of every feature are included in the computation, thus increasing the number of measured marginals under Gaussian noise to $T \cdot (s + d)$. In all variations, the exponential mechanism is only applied once for each local round and thus $Ts$ times in total. For AugFLAIM, the augmented utility scores $u(q;D_k)$ lead to a doubling of the sensitivity compared to AIM, since $M_q(D_k)$ is used twice in the utility score and thus $\Delta := 2\cdot \max_q w_q$.
\end{proof}

\begin{table*}[t]
  \caption{Datasets\label{tab:datasets} - Those marked * have been subsampled for computational reasons.}
    \centering
    \begin{tabular}{cccc}
      \toprule
      \textbf{Dataset} & \textbf{\# of training samples} & \textbf{\# of features} & \textbf{\# of classes}\\
      \midrule
      Adult \citep{adult} & 43,598 & 14 & 2 \\
      Credit \citep{credit} & 284,807 & 30 & 2 \\
      Covtype* \citep{misc_covertype_31} & 116,203 & 55 & 7 \\
      Census* \citep{SDV} & 89,786  & 41 & 2 \\
      Intrusion*\cite{intrusion} & 197,608 & 40 & 5 \\
      Marketing \citep{misc_bank_marketing_222} & 41,188 & 21 & 2 \\
      Magic \citep{magic} & 17,118 & 11 & 2 \\
      SynthFS (see Appendix \ref{appendix:synthfs}) & 45,000 & 10 & N/A \\
      \bottomrule
  \end{tabular}
  \end{table*}
\section{Experimental Setup}

\subsection{Datasets}
\label{appendix:datasets}
In our experiments we use a range of tabular datasets from the UCI repository \citep{Dua:2019} and others available directly from the Synthetic Data Vault (SDV) package \citep{SDV}. Additionally, we use one synthetic dataset that we construct ourselves. A summary of all datasets in terms of the number of training samples, features and classes is detailed in Table \ref{tab:datasets}. All datasets are split into a train and test set with $90\%$ forming the train set. From this, we form clients' local datasets via a partitioning method (see Appendix \ref{appendix:hetero}). In more detail:
\begin{itemize}
    \item \textbf{Adult} --- A census dataset that contains information about adults and their occupations. The goal of the dataset is to predict the binary feature of whether their income is greater than \$50,000. The training set we use contains 43,598 training samples and 14 features.
    \item \textbf{Credit} --- A credit card fraud detection dataset available from Kaggle. The goal is to predict whether a transaction is fraudulent. 
    \item \textbf{Covtype} --- A forest cover type prediction dataset available from the UCI repository. We subsampled the dataset for computational reasons. Our train and test sets were formed from $20\%$ of the original dataset.
    \item \textbf{Census} --- US census dataset available through the synthetic data vault (SDV) package. This dataset was subsampled for computational reasons. Our training and test sets were formed from $30\%$ of the original dataset.
    \item \textbf{Intrusion} --- The DARPA network intrusion detection dataset containing network logs, available through the synthetic data vault (SDV) package. This was subsampled for computational reasons. Our training and test sets were formed from $40\%$ of the original dataset.
    \item \textbf{Marketing} --- A bank marketing dataset available from the UCI repository. The goal is to predict whether a client will subscribe to a term deposit.
    \item \textbf{Magic} --- A dataset on imaging measurements from a telescope. The classification task is to predict whether or not the measurements are signal or background noise. The training set we use contains 17,118 samples and 11 features.
    \item \textbf{SynthFS} --- A synthetic dataset formed from sampling features from a Gaussian distribution with different means. The precise construction is detailed in Appendix \ref{appendix:synthfs}.  In our experiments, the training set contains 45,000 samples with 10 features.
\end{itemize}
All continuous features are binned uniformly between the minimum and maximum which we assume to be public knowledge. We discretize our features with $32$ bins, although experiments varying this size presented no significant change in utility. This follows the pre-processing steps taken by prior work \citep{mckenna2022aim, aydore2021differentially}.

\subsubsection{SynthFS}
\label{appendix:synthfs}
\begin{figure*}[t!]
	\centering
	\subfloat[$\beta = 5$]{%
		\includegraphics[width=0.25\linewidth]{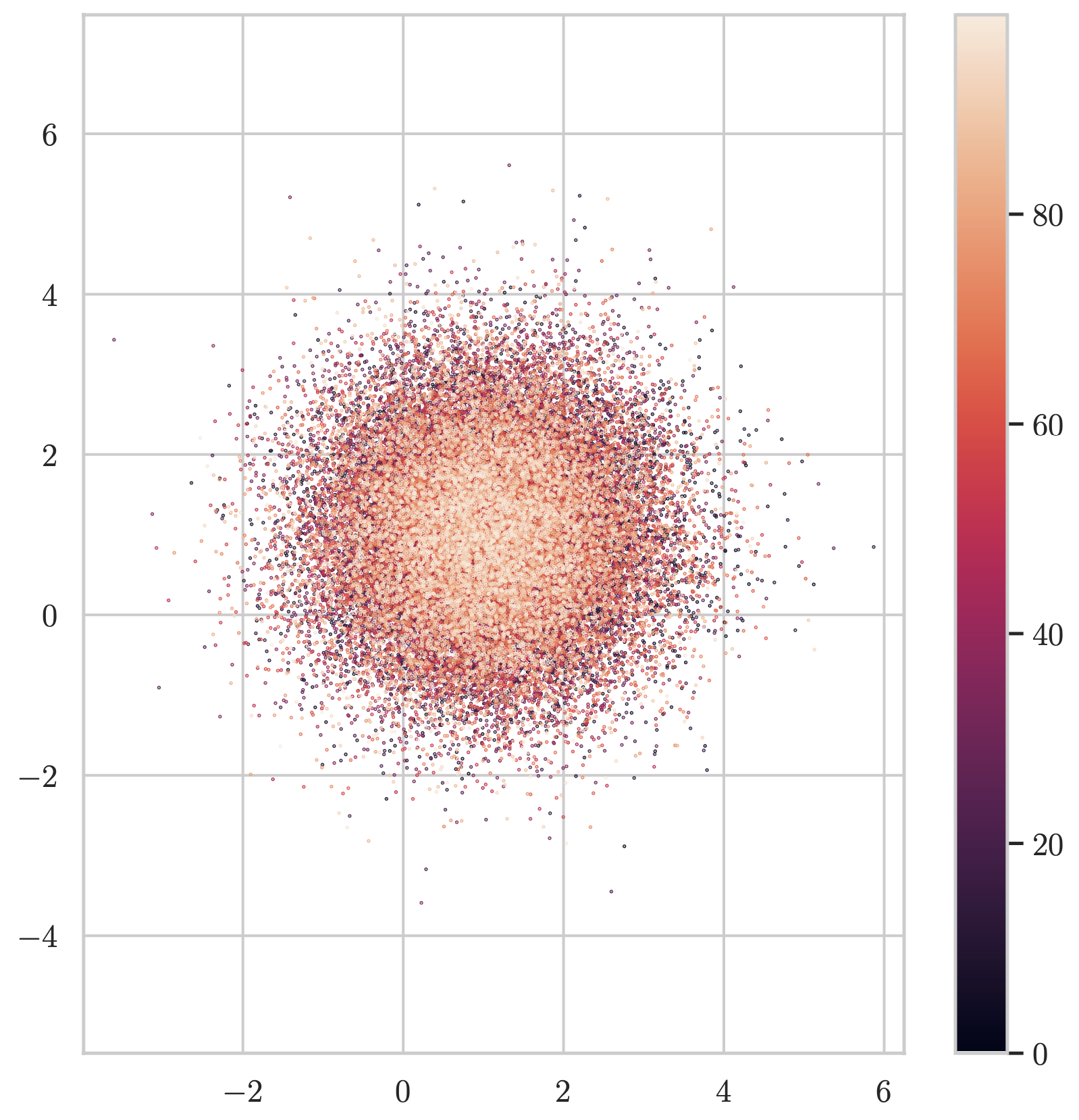}}
	\subfloat[$\beta = 3$]{%
		\includegraphics[width=0.25\linewidth]{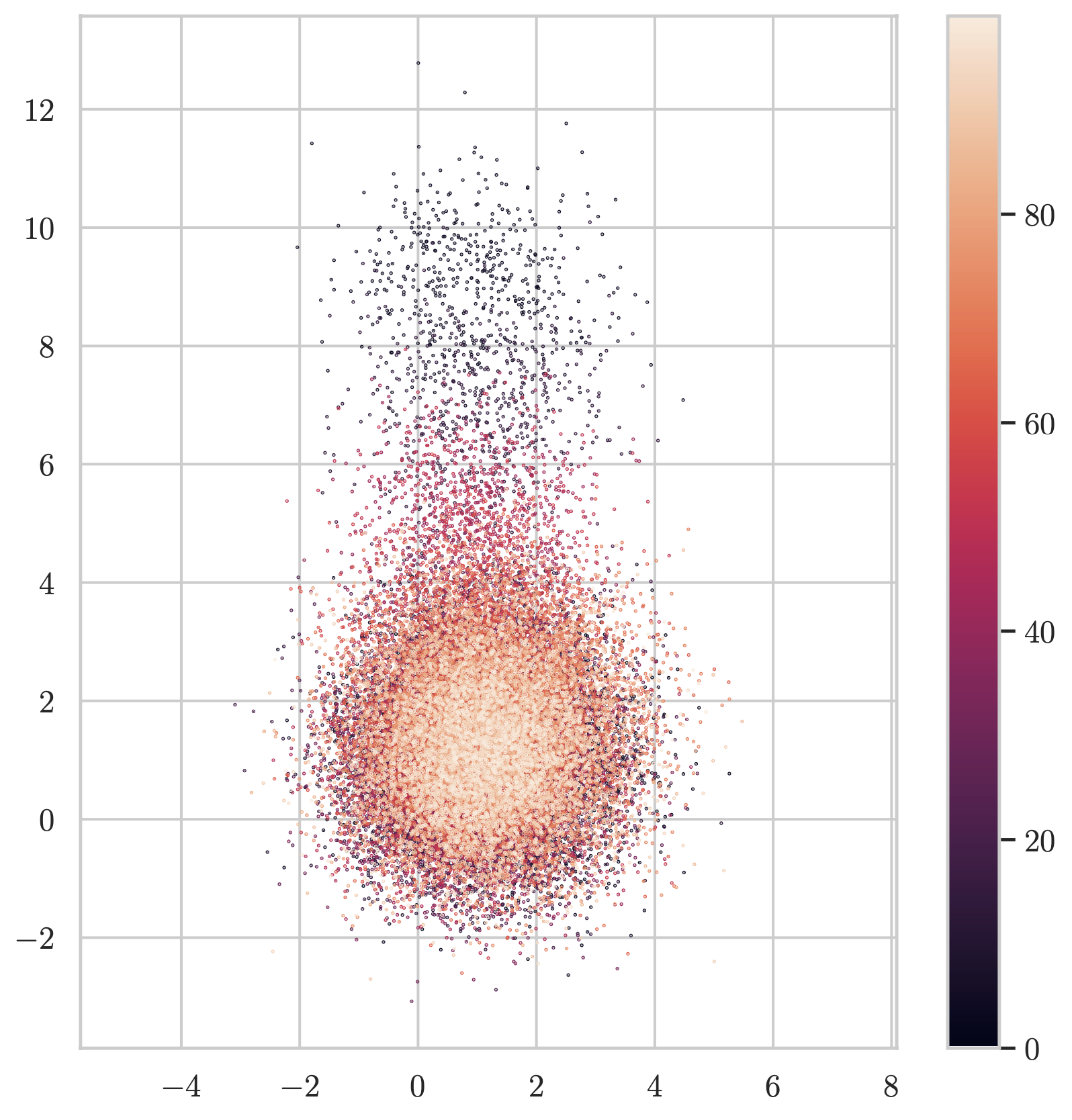}}
	\subfloat[$\beta = 2$]{%
		\includegraphics[width=0.25\linewidth]{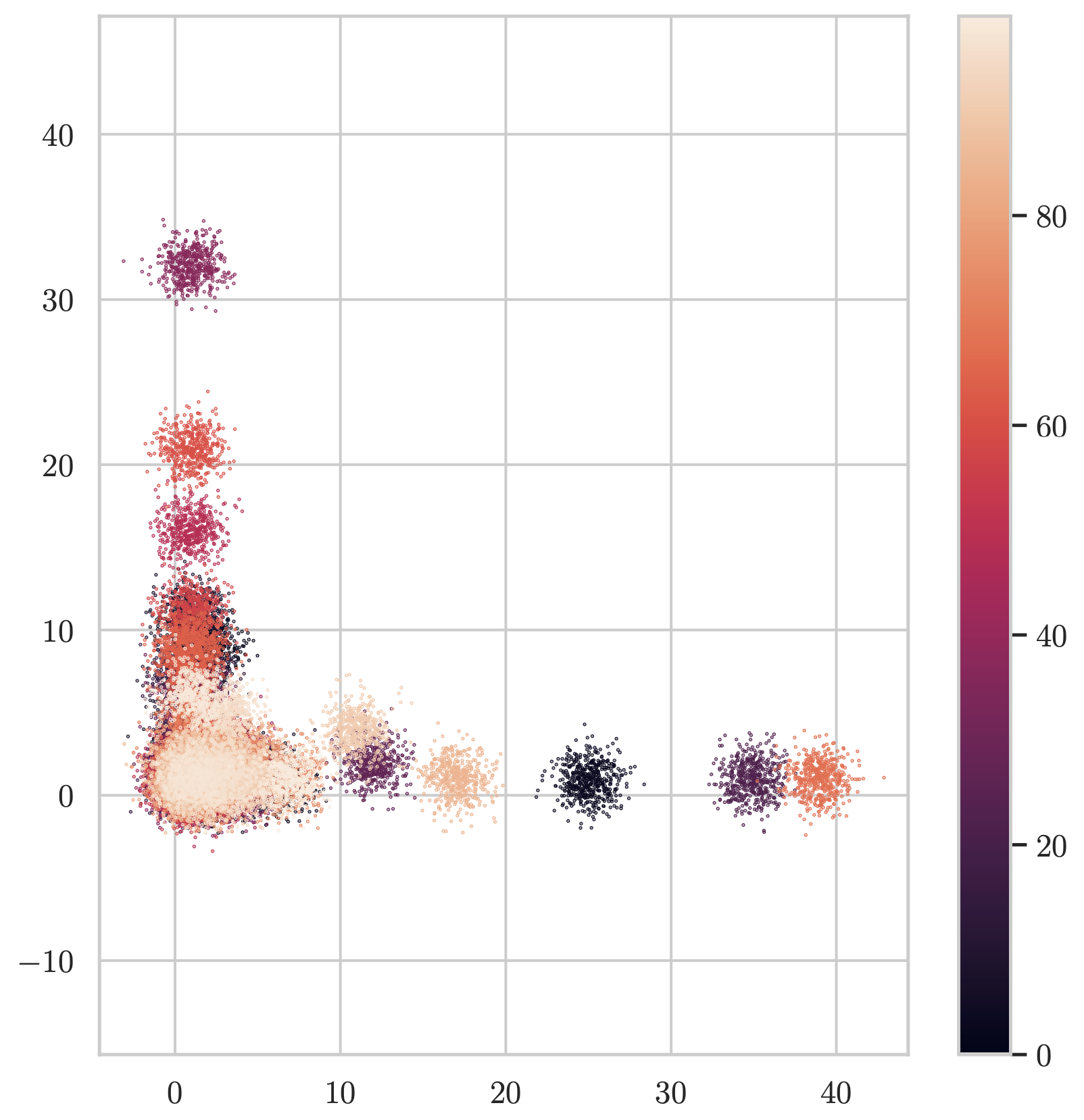}}
	\subfloat[$\beta =1 $]{%
		\includegraphics[width=0.25\linewidth]{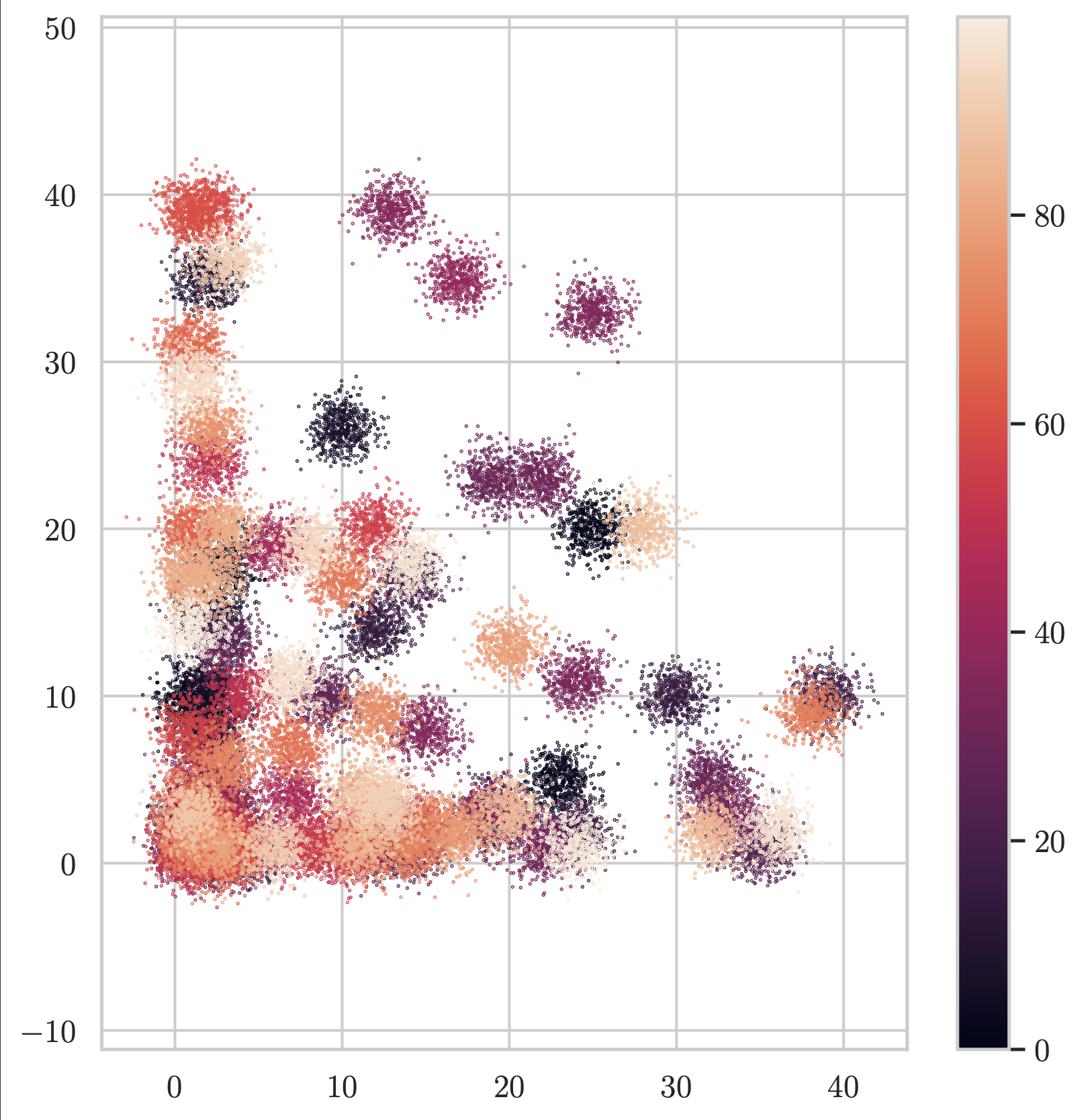}}
	\caption{SynthFS: Synthetic dataset constructed with feature skew, varying $\beta \in \{1,2,3,5\}$}
	\label{fig:synthfs} 
\end{figure*}
In order to simulate feature-skew in an ideal setting for FLAIM, we construct a synthetic dataset that we denote SynthFS.
To create SynthFS, we draw independent features from a Gaussian distribution where the mean is chosen randomly from a Zipfian distribution whose parameter $\beta$ controls the skew.
This is done in the following manner:
\begin{itemize}
    \item For each client $k \in [K]$ and feature $m \in [d]$ sample mean $\mu_m^k \sim \operatorname{Zipf}(\beta, n_\text{zipf})$
    \item For each feature $m \in [d]$, sample $n/K$ examples for client $k$ from $N(\mu_m^k, 1)$
\end{itemize}
In our experiments we set $n = 50,000$ such that for $K=100$ each client is assigned 500 samples. In order to form a test set we sample 10\% from the dataset and assign the rest to clients. We fix $d=10$ and $n_\text{zipf} = 40$ in all constructions. We highlight this process for $\beta \in \{1, 2, 3, 5\}$ in Figure \ref{fig:synthfs}, with $d=2$ features for visualization purposes only. By increasing $\beta$, we decrease the skew of the means being sampled from the Zipf distribution. Hence, for larger $\beta$ values, each client's features are likely to be drawn from the same Gaussian and there is no heterogeneity. Decreasing $\beta$ increases the skew of client means and each feature is likely to be drawn from very different Gaussian distributions, as shown when $\beta = 1$.

\subsection{Heterogeneity: Non-IID Client Partitions}
\label{appendix:hetero}
In order to simulate heterogeneity on our benchmark datasets, we take one of the tabular datasets outlined in Appendix~\ref{appendix:datasets} and form partitions for each client. The aim is to create client datasets that exhibit strong data heterogeneity by varying the number of samples and inducing feature-skew. We do this in two ways:
\begin{itemize}
    \item \textbf{Clustering Approach} --- In the majority of our experiments, we form client partitions via dimensionality reduction using UMAP \citep{mcinnes2018umap}. An example of this process is shown in Figure \ref{fig:cluster} for the Adult dataset. Figure \ref{fig:umap} shows a UMAP embedding of the training dataset in two-dimensions where each client partition (cluster) is highlighted a different color. To form these clusters we simply use $K$-means where $K=100$ is the total number of clients we require. In Figures~\ref{fig:age}-\ref{fig:income}, we display the same embedding but colored based on different feature values for age, hours worked per-week and income $>$ 50k. We observe, for instance, the examples that are largest in age are concentrated around $x=10$ while those who work more hours are concentrated around $y=-7$. Thus clients that have datasets formed from clusters in the area of $(10,-7)$ will have significant feature-skew with a bias towards older adults who work more hours. These features have been picked at random and other features in the dataset have similar skew properties. The embedding is used only to map the original data to clients, and the raw data is used when training AIM models.
    \item \textbf{Label-skew Approach} --- While the clustering approach works well to form non-IID client partitions, there is no simple parameter to vary the heterogeneity of the partitions. In experiments where we wish to vary heterogeneity, we follow the approach outlined by \cite{li2022federated}. For each value the class variable can take, we sample the distribution $p_C \sim \operatorname{Dirichlet}(\beta) \subseteq [0,1]^K$ and assign examples with class value $C$ to the clients using this distribution. This produces client partitions that are skewed via the class variable of the dataset, where a larger $\beta$ decreases the skew and reduces heterogeneity.
\end{itemize}

Table \ref{tab:non_iid} presents the average heterogeneity for a fixed workload of queries across the Adult and Magic dataset with different partition methods for $K=100$ clients. We look at the following methods: IID sampling, clustering approach, label-skew with $\beta=0.1$ (large-skew) and label-skew with $\beta=0.8$ (small-skew). Observe in all cases that our non-IID methods have higher heterogeneity than IID sampling. Specifically, the clustering approach works well to induce heterogeneity and can result in twice as much skew across the workload. Note also that increasing $\beta$ from $0.1$ to $0.8$ decreases average heterogeneity and at $\beta=0.8$, the skew is close to IID sampling. This confirms that simulating client partitions in this way is useful for experiments where we wish to vary heterogeneity, since we can vary $\beta$ accordingly and $\beta \in (0, 1]$ in experiments is well-chosen.

\begin{figure*}[t!]
	\centering
	\subfloat[Client partitions\label{fig:umap}]{%
		\includegraphics[width=0.25\linewidth]{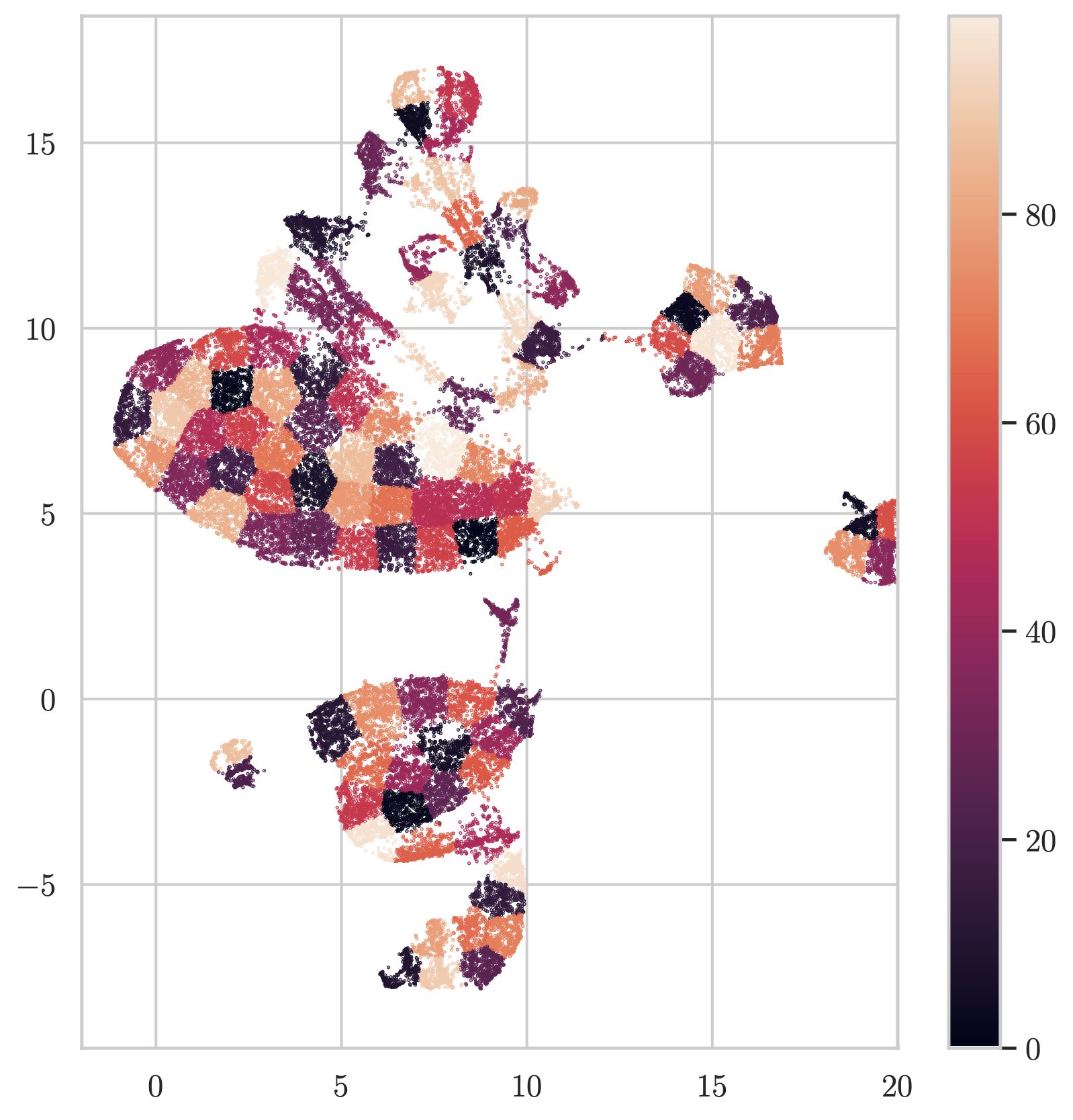}}
	\subfloat[Age\label{fig:age}]{%
		\includegraphics[width=0.25\linewidth]{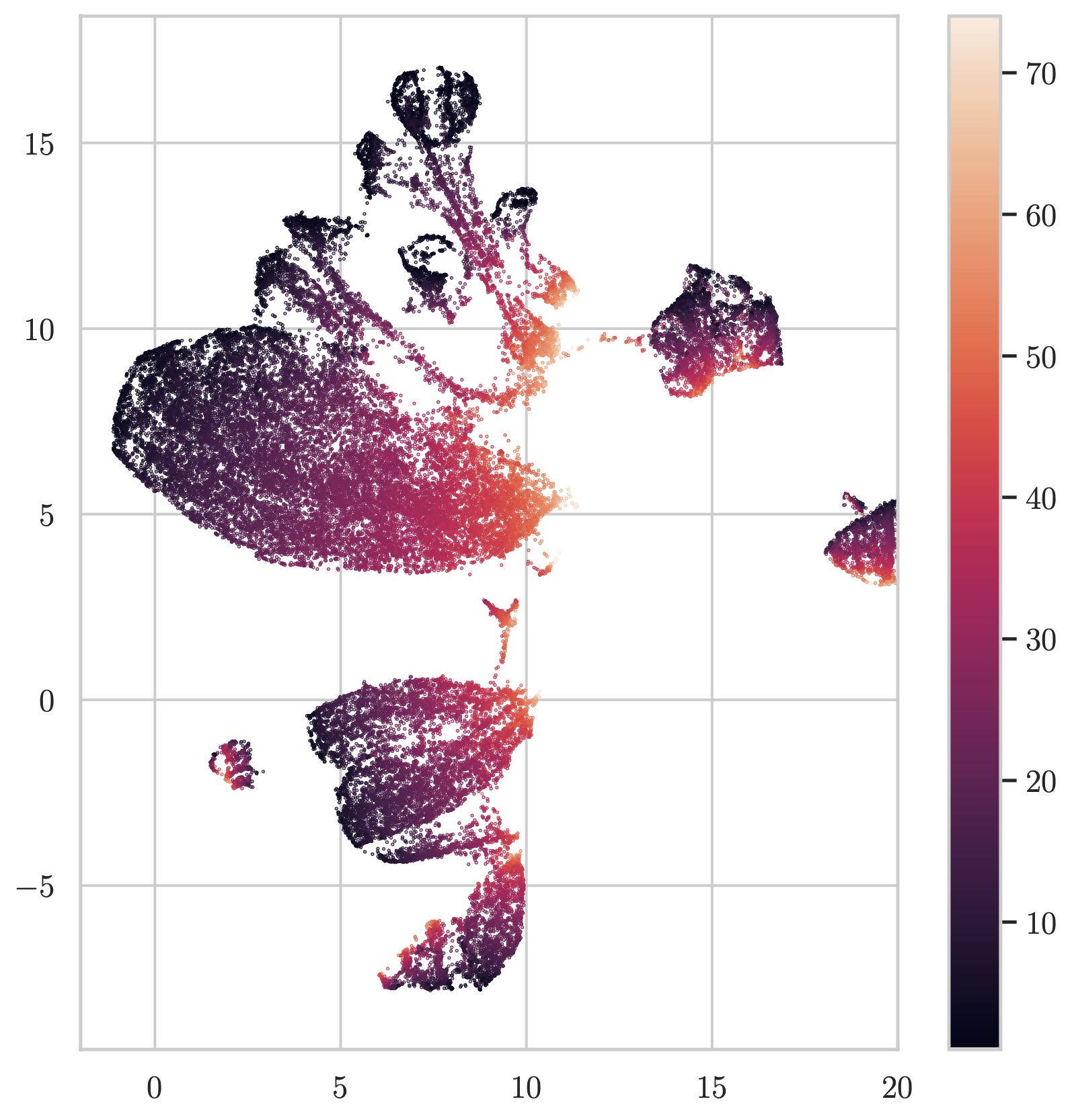}}
	\subfloat[Hours per-week]{%
		\includegraphics[width=0.25\linewidth]{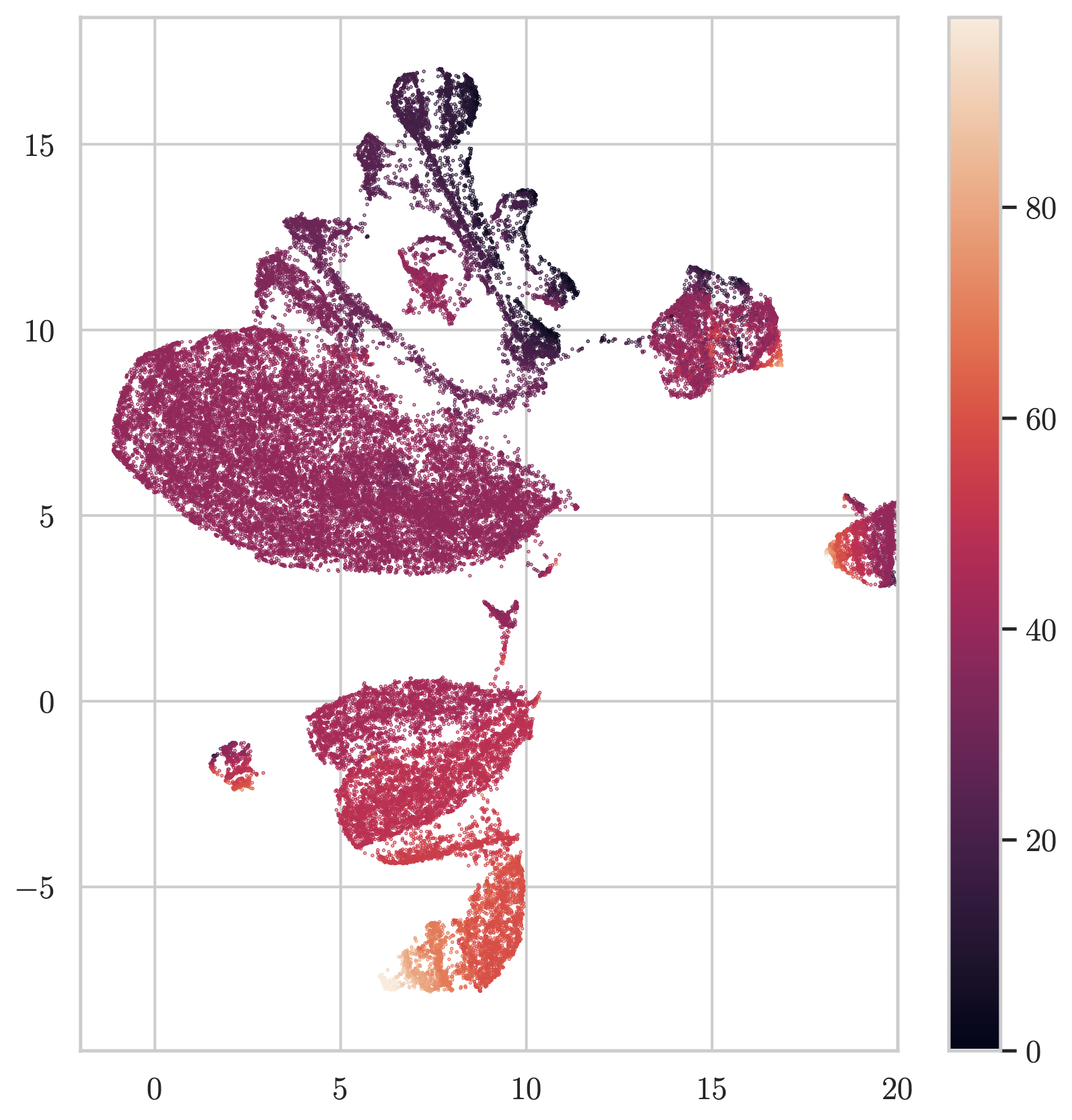}}
	\subfloat[Income > 50K\label{fig:income}]{%
	\includegraphics[width=0.25\linewidth]{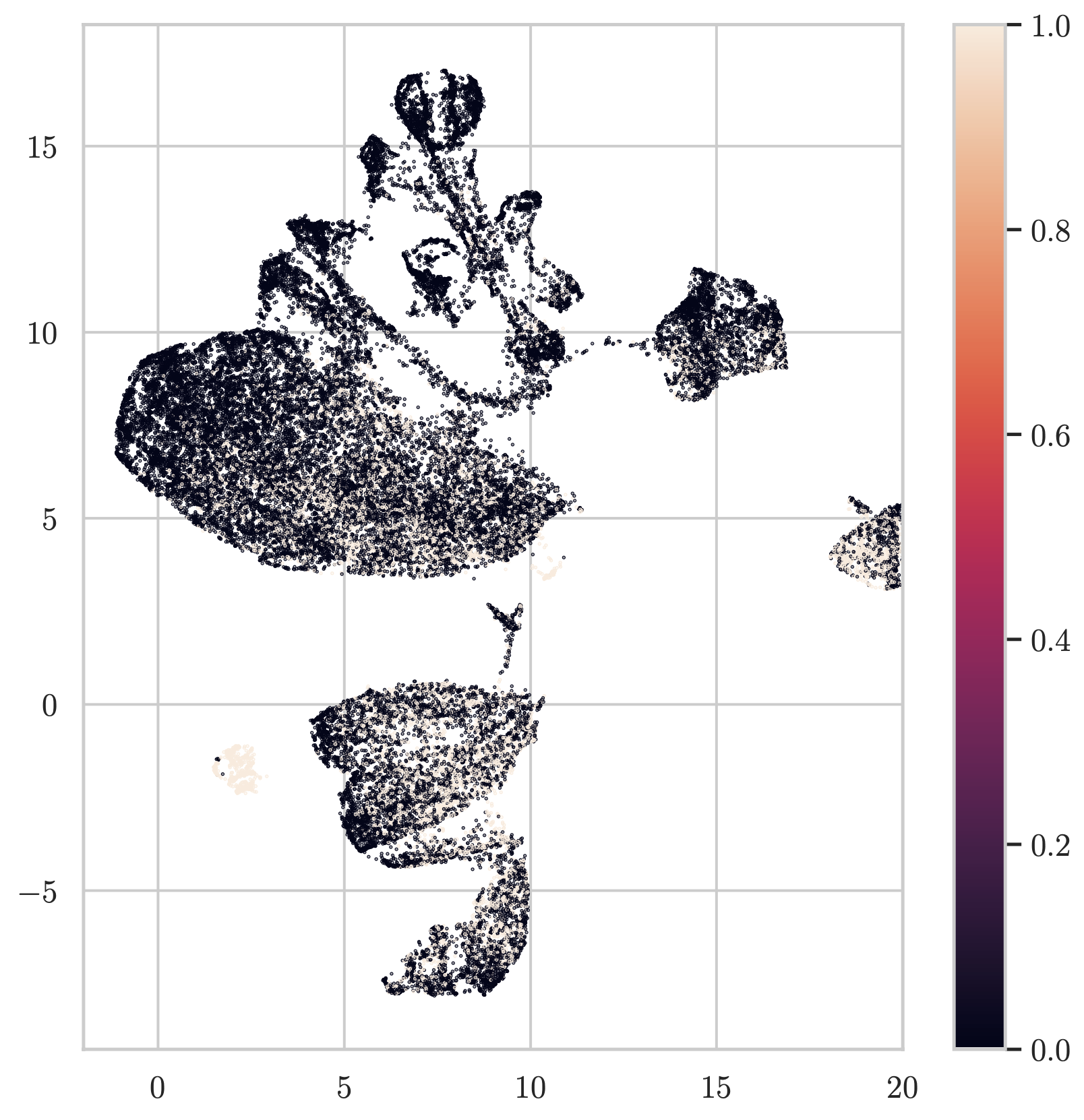}}
	\caption{Clustering approach to form non-IID splits on Adult dataset, $K=100$ clients. All plots show the same embedding formed from UMAP, with Figure \ref{fig:umap} showing each client's local dataset formed by clustering in the embedding space. Figures \ref{fig:age}-\ref{fig:income} show the same embedding but colored based on three features: age, hours worked per-week and income > 50k. The embedding is used only to map examples to clients, and AIM models are trained on the raw data.}
	\label{fig:cluster} 
\end{figure*}
\begin{table}[t!]
	\centering
	\caption{Average heterogeneity over a workload of uniform queries computed as $\frac{1}{K}\sum_k \sum_{q \in Q} \tau_k(q)$ whilst varying different client partition methods with $K=100$ total clients.}
    \small
	\begin{tabular}{ llllll }
		\hline
		\textbf{Dataset / Partition} & IID & Clustering & Label-skew & Label-skew \\ 
        & &  & ($\beta=0.1$) & ($\beta=0.8$)\\
		\hline
		Adult  & 0.241 & 0.525 & 0.531 & 0.332  \\
		Magic  & 0.538 & 0.792 & 0.767 & 0.603 \\
		\hline
	\end{tabular}
	\label{tab:non_iid}
\end{table}

\subsection{Evaluation}
In our experiments we evaluate our methods with three different metrics:

\textbf{1. Average Workload Error.} We mainly evaluate (FL)AIM methods via the average workload error. For a fixed workload of marginal queries $Q$, we measure $\operatorname{Err}(D, \hat D; Q) := \frac{1}{|Q|} \sum_{q \in Q} \|M_{q}(D) - M_{q}(\hat D)\|_1$ where $D := \cup_k D_k$. This can be seen as a type of training error since the models are trained to answer the queries in $Q$.

\textbf{2. Negative Log-likelihood.} An alternative is the (mean) negative log-likelihood of the synthetic dataset sampled from our (FL)AIM models when compared to a heldout test set. This metric can be viewed as a measure of generalisation, since the metric is agnostic to the specific workload chosen.

\textbf{3. Test ROC-AUC.} In some cases we evaluate our models by training a gradient boosted decision tree (GBDT) on the synthetic data it produces. We test the performance of the classifier on a test set and evaluate the ROC-AUC.

\subsection{Experiment Hyperparameters} \label{appendix:hyper}

\subsubsection{CTGAN}
In our baseline comparisons in Section \ref{sec:baselines} we use the DP-CTGAN implementation contained in the synthetic data vault (SDV) package \citep{SDV}. We performed a hyperparameter search over epochs, learning rates and gradient clipping norm. We found training for $20$ epochs, with a gradient norm of $1$, batch size of $128$, discriminator LR of $1e{-3}$ and generator LR of $1e{-5}$ gave best performance. For the federated setting we train the DP-CTGAN using DP-FedSGD implemented via the FLSim framework \cite{flsim}. We found training for $50$ epochs with a local batch size of $128$, clipping norm of $0.5$, server LR of $0.5$ and discriminator/generator LRs of $1e{-4}$ performed best.

\subsubsection{(FL)AIM}
\textbf{PGM Iterations:} The number of PGM iterations determines how many optimisation steps are performed to update the parameters of the graphical model during training. AIM has two parameters, one for the number of training iterations between global rounds of AIM and one for the final number of iterations performed at the end of training. We set this to 100 training iterations and 1000 final iterations. This is notably smaller than the default parameters used in central AIM, but we verified that there is no significant impact on utility.

\noindent
\textbf{Model Initialisation:} We follow the same procedure as in central AIM, where every 1-way marginal is estimated to initialise the model. Instead in our federated settings, we take a random sample of clients and have them estimate the 1-way marginals and initialise the model from these measurements.
    
\noindent
\textbf{Budget Annealing Initialisation:} When using budget annealing, the initial noise is calibrated under a high number of global rounds. In central AIM, initially $T = 16 \cdot d$ results in a large amount of noise until the budget is annealed. We instead set this as $T = 8 \cdot d$ since empirically we have verified that a smaller number of global rounds is better for performance in the federated setting.

\noindent
\textbf{Budget Annealing Condition:} In central AIM, the budget annealing condition compares the previous model estimate with the new model estimate of the current marginal. If the annealing condition is met, the noise parameters are decreased. In the federated setting, it is possible that PGM receives multiple new marginals at a particular round. We employ the same annealing condition, except we anneal the budget if at least one of the marginals received from the last round passes the check.

\section{Further Experiments}
\label{appendix:experiments}

\textbf{Varying $\eps$:} In Figure \ref{fig:appendix:vary_eps}, we vary $\eps$ across our datasets under a clustering partition with $K=100$ clients and $\eps = 1$. These plots replicate Figure \ref{fig:vary_eps} across the other datasets. We observe similar patterns to that of Figure \ref{fig:vary_eps} with NaiveFLAIM performing worst across all settings, and our AugFLAIM methods helping correct this to closely match the performance of DistAIM and in some cases even exceed it with lower workload error. There are however some consistent differences when compared to the Adult datasets. For example, on the Magic dataset, AugFLAIM (Private) performance comes very close to DistAIM but there is a consistent gap in workload error. This is in contrast to the Adult dataset where AugFLAIM (Private) shows a more marked improvement.

\begin{table*}[t]
    \caption{\textbf{Budget annealing ranking} across workload error and negative log-likelihood. Ranks are averaged across each dataset, with each method repeated 10 times. $T$ is set adaptively via annealing.}
    \label{tab:annealing_ranks}
    \centering
    \begin{tabular}{lllll}
    \toprule
    \textbf{Method / $\eps =$} &            1 &              2 &              3 &             5 \\
    \midrule
    \naive{NaiveFLAIM}               &  4.65 / 4.75 &    4.875 / 4.9 &   4.975 / 4.95 &     5.0 / 5.0 \\
    \nonpriv{AugFLAIM (Oracle)} &    3.6 / 3.3 &   3.85 / 3.525 &    3.9 / 3.775 &   3.925 / 3.6 \\
    \flaim{AugFLAIM (Private)}     &  3.75 / 3.45 &    3.25 / 3.15 &  3.125 / 3.125 &   3.05 / 3.25 \\
    DistAIM             &    2.0 / 2.5 &  2.025 / 2.425 &     2.0 / 2.15 &  2.025 / 2.15 \\
    \midrule 
    AIM                 &    1.0 / 1.0 &      1.0 / 1.0 &      1.0 / 1.0 &     1.0 / 1.0 \\
    \bottomrule
    \end{tabular}
\end{table*}

\begin{table}[t]
    \caption{\textbf{Total overhead of DistAIM vs AugFLAIM (Private) measure via the average client throughput (total received and sent communication) for $\eps=1$ and $T=4,32,96$.}}
    \label{tab:appendix:overhead}
    \centering
    \begin{tabular}{lllll}
    \toprule
    \textbf{Method} & $T=4$ & $T=32$ & $T=96$ \\
    \midrule
    Adult & $3990$x & $1223$x & $410$x \\ 
    Magic & $2467$x & $746$x & $267$x \\
    Intrusion & $2313$x & $630$x & $233$x \\ 
    Marketing & $603$x & $174$x & $66$x\\ 
    Covtype & $199$x & $57$x & $15$x \\ 
    Credit & $2363$x & $714$x & $220$x \\ 
    Census & $832$x & $221$x & $77$x \\ 
    \bottomrule
    \end{tabular}
\end{table}

\noindent
\textbf{Varying $T$:} In Figure \ref{fig:appendix:vary_t}, we vary the global rounds $T$ while fixing $\eps=1$ and $K=100$ clients under a clustering partition. This replicates Figure \ref{fig:vary_t} but over the other datasets. Across all figures we plot dashed lines to show the mean workload error under the setting where $T$ is chosen adaptively via budget annealing. On datasets other than Adult, we observe more clearly the choice of $T$ is very significant to the performance of AugFLAIM (Private) and choosing $T > 30$ can result in a large increase in workload error for some datasets (marketing, covtype, intrusion, census). In contrast, increasing $T$ for DistAIM often gives an improvement to the workload error. Recall, DistAIM has participants secret-share their workload answers and these are aggregated over a number of rounds. Hence, as $T$ increases the workload answers DistAIM receives approaches that of the central dataset. 
For budget annealing, on 3 of the 6 datasets, AugFLAIM (Private) has improved error over NaiveFLAIM but does not always result in performance that matches DistAIM. Instead, it is recommended to choose $T \in [5,30]$ which has consistently good performance across all of the datasets.

\noindent
\textbf{Varying $p$:} In Figure \ref{fig:appendix:vary_p} we vary the %
participation rate $p$ while fixing $\eps=1, T=10$ and $K=100$ clients under a clustering partition. This replicates Figure \ref{fig:vary_p} but across the other datasets. We observe similar patterns as we did on Adult. DistAIM approaches the utility of central AIM as $p$ increases. We note that for NaiveFLAIM, often the worklaod error does not increase as $p$ increases. Again, as in Figure \ref{fig:vary_p} the likely cause for this is local skew. For AugFLAIM the workload error decreases as $p$ increases and often matches that of DistAIM, except on Magic and Marketing where it stabilises for $p > 0.3$. Generally, when $p$ is large, DistAIM is preferable but we note this does not correspond to a practical federated setting where sampling rates are typically much smaller ($p < 0.1$) and in this regime DistAIM and AugFLAIM performance is matched.

\noindent
\textbf{Varying $\beta$:} In Figure \ref{fig:appendix:vary_beta}, we vary the label-skew partition across datasets via the parameter $\beta$. A larger $\beta$ results in less label-skew and so less heterogeneity. These experiments replicate that of Figure~\ref{fig:vary_beta}. As before, we clearly observe that NaiveFLAIM is subject to poor performance and that this is particularly the case when there is high skew (small $\beta$) in participants' datasets. We can see that the AugFLAIM methods help to stabilise performance and when skew is large $(\beta < 0.1)$ can help match DistAIM across the datasets.

\noindent
\textbf{Local updates:} In Figures \ref{fig:appendix_vary_s1} and \ref{fig:appendix_vary_s2} we vary the local updates $s \in \{1,4\}$ while fixing $\eps=1, T=10$ and $K=100$.  This replicates Figure \ref{fig:vary_s_eps1} and \ref{fig:vary_s_eps10} but across the other datasets. When using $s=4$ local rounds, the workload error across methods often increases for NaiveFLAIM and AugFLAIM methods. However, when looking at the test AUC performance, taking $s=4$ local updates often gives better AUC performance than $s=1$ on the Census, Magic and Credit datasets. This results in AUC that is closer to that of DistAIM than the other FLAIM methods.

\noindent
\textbf{Budget Annealing:} In Table \ref{tab:annealing_ranks} we present the average rank of methods across all datasets. We rank based on two metrics: workload error and negative log-likelihood. The number of rounds $T$ is set adaptively via budget annealing. We vary $\eps \in \{1,2,3,4,5\}$ with the goal of understanding how annealing affects utility across methods. DistAIM achieves the best rank across all settings when using budget annealing, only beaten by central AIM. When $\eps$ is small, AugFLAIM (Oracle) achieves a better average ranking across both metrics when compared to AugFLAIM (Private). However, as $\eps$ increases, AugFLAIM (Private) achieves better rank, only beaten by DistAIM. AugFLAIM (Private) can achieve better performance by choosing $T$ reasonably small ($T < 30$) as previously mentioned.

\noindent
\textbf{DistAIM vs. FLAIM Communication:} In Table \ref{tab:appendix:overhead}, we present the overhead of DistAIM vs. AugFLAIM (Private) in terms of the average client throughput (total sent and received communication) for $T=4,32,96$. In DistAIM, the amount of communication a client sends is constant no matter the value of $T$, since they only send secret-shared answers once (when they participate in a round). In the case where the total dimension of a workload is large, the gap in client throughput between AugFLAIM and DistAIM is also large. For example on Adult, clients must send $140$Mb in shares whereas AugFLAIM is an order of magnitude smaller. Sending 140Mb of shares may not seem prohibitive but this size quickly scales in the dimensions of features and in practice could be large e.g., on datasets with many continuous features discretized to a reasonable number of bins. Note that if $T$ increases to be very large, eventually AugFLAIM would meet or exceed the communication of DistAIM. However, this would not occur in practice since the best utility is obtained when $T$ is small (e.g., $T < 100$), as observed in Figure \ref{fig:vary_t}.

AugFLAIM (Private) communication is mostly consistent across each dataset for a particular value of $T$ e.g., at $T=4$ average client throughput is $0.035$MB up to $0.5$MB at $T=96$. This is in contrast to DistAIM which varies between $7$MB of communication (on Covtype) up to 140MB (on Adult) with the dominating factor for DistAIM being the total dimension of the workload e.g., datasets that have many high cardinality marginals will have large communication overheads under DistAIM.

\begin{figure*}[t]
	\centering
	\subfloat[Magic]{%
		\includegraphics[width=0.24\linewidth]{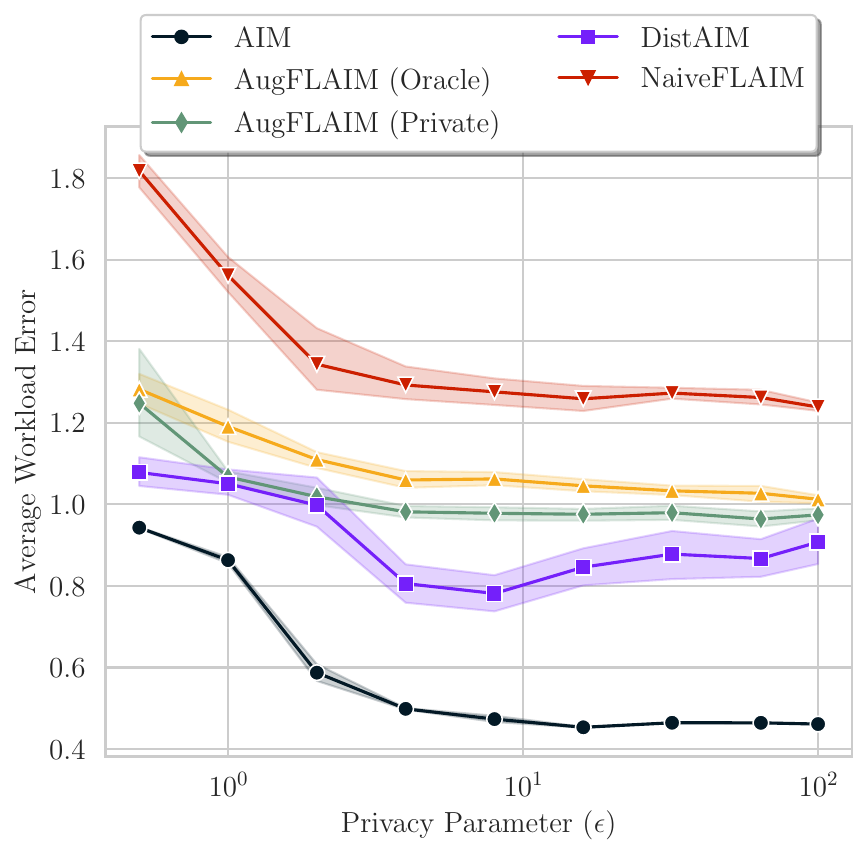}}
	\subfloat[Credit]{%
		\includegraphics[width=0.24\linewidth]{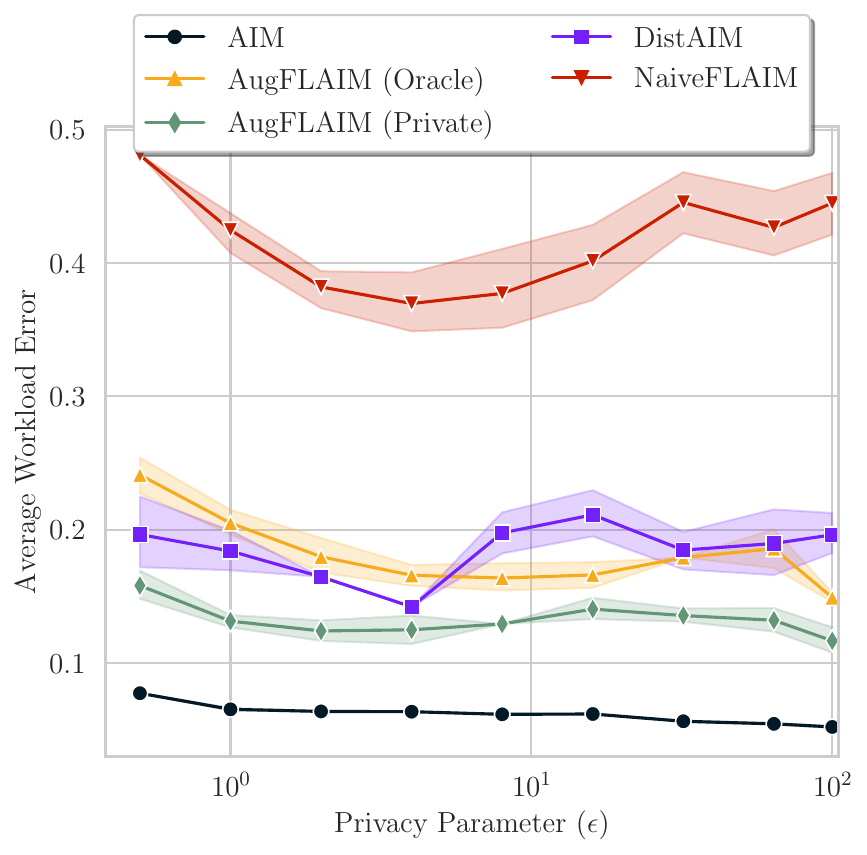}}
	\subfloat[Census]{%
		\includegraphics[width=0.24\linewidth]{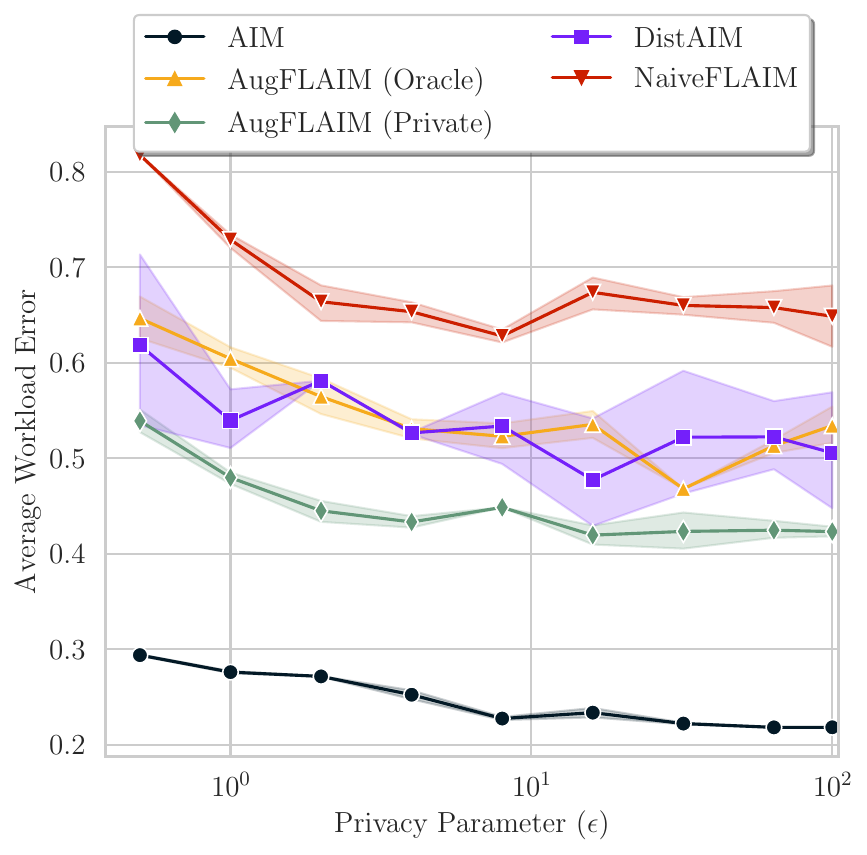}} \\
  \subfloat[Marketing]{%
		\includegraphics[width=0.24\linewidth]{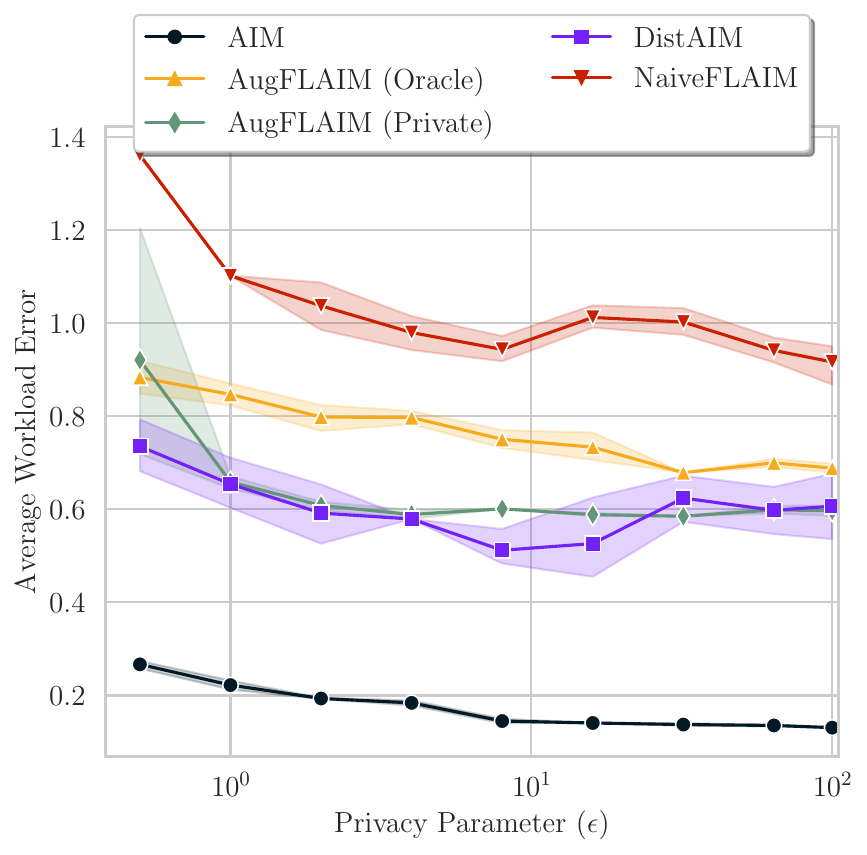}}
	\subfloat[Covtype]{%
		\includegraphics[width=0.24\linewidth]{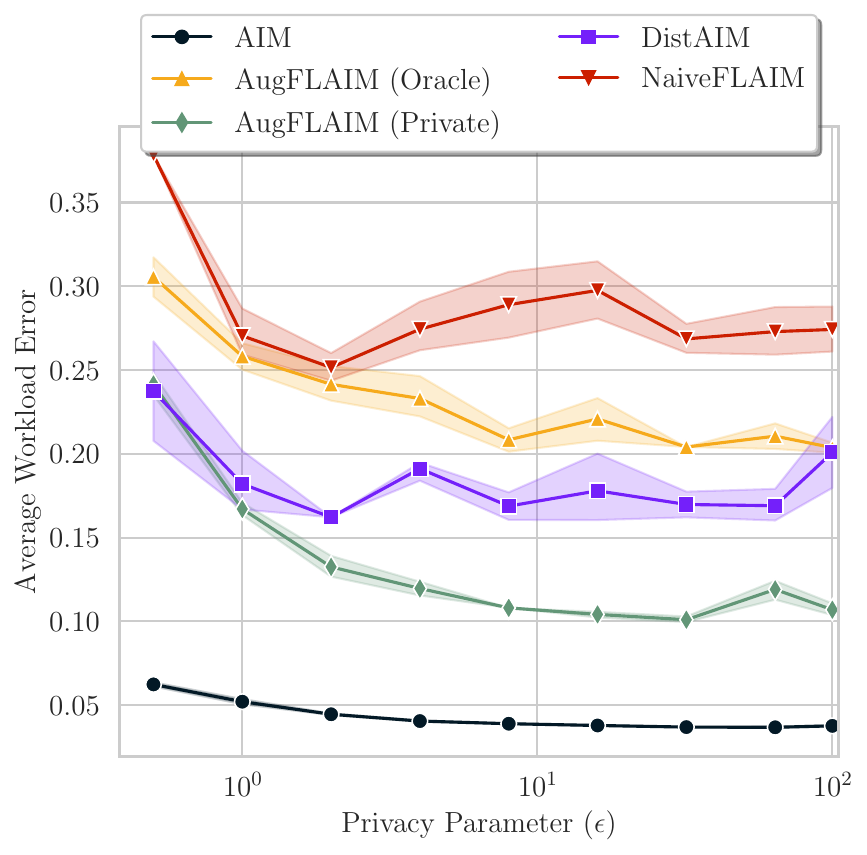}}
	\subfloat[Intrusion]{%
		\includegraphics[width=0.24\linewidth]{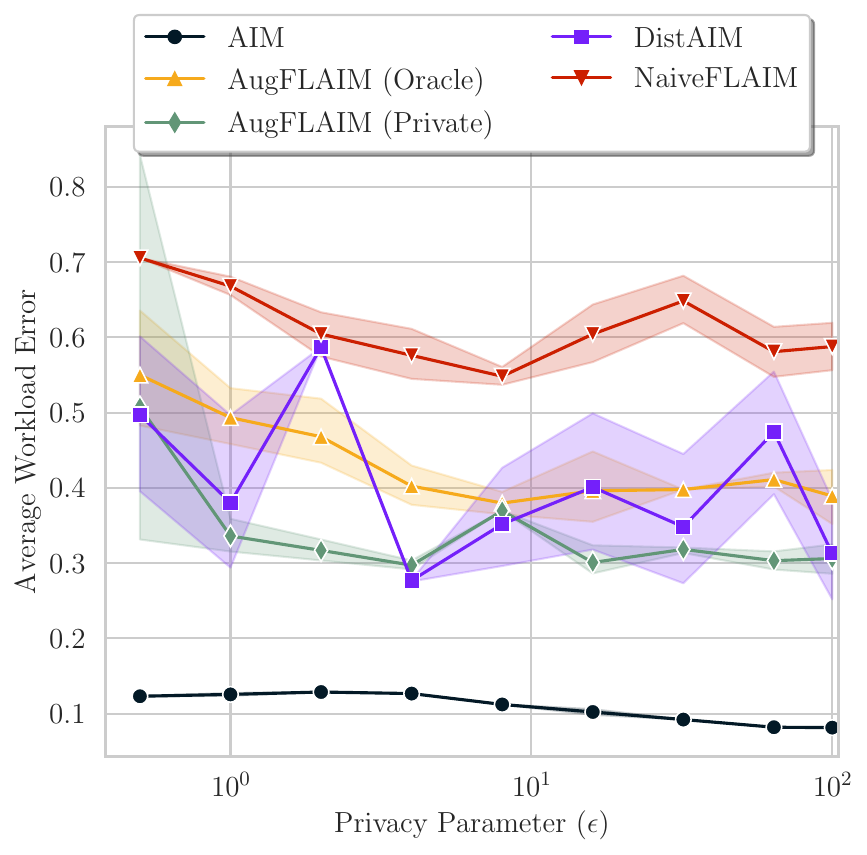}} 
	\caption{Varying $\eps$\label{fig:appendix:vary_eps}}
\end{figure*}

\begin{figure*}[t]
	\centering
	\subfloat[Magic]{%
		\includegraphics[width=0.24\linewidth]{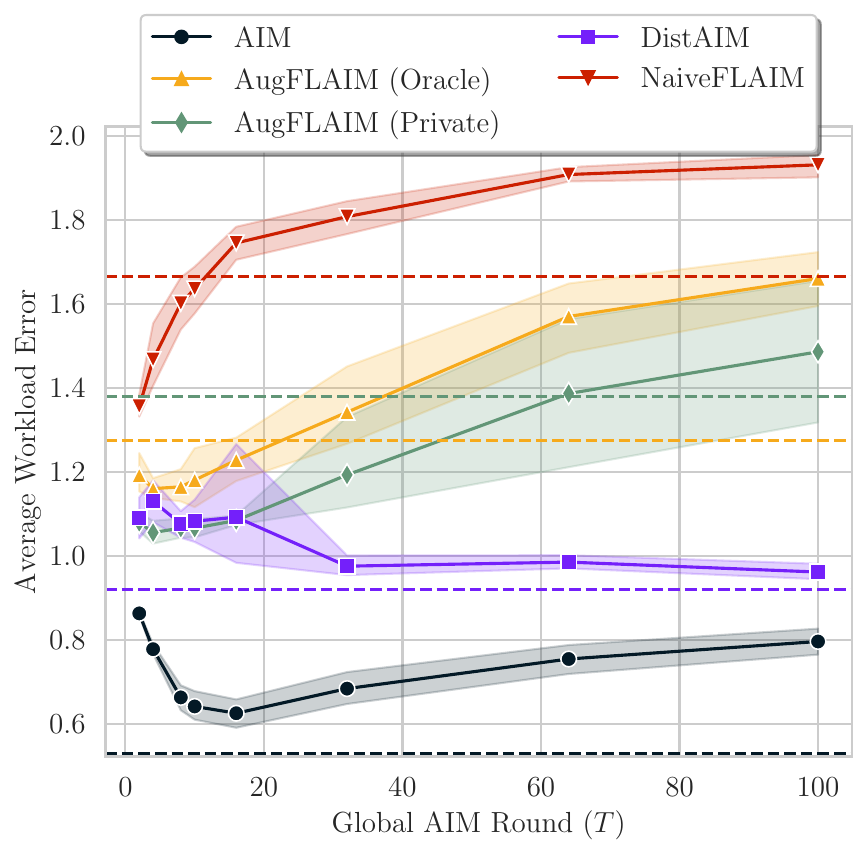}}
	\subfloat[Credit]{%
		\includegraphics[width=0.24\linewidth]{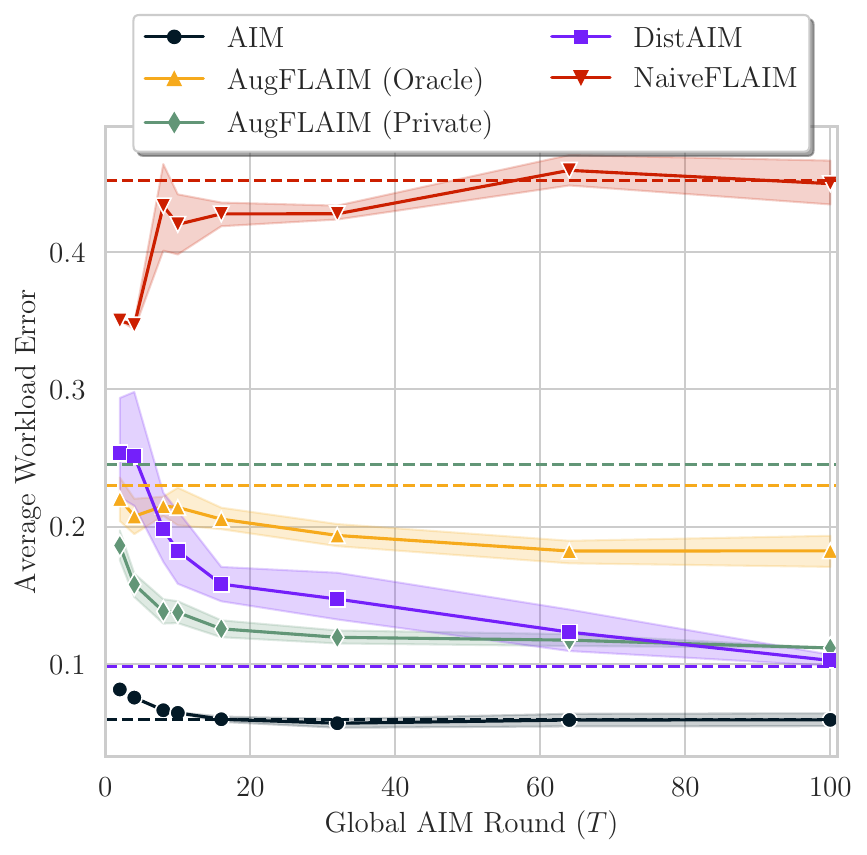}}
	\subfloat[Census]{%
		\includegraphics[width=0.24\linewidth]{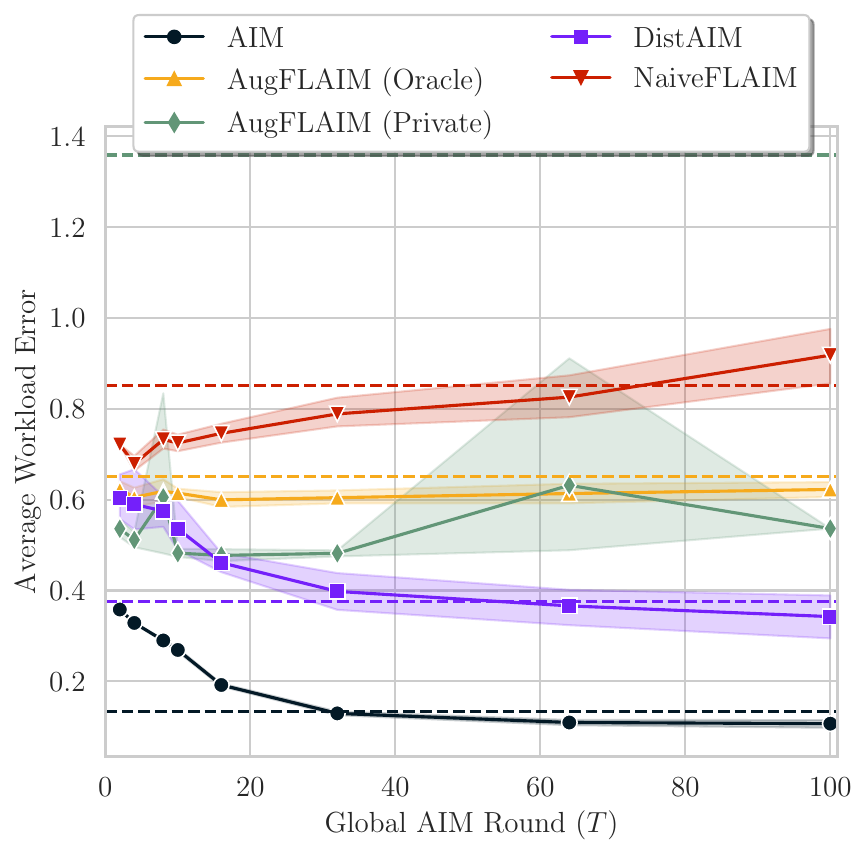}} \\
  \subfloat[Marketing]{%
		\includegraphics[width=0.24\linewidth]{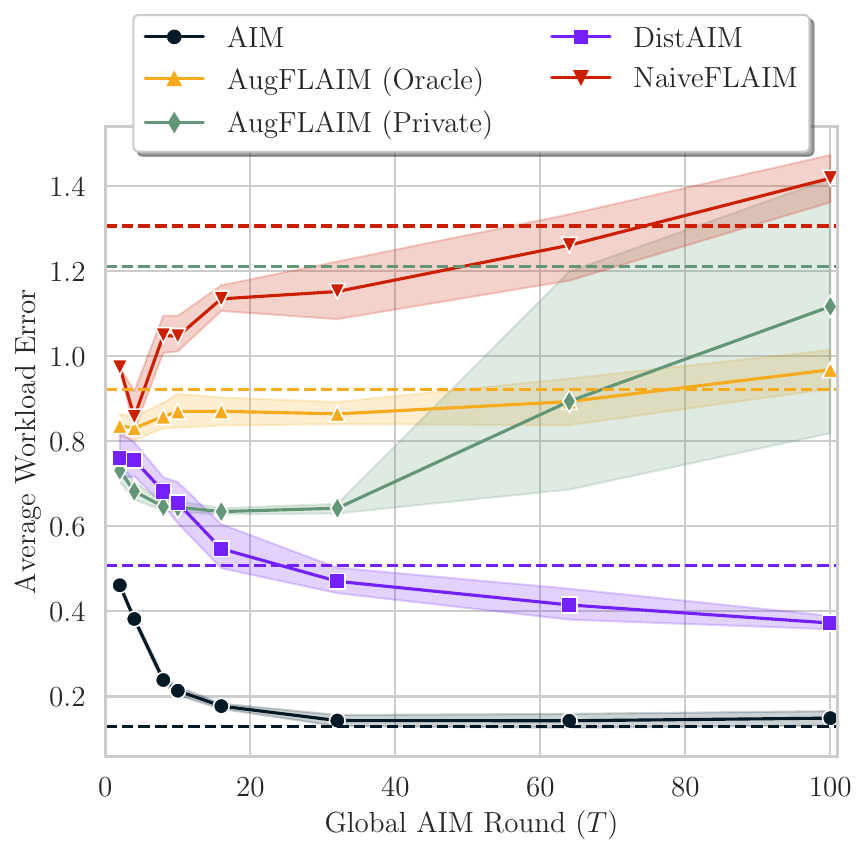}}
	\subfloat[Covtype]{%
		\includegraphics[width=0.24\linewidth]{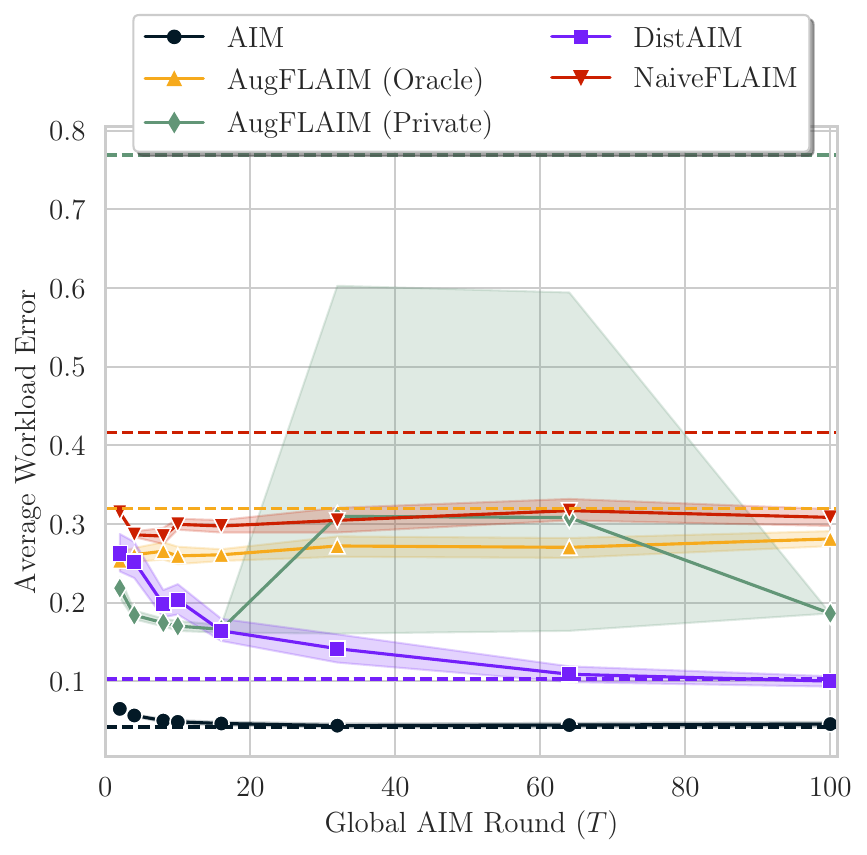}}
	\subfloat[Intrusion]{%
		\includegraphics[width=0.24\linewidth]{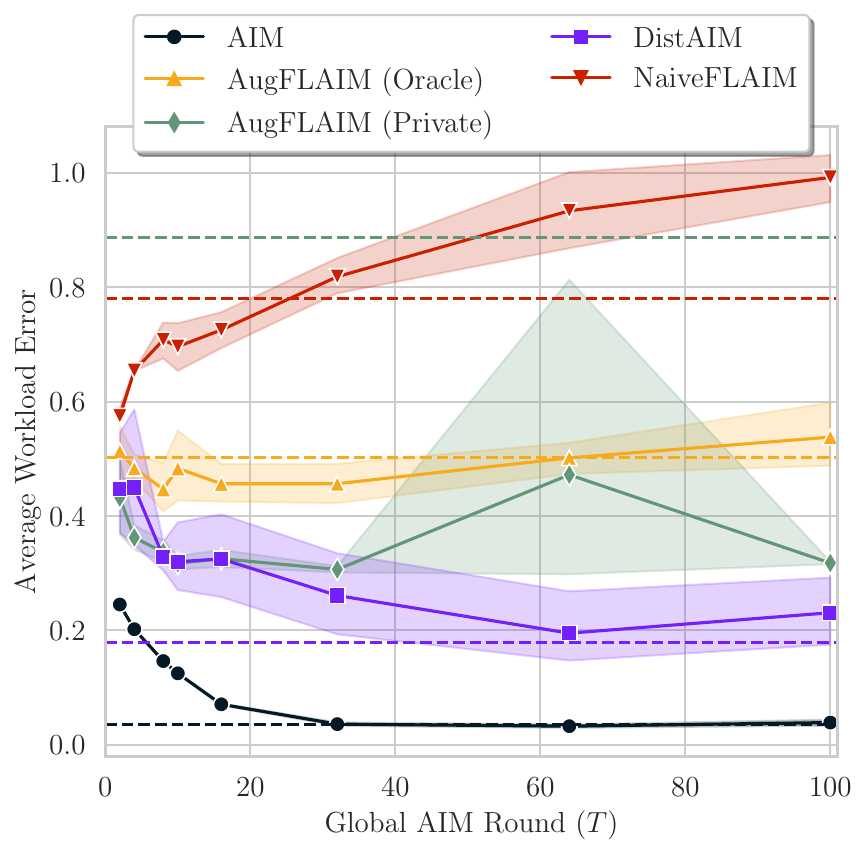}} 
	\caption{Varying $T$\label{fig:appendix:vary_t}}
\end{figure*}

\begin{figure*}[t]
	\centering
	\subfloat[Magic]{%
		\includegraphics[width=0.24\linewidth]{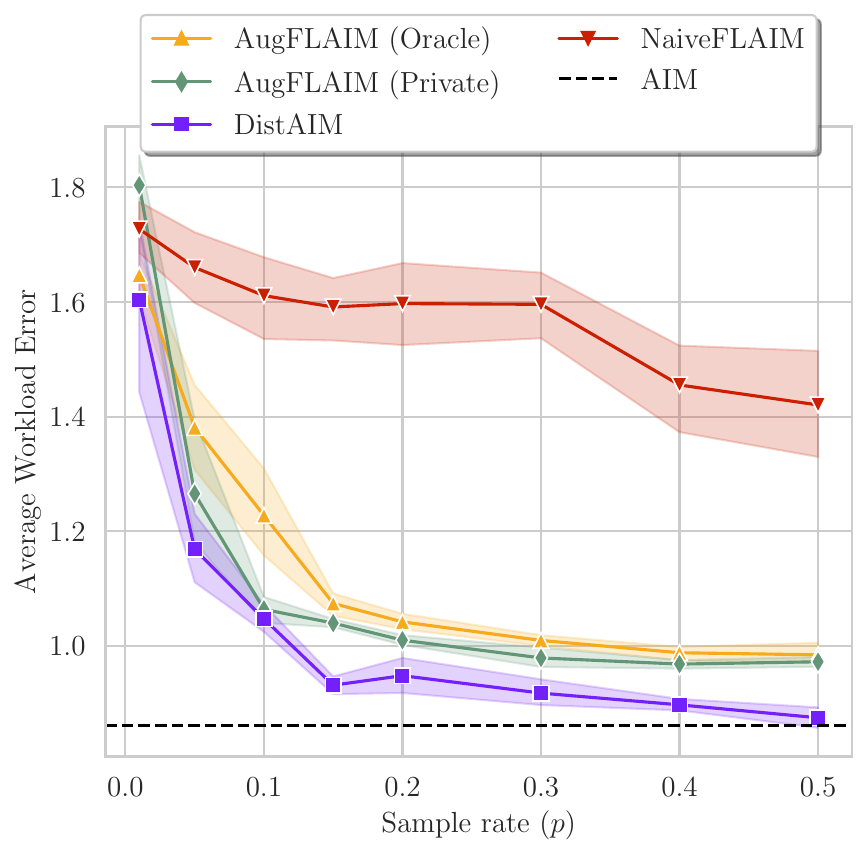}}
	\subfloat[Credit]{%
		\includegraphics[width=0.24\linewidth]{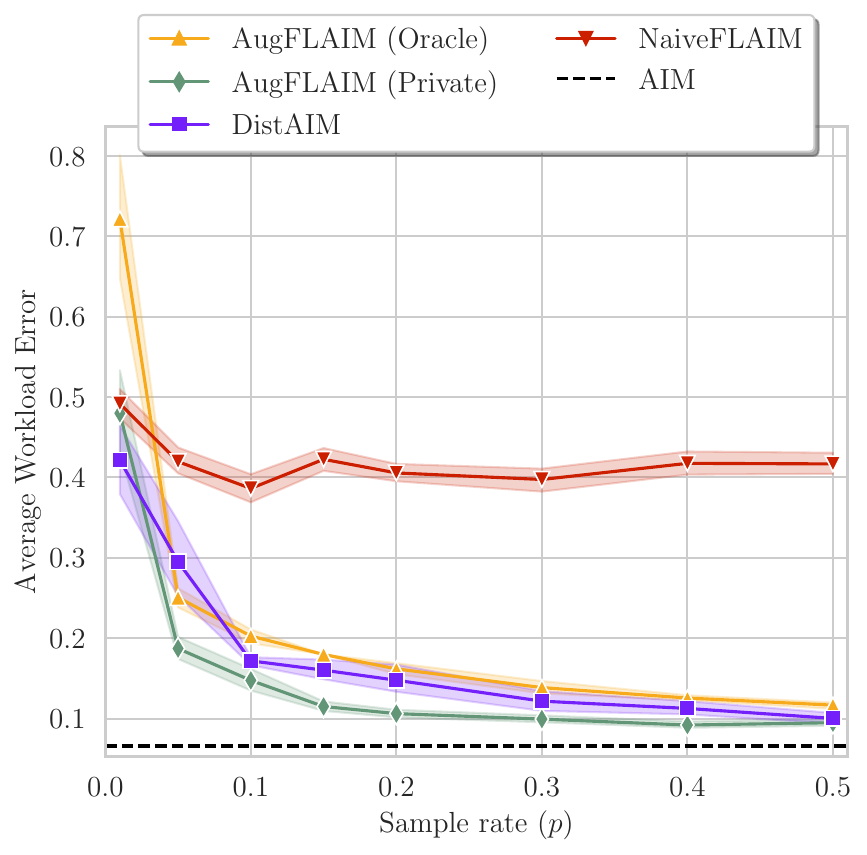}}
	\subfloat[Census]{%
		\includegraphics[width=0.24\linewidth]{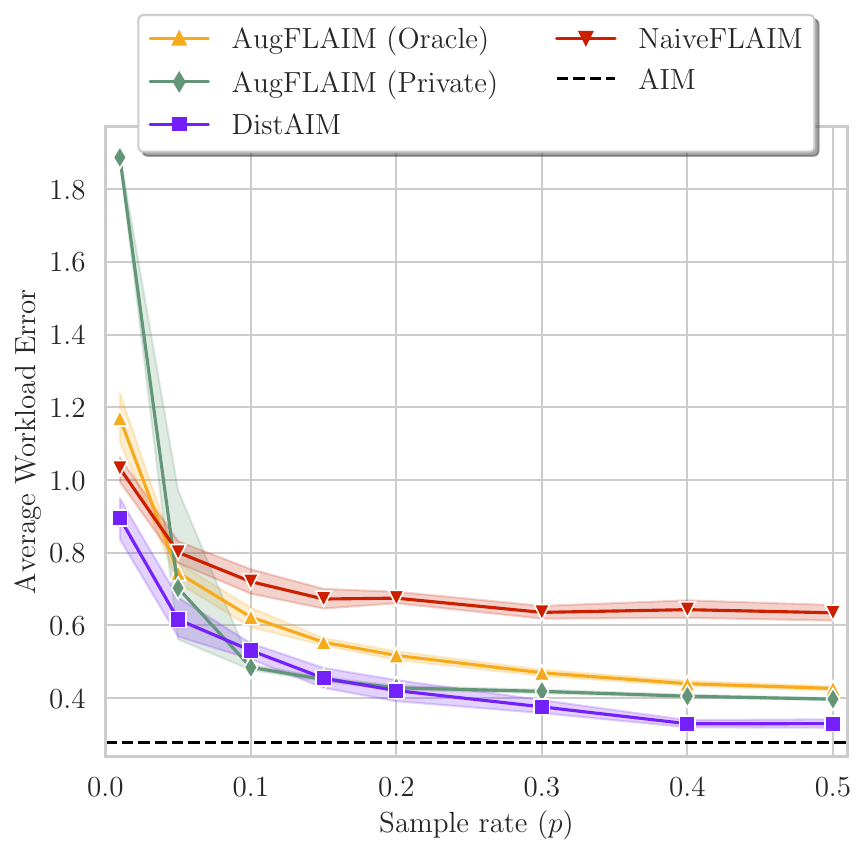}} \\
  \subfloat[Marketing]{%
		\includegraphics[width=0.24\linewidth]{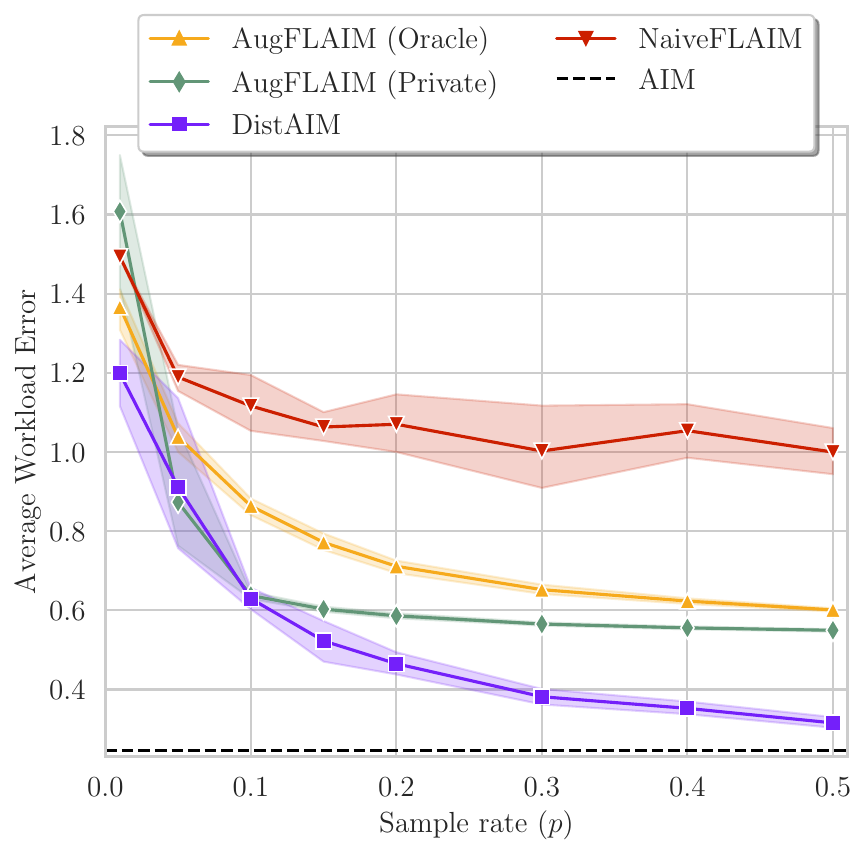}}
	\subfloat[Covtype]{%
		\includegraphics[width=0.24\linewidth]{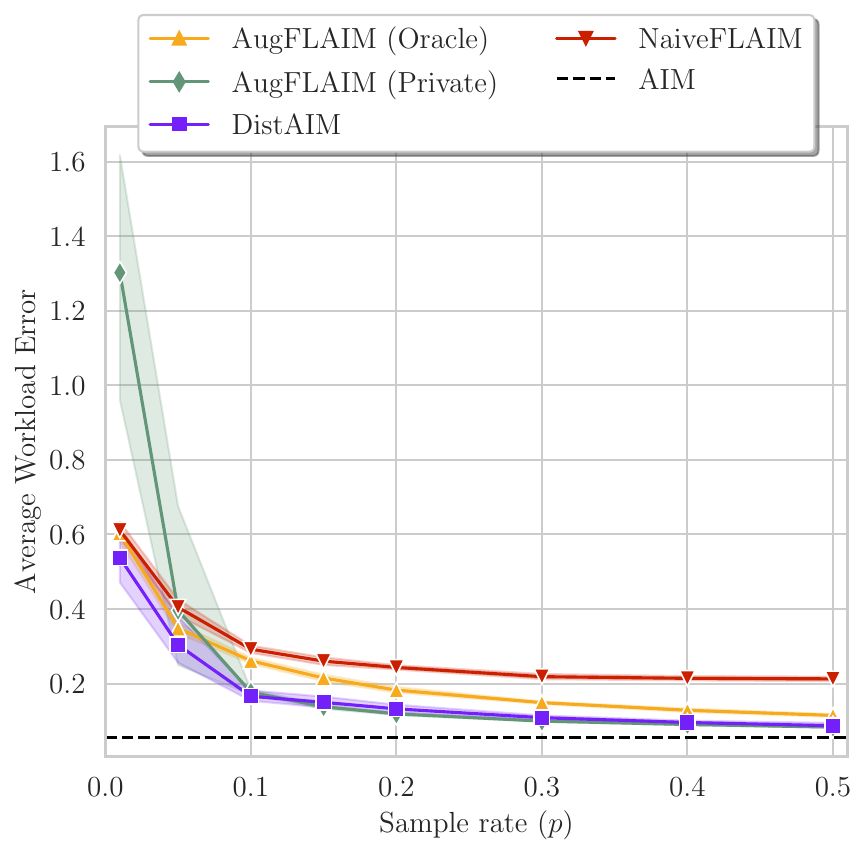}}
	\subfloat[Intrusion]{%
		\includegraphics[width=0.24\linewidth]{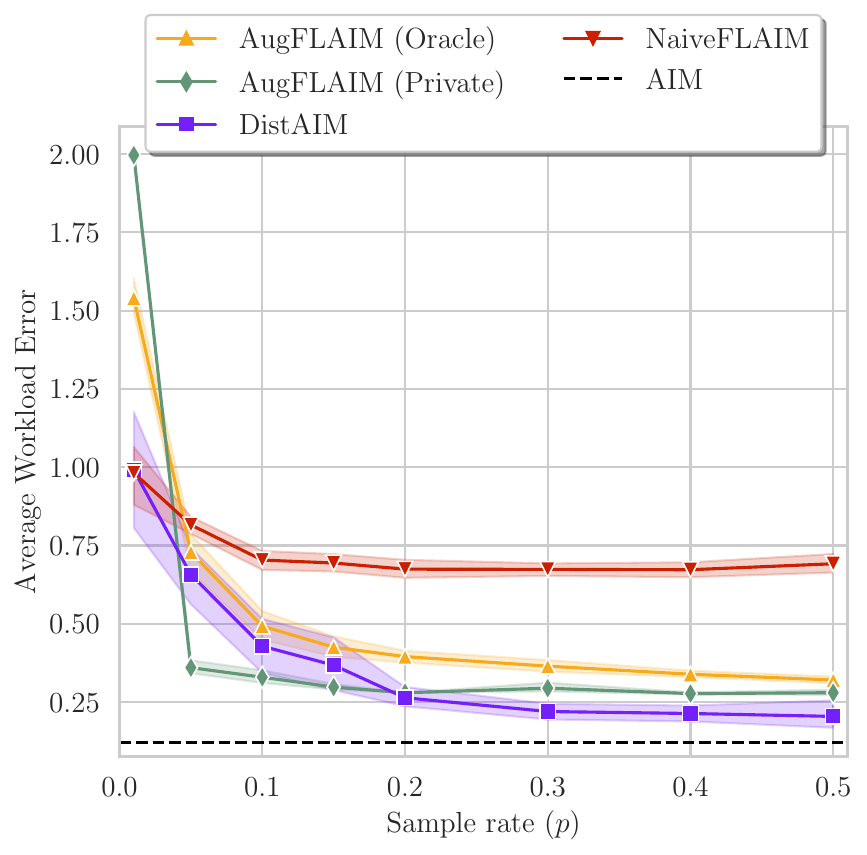}} 
	\caption{Varying $p$\label{fig:appendix:vary_p}}
\end{figure*}

\begin{figure*}[t]
	\centering
	\subfloat[Magic]{%
		\includegraphics[width=0.24\linewidth]{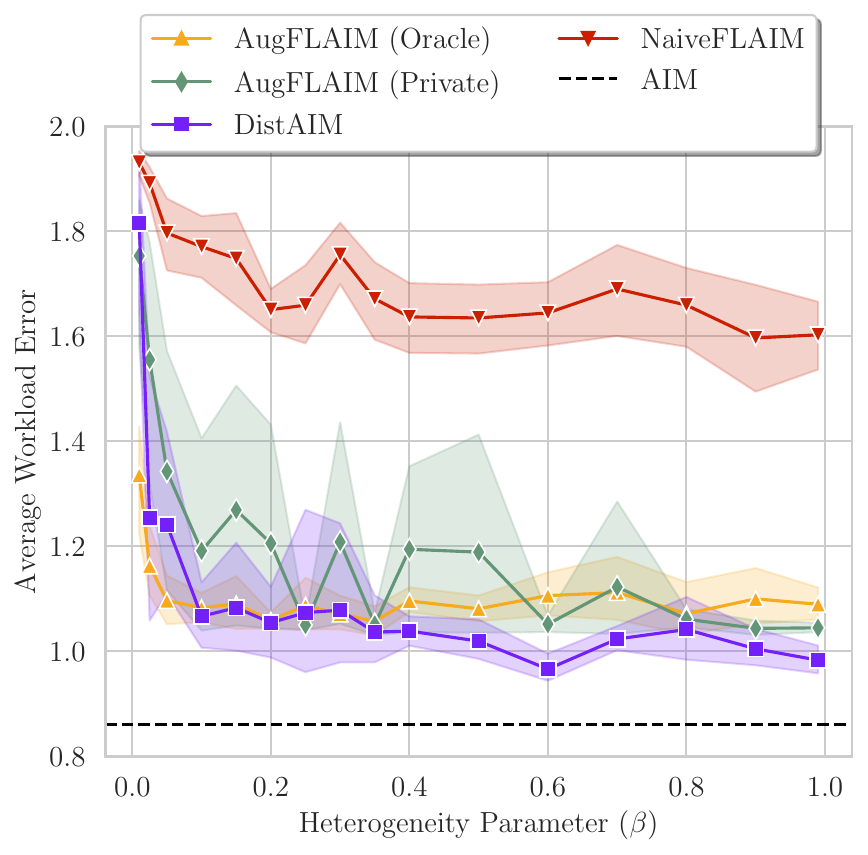}}
	\subfloat[Credit]{%
		\includegraphics[width=0.24\linewidth]{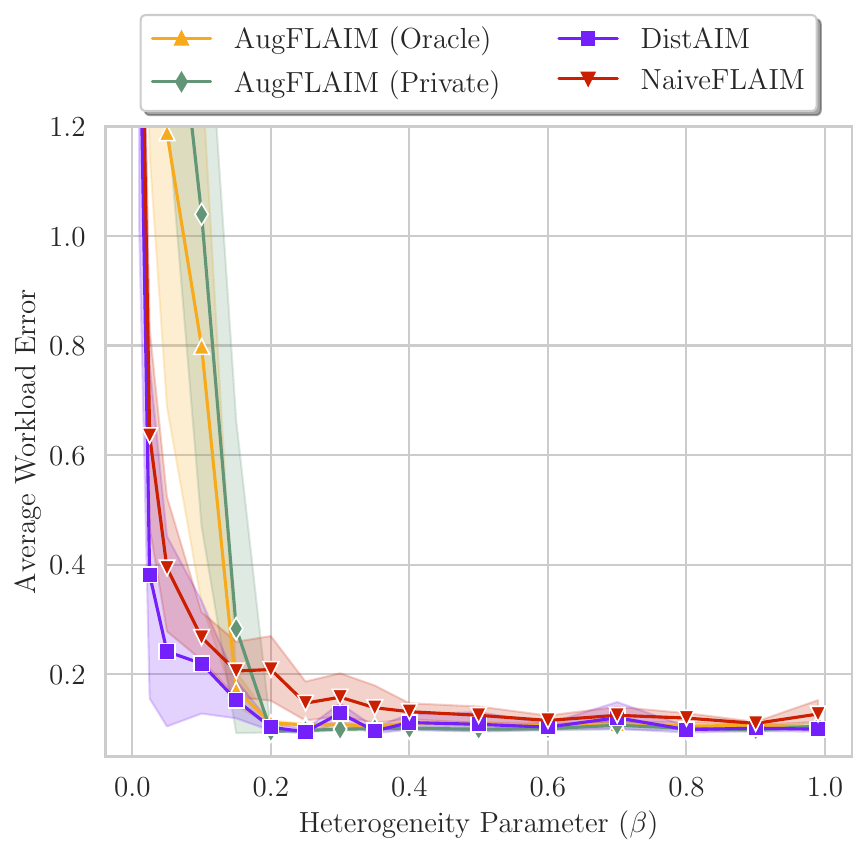}}
	\subfloat[Census]{%
		\includegraphics[width=0.24\linewidth]{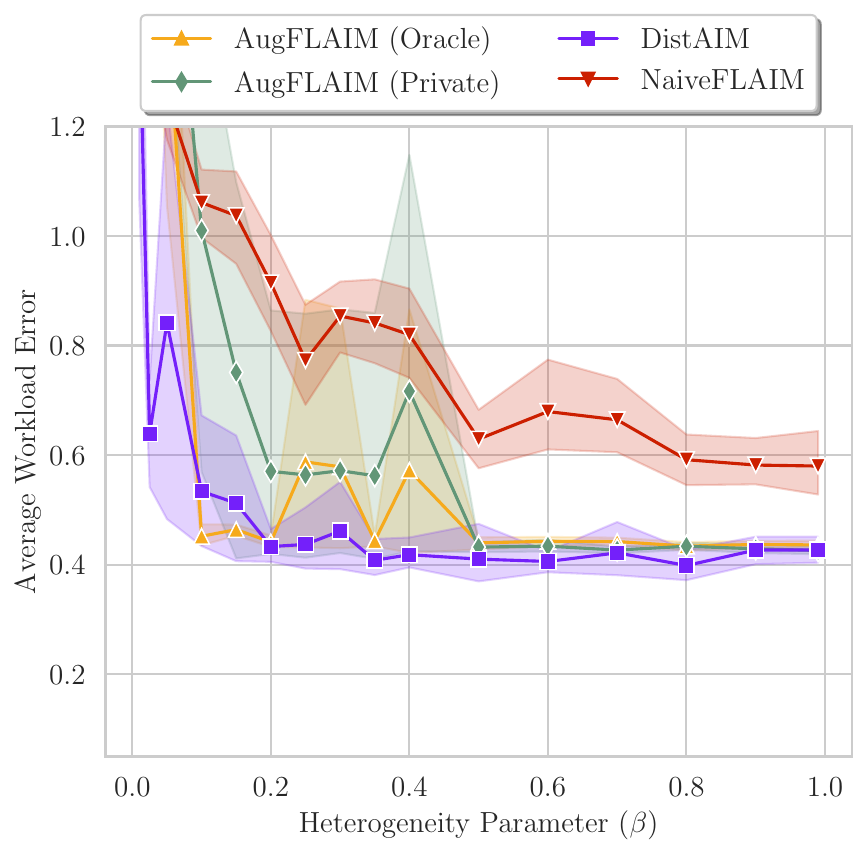}} \\
  \subfloat[Marketing]{%
		\includegraphics[width=0.24\linewidth]{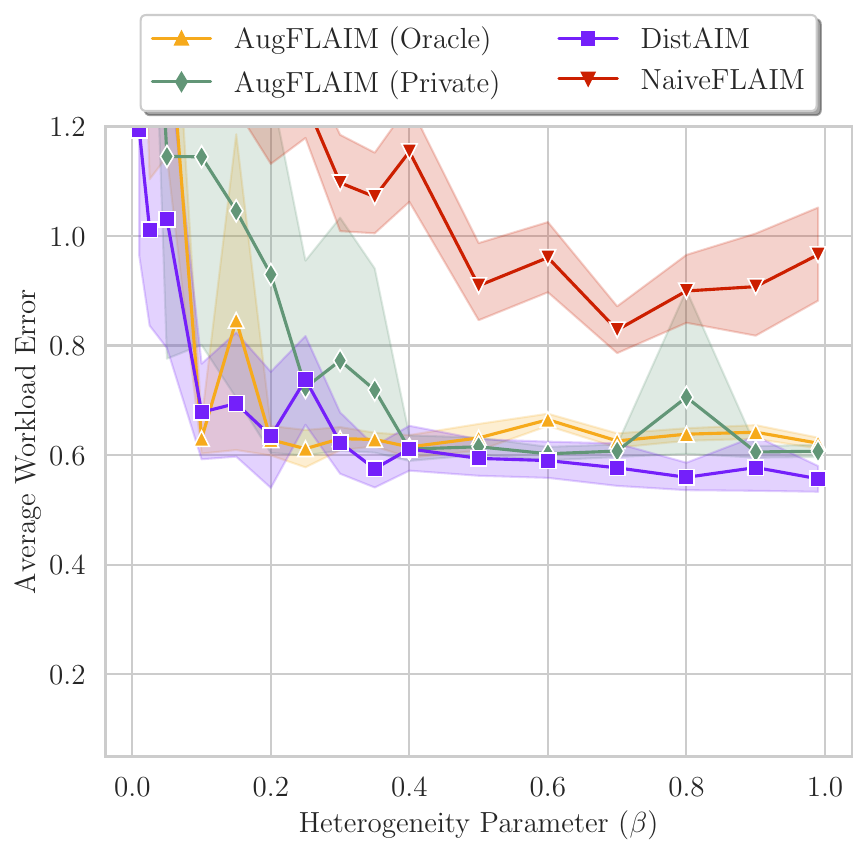}}
	\subfloat[Covtype]{%
		\includegraphics[width=0.24\linewidth]{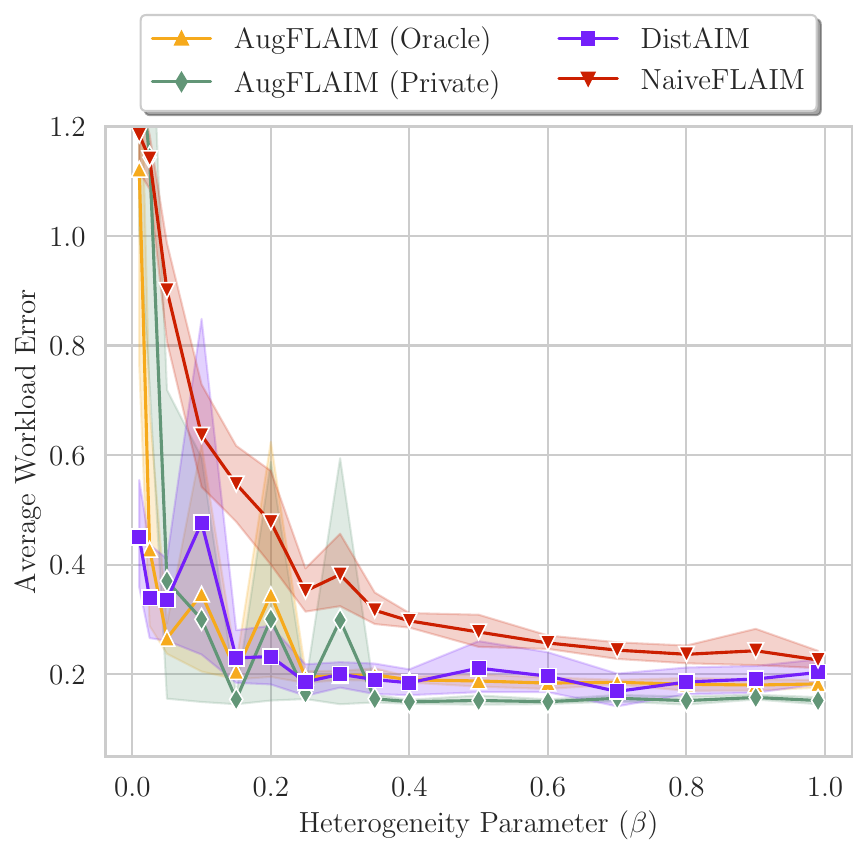}}
	\subfloat[Intrusion]{%
		\includegraphics[width=0.24\linewidth]{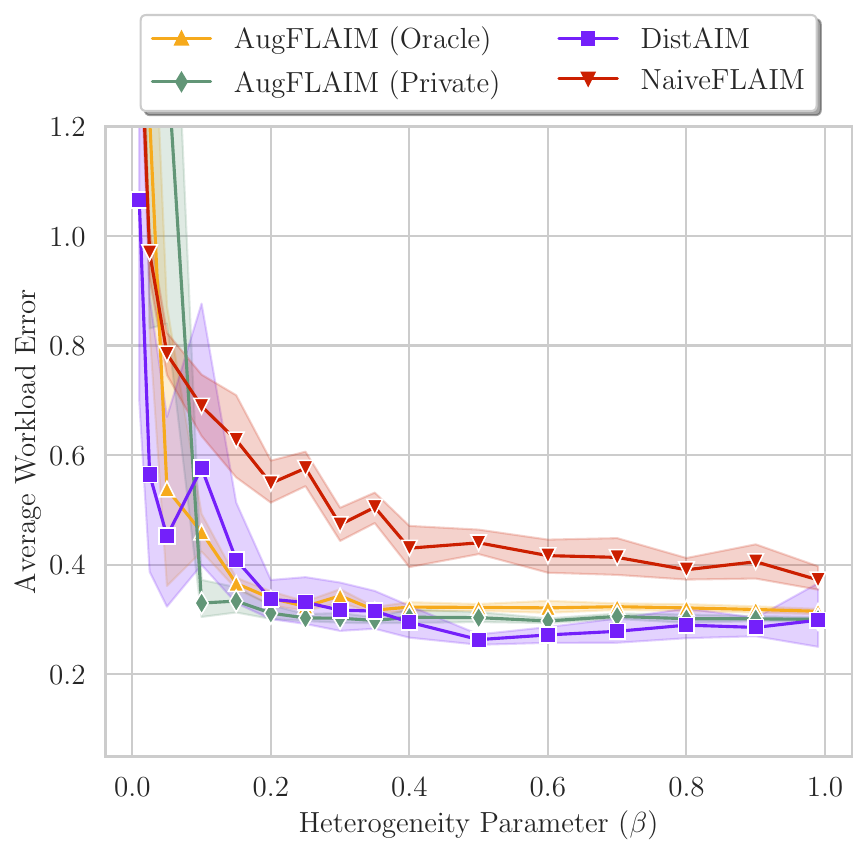}} 
	\caption{Varying $\beta$\label{fig:appendix:vary_beta}}
\end{figure*}

\begin{figure*}[t]
	\centering
	
	\subfloat[Magic (Workload Error)]{%
	\includegraphics[width=0.35\linewidth]{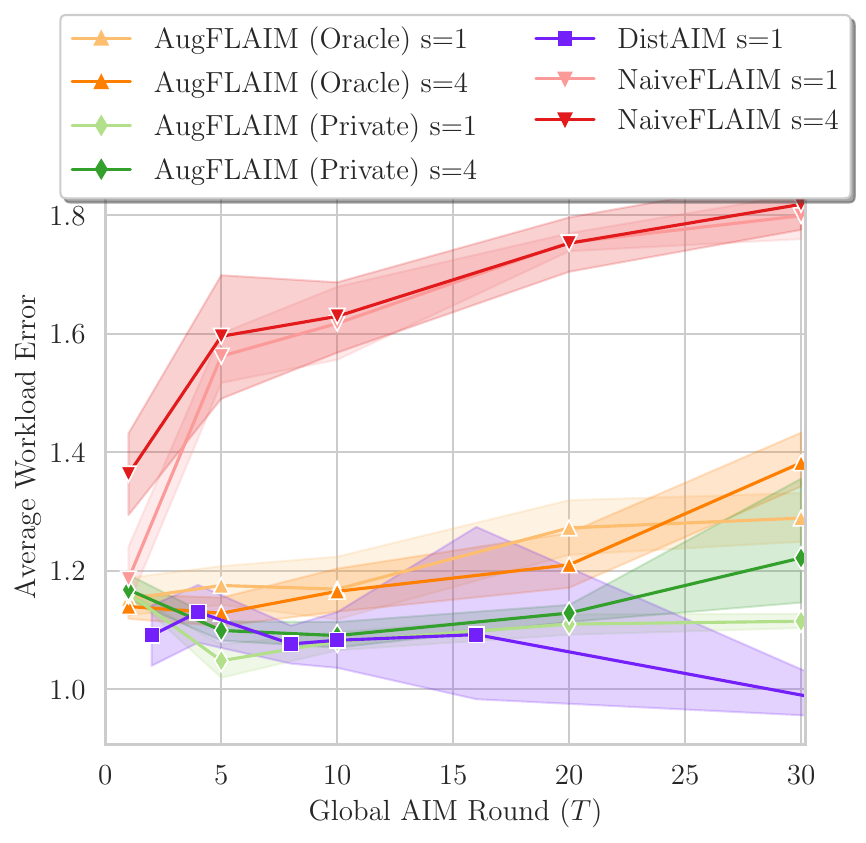}}
	\subfloat[\review{Magic (Test AUC)}]{%
		\includegraphics[width=0.35\linewidth]{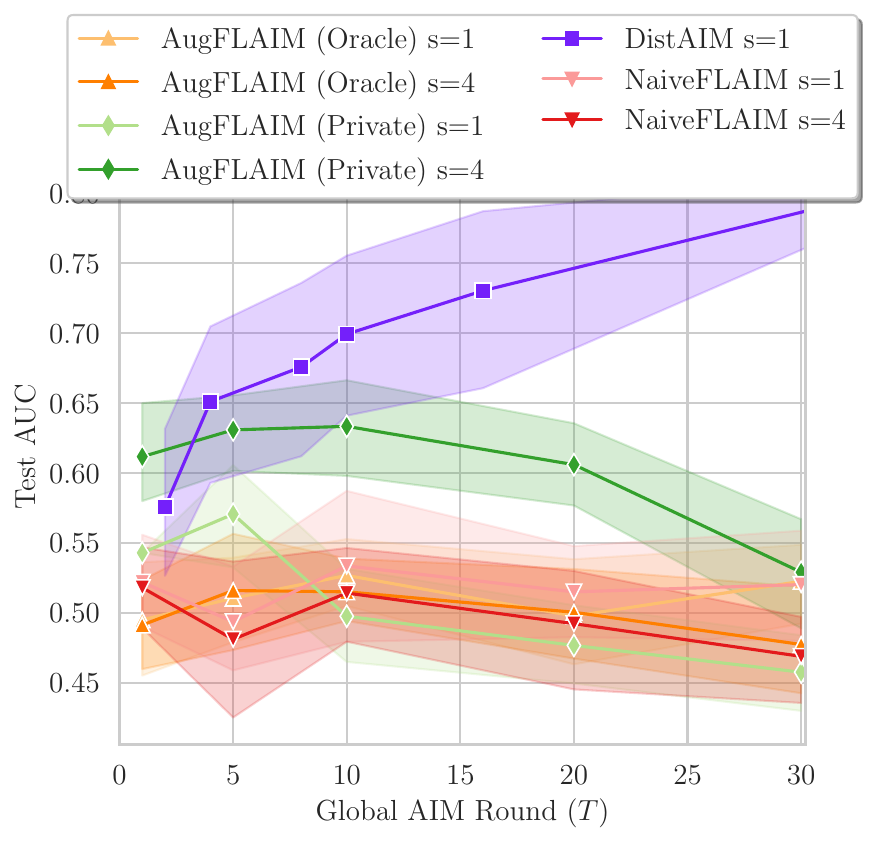}} \\

	\subfloat[Credit (Workload Error)]{%
	\includegraphics[width=0.35\linewidth]{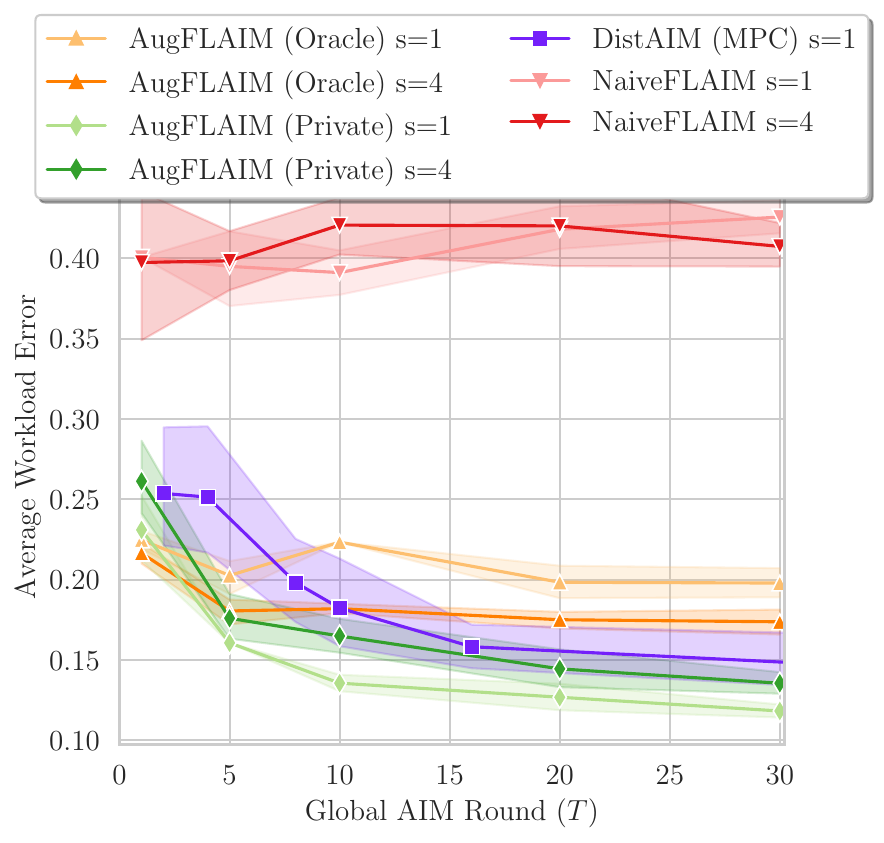}}
	\subfloat[\review{Credit (Test AUC)}]{%
		\includegraphics[width=0.35\linewidth]{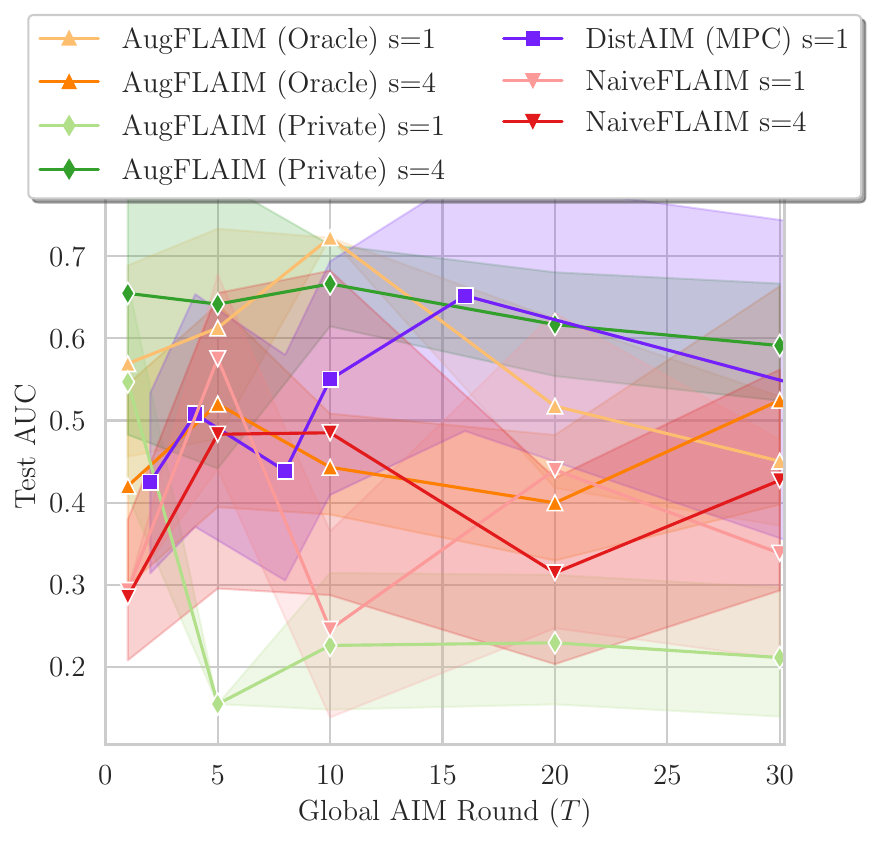}} \\

	\subfloat[Census (Workload Error)]{%
	\includegraphics[width=0.35\linewidth]{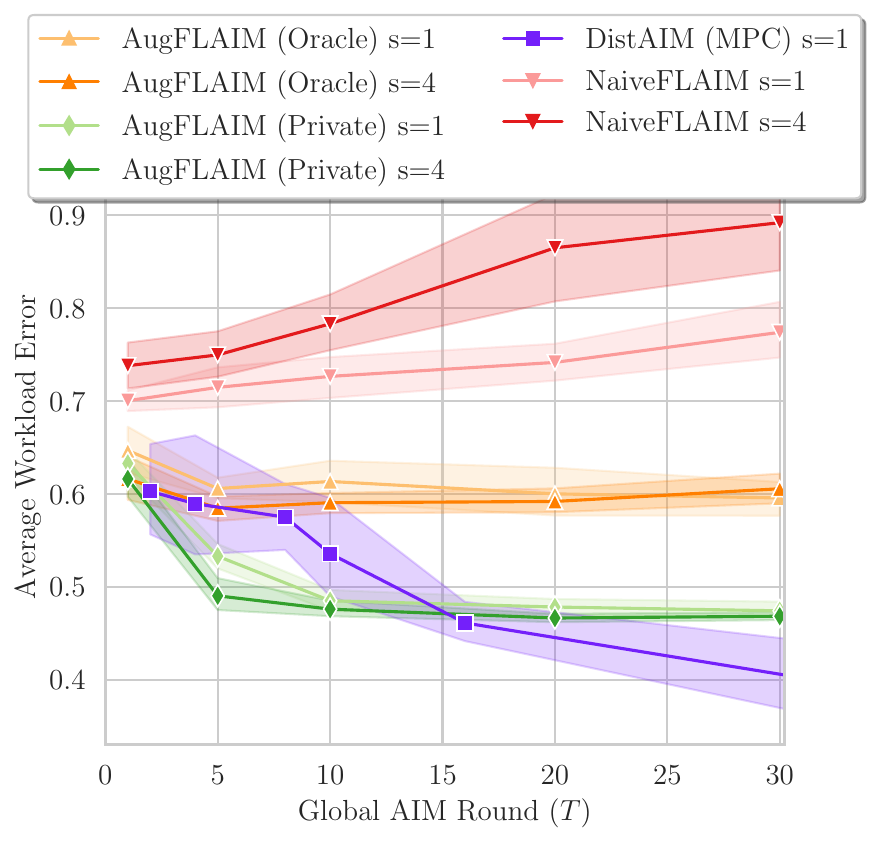}}
	\subfloat[\review{Census (Test AUC)}]{%
		\includegraphics[width=0.35\linewidth]{appendix_figs/fig5_auc_credit_sdv_eps=1_weight=sigma_scaled_squared_nweight.pdf}} 
        \caption{\review{Varying local rounds $s \in \{1,4\}$ as in Figure \ref{fig:vary_s_eps1} but on alternative datasets.}}	\label{fig:appendix_vary_s1} 
\end{figure*}

\begin{figure*}[t]
	\centering
	
	\subfloat[Marketing (Workload Error)]{%
	\includegraphics[width=0.35\linewidth]{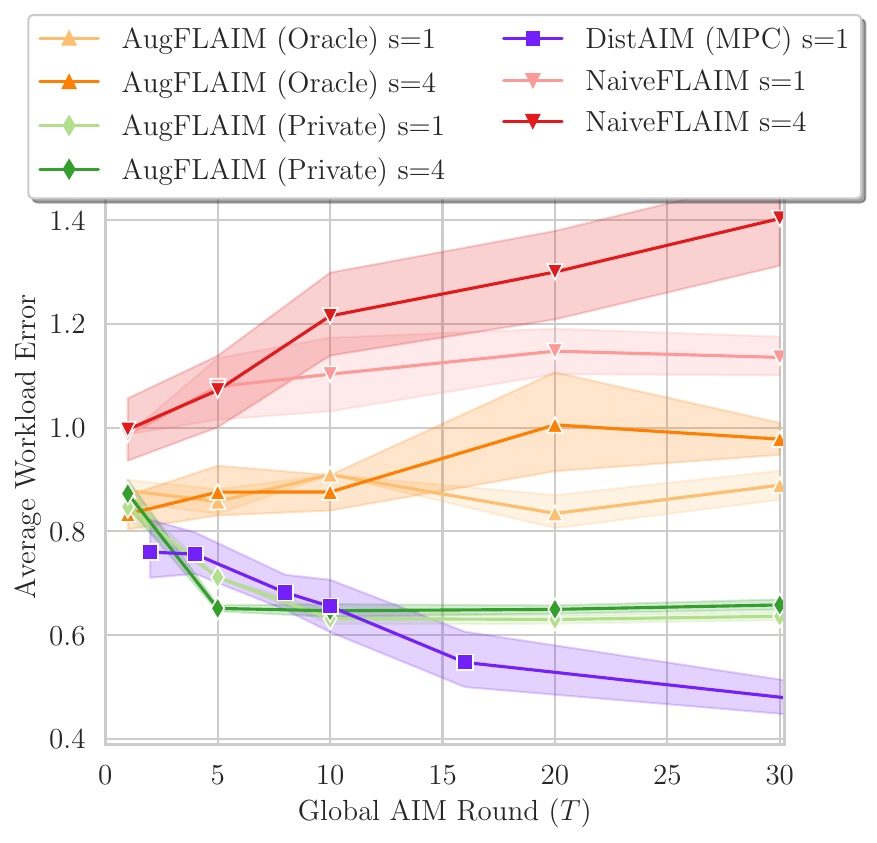}}
	\subfloat[\review{Marketing (Test AUC)}]{%
		\includegraphics[width=0.35\linewidth]{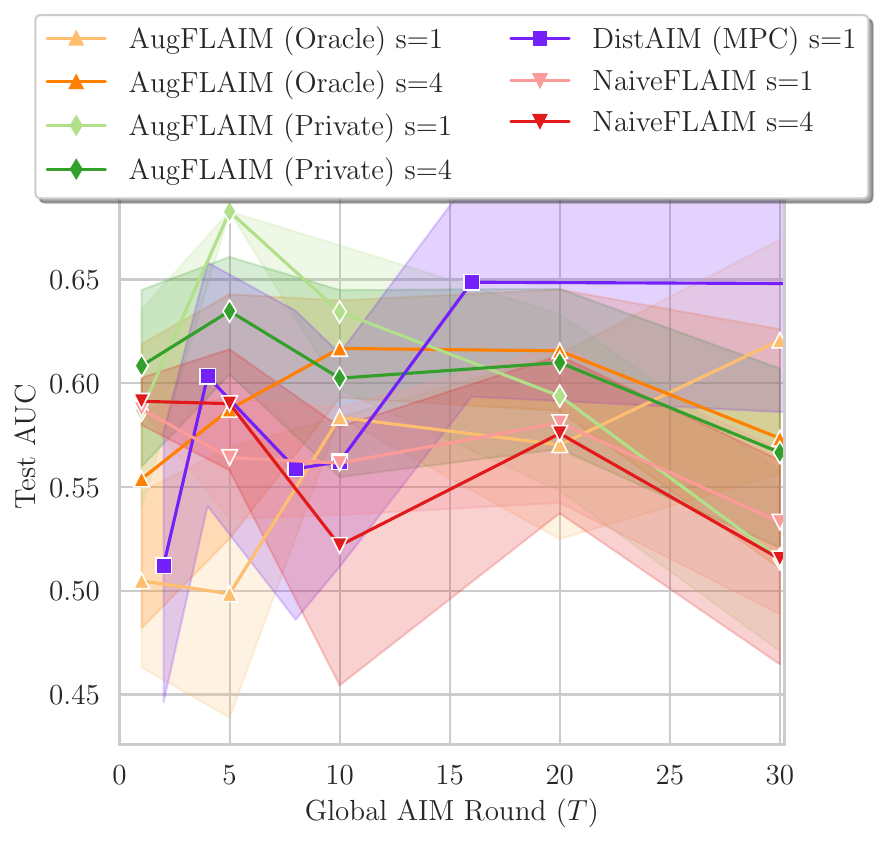}} \\

	\subfloat[Covtype (Workload Error)]{%
	\includegraphics[width=0.35\linewidth]{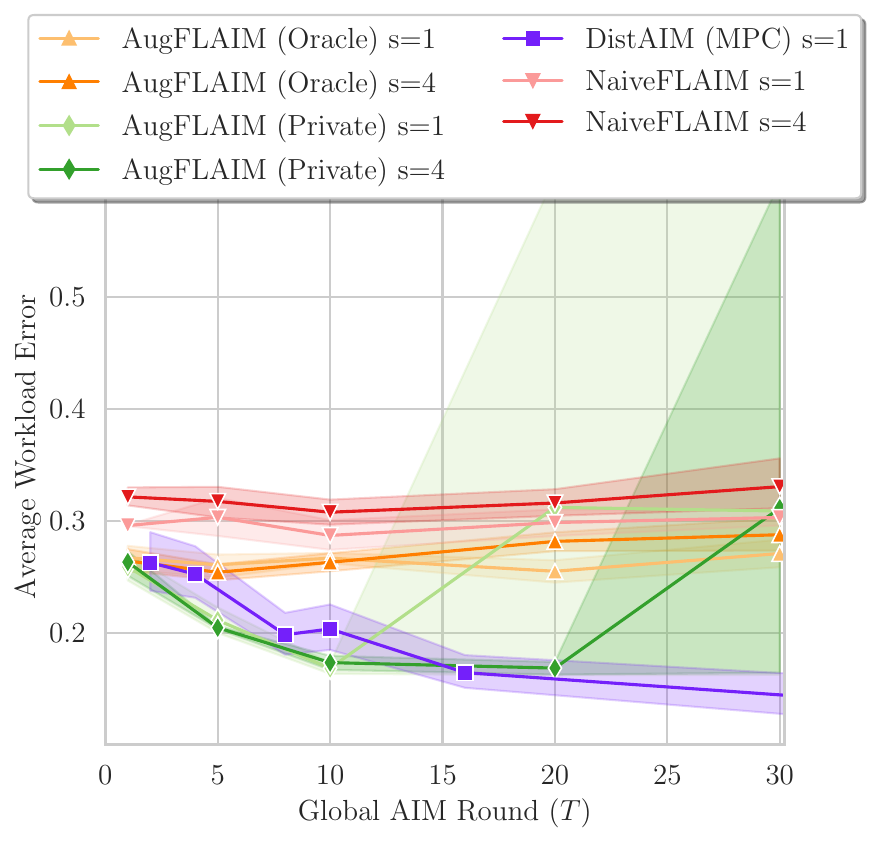}}
	\subfloat[\review{Covtype (Test AUC)}]{%
		\includegraphics[width=0.35\linewidth]{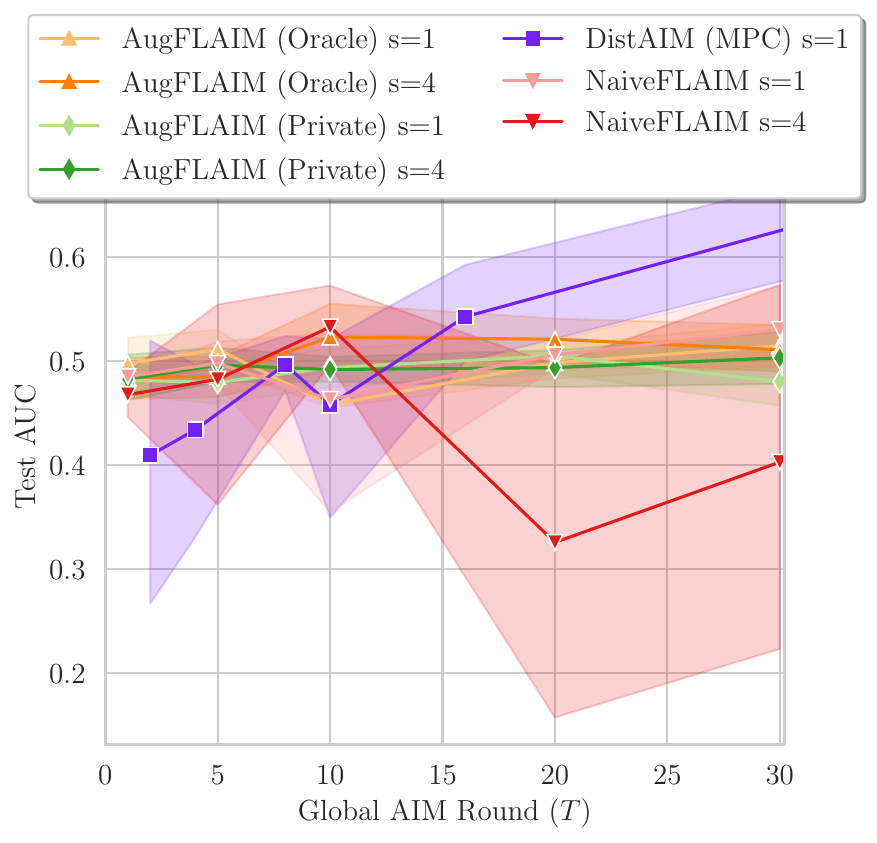}} \\

	\subfloat[Intrusion (Workload Error)]{%
	\includegraphics[width=0.35\linewidth]{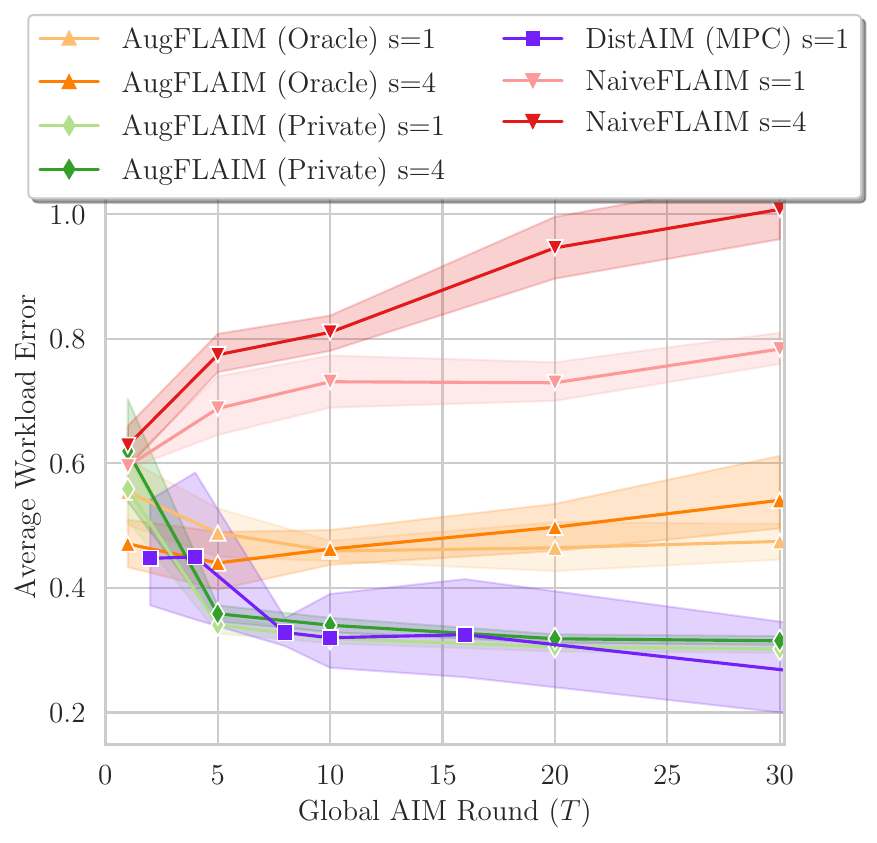}}
	\subfloat[\review{Intrusion (Test AUC)}]{%
		\includegraphics[width=0.35\linewidth]{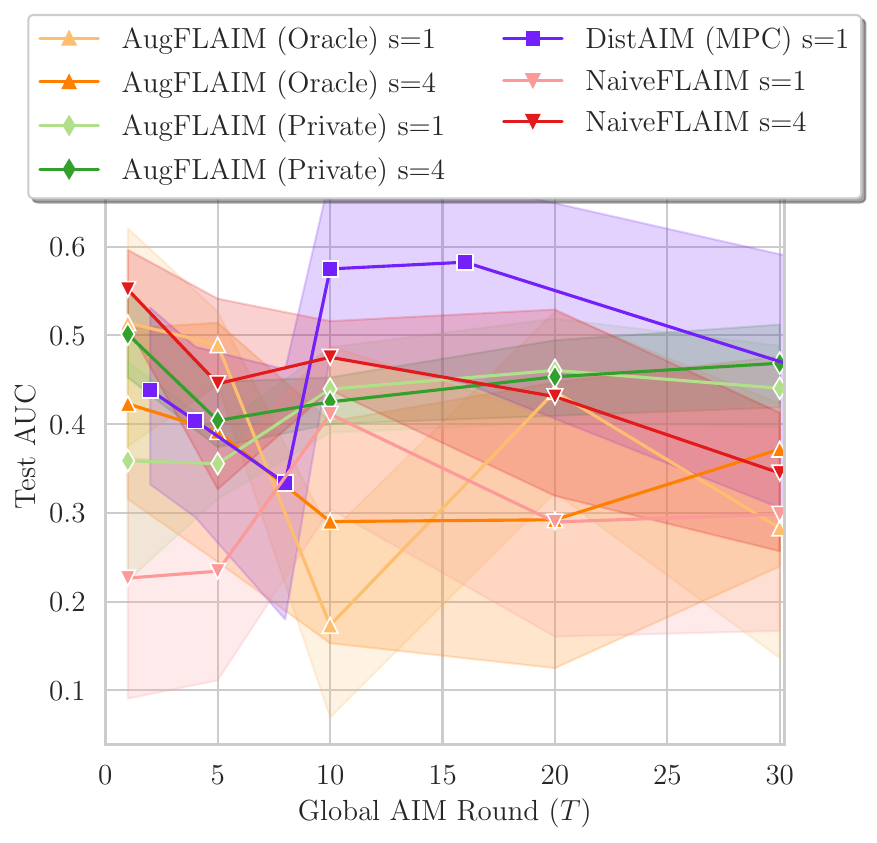}} 
        \caption{\review{Varying local rounds $s \in \{1,4\}$ as in Figure \ref{fig:vary_s_eps1} but on alternative datasets.} \label{fig:appendix_vary_s2}} 
\end{figure*}

\end{document}